%% file: spygame-bip.tex
\documentclass[runningheads]{llncs}
\usepackage[left=2.5cm,right=2.5cm,bottom=3cm]{geometry}

\usepackage{graphicx}
\usepackage{amssymb, amsmath}
\usepackage{tikz}
\usetikzlibrary{calc}

\usepackage[final]{pdfpages}

\newcommand{\ZZ}{\mathbb{Z}}
\DeclareMathOperator{\ggn}{gn'}
\DeclareMathOperator{\gn}{gn}
\DeclareMathOperator{\cn}{cn}

\newcommand{\qedclaim}{\hfill $\diamond$ \medskip}

\title{Algorithms, hardness and graph products on a pursuit-evasion game}
\author{
   Eurinardo Costa \inst{1}\and
   Nicolas Martins \inst{2}\and
   Rudini Sampaio \inst{3}
}

\institute{
   Universidade Federal do Cear\'a, Campus Russas, Russas, Brazil.\\\email{eurinardo@ufc.br}
   \and
   Universidade Integr. Internacional Lusofonia Afrobrasileira, Unilab, Redenção, Brazil.\\\email{nicolasam@unilab.edu.br}
   \and
   Departamento de Computa\c c\~ao, Universidade Federal do Cear\'a, Fortaleza, Brazil.\\\email{rudini@dc.ufc.br}
}

\begin{document}
\mainmatter
\maketitle

\begin{abstract}
In the $(s,d)$-spy game over a graph, introduced by Cohen et al. in 2016, one spy and $k$ guards occupy vertices of a graph and, at each turn, each guard may move along one edge and the spy may move along at most $s$ edges. The guards win if, after a finite number of turns, they ensure that the spy always remains at distance at most $d$ from at least one guard. 
The guard number is the minimum number of guards such that the guards have a winning strategy. In this paper, we investigate the spy game variant in which the guards are placed first, before the spy. We obtain a polynomial time algorithm for every speed $s\geq 2$ and distance $d\geq 0$ when the number of guards is a constant, which leads to a fixed parameter tractable algorithm on the $P_4$-fewness of the graph. We also prove that the spy game is NP-hard even in bipartite graphs with bounded diameter, for every speed $s\geq 2$ and distance $d\geq 0$.
\end{abstract}

\keywords{Spy game on graphs, XP algorithm, NP-hardness.}

\section{Introduction}
\label{sec:intro}

Given a graph $G=(V,E)$ and $v \in V$, let $N(v)=\{w \mid vw \in E\}$ denote the set of neighbors of $v$ and let $N[v]=N(v) \cup \{v\}$ be the closed neighborhood of $v$. Let $P_n$ be the graph consisting of a path with $n$ vertices.

In the two-player $(s,d)$-spy game on a finite graph $G$, where $s\geq 1$, $d\geq 0$ and $k\geq 1$ are the spy speed, the surveillance distance and the number of guards, respectively, a player controls one spy and the other player controls $k$ guards, all of them occupying vertices of $G$. The spy and the guards may occupy the same vertex. It is a perfect information sequential game with no chance moves.
The game proceeds turn-by-turn: first the spy move with speed $s$ (along at most $s$ edges) and then each guard may move along at most one edge. 
The guards win if, after a finite number of turns, they ensure that the spy always remains at distance at most $d$ from at least one guard.

The spy game is closely related to the well known {\it Cops and robber} game \cite{BonatoN11,NowakowskiW83}. In this game, first $k$ cops occupy some vertices of the graph and then one robber occupies a vertex. Turn-by-turn, each player may move (the cops first and then the robber) along one edge. The cops win if one cop occupies the same vertex of the robber after a finite number of turns. The {\it cop-number} $\cn(G)$ of a graph $G$ is the minimum number of cops required to win in $G$~\cite{AF84}.

There are many generalizations of the Cops and robber game \cite{AlonM11,BonatoCP10,ChalopinCNV11,FominGL10,FominGP12}. For example, allowing a faster robber with speed $s\geq 2$. In this variant, the exact number of cops with speed one required to capture a robber with speed two is unknown even in 2-dimensional grids~\cite{BalisterBBN16,FominGKNS10}. In 2010, Bonato et al. \cite{BonatoCP10} introduced other variant of the cops and robber in which the game is over if a cop occupies a vertex at distance at most a given integer $d$ from the robber. This is equivalent to the spy game with speed $s=1$ in a variant in which the spy is placed after the guards. For speed $s\geq 2$, the equivalence is not true and the games are significantly different.

Other related well know game is the {\it eternal domination} game~\cite{GHH05,GoldwasserK08,KM09,Klostermeyer2011}. A set of $k$ {\it defenders} occupy some vertices of a graph $G$. At each turn, an {\it attacker} chooses a vertex $v\in V$ and the defenders may move along an edge in such a way that at least one defender is at distance at most a given integer $d$ from $v$. There are some variants of the eternal domination game allowing more defenders to move at each turn and to occupy the same vertex \cite{GoldwasserK08,KM09,Klostermeyer2011}, which is equivalent to the spy game when the spy speed is at least the diameter of the graph.

The spy game was introduced by Cohen et al. in 2016 \cite{cohen16} in two variants. In the main variant, the spy is placed first and we may consider that she wins if she reaches a vertex at distance more than the surveilling distance $d$ from every guard.
From the classical Zermelo-von Neumann theorem \cite{zermelo13}, one of the two players has a winning strategy, since it is a perfect information finite game without draw. Here we can consider the spy game as a finite game since the number of possible configurations of the spy and the guards is finite in a finite graph $G$. For example, we may consider that the guards win the game if the spy repeats a game configuration (after her move).

The \emph{guard number} is the minimum number of guards such that the guards have a winning strategy.
We denote by $\ggn_{s,d}(G)$ and $\gn_{s,d}(G)$ the guard number in the $(s,d)$-spy game variants in which the spy is placed first and the guards are placed first, respectively.
An example of the difference between these parameters is the graph $G$ with $n$ isolated vertices: $\gn_{s,d}(G)=n$ and $\ggn_{s,d}(G)=1$.
Cohen et al. \cite{cohen18} obtained some bounds in the following Proposition \ref{claim1}.
Notice that they are tight, since $\gn_{1,0}(G)=\cn(G)$ and $\ggn_{1,0}(G)=1$.
Moreover, graphs $G$ such as the path $P_n$ with cop number 1 satisfy $\gn_{s,d}(G)=\ggn_{s,d}(G)$.

\begin{proposition}[\cite{cohen18}]
Given a graph $G$ and $s\geq 1$ and $d\geq 0$,
$$\ggn_{s,d}(G)\leq\gn_{s,d}(G)\leq\ggn_{s,d}(G)+\cn(G)-1,$$ where $\cn(G)$ is the cop number of $G$.
\label{claim1}
\end{proposition}

In \cite{cohen18}, it was proved that the guard number is NP-hard in general graphs and a directed version is PSPACE-hard even in DAGs. The authors left open the question of the PSPACE-hardness on the undirected case and the question of the spy game on grids: ``\emph{Many open questions remain such as the characterization of the guard-number in other graph classes, e.g., in grids}'' \cite{cohen18}.
Moreover, in 2020, it was proved \cite{cohen18alg} that the guard number is computable in polynomial time for trees by using Linear Programming and a fractional relaxation of the game. The authors also obtained an upper bound for the fractional guard number on the square grid $G_{n\times n}$, a parameter different from the guard number, which was proved to be equal in trees.

In this paper, we investigate the spy game variant in which the guards are placed first and the parameter $\gn(G)$. Here we consider the decision problem of the spy game with instance $G$ and $k$ as the problem of deciding if $k$ guards have a winning strategy in $G$ when they are placed first.
We first obtain a polynomial time algorithm when the the number $k$ of guards is a constant, for any speed $s\geq 1$ and distance $d\geq 0$. We use this to obtain a fixed parameter tractable (fpt) algorithm for the spy game on the $P_4$-fewness of the graph, which is a relevant result, since games on graphs are expected to be PSPACE-hard in general graphs and therefore it is not expected to find polynomial algorithms for many classes of graphs.
We also obtain a tight upper bound on the strong product $G_1\boxtimes G_2$ of graphs $G_1$ and $G_2$. 
Finally, we prove that for every speed $s\geq 2$ and distance $d\geq 0$ the spy game decision problem is NP-hard even in bipartite graphs with diameter at most $4\cdot(d+1)$. This hardness result is a generalization of the NP-hardness reduction for general graphs in \cite{cohen18}.

\section{Polynomial algorithm for fixed number $k$ of guards}

We prove in this section that the spy game decision problem is solvable in polynomial time $O(n^{3k+2})$ for any speed $s\geq 1$ and distance $d\geq 0$, where $k$ is the number of guards, which leads to an XP algorithm parameterized by the number $k$ of guards. For this, we first define spy and guard configurations.

\begin{definition}
Given a graph $G$ and integers $s\geq 1$, $d\geq 0$, $k\geq 1$, let a \emph{configuration} in $G$ be a possible scenario of the spy game, with the spy and all $k$ guards occupying vertices. There are two types of configurations: \emph{spy configuration} (before the spy's move) and \emph{guard configuration} (before the guards' move). A spy configuration may be equal to a guard configuration (the only difference is that the spy is the next to move). Clearly there are exactly $2n^{k+1}$ configurations (2 possibilities for the next to move, $n$ possible vertices for the spy and all $k$ guards).
A spy configuration $C_1$ \emph{leads to} a guard configuration $C_2$ if $C_2$ can be obtained from $C_1$ by moving the spy along at most $s$ edges. A guard configuration $C_1$ \emph{leads to} a spy configuration $C_2$ if $C_2$ can be obtained from $C_1$ by moving the guards along at most one edge each.
Let the \emph{digraph $D^*$} defined as follows: for every spy or guard configuration $C$, create an associated vertex $C$ in $D^*$. If a configuration $C_1$ leads to a configuration $C_2$, we add in $D^*$ the directed edge from $C_1$ to $C_2$. Clearly, $D^*$ has exactly $2n^{k+1}$ vertices.
\end{definition}

In the next theorem, we show how to proceed in this constructed digraph $D^*$ in order to solve the decision problem of the spy game variant in which the guards are placed first.

\begin{theorem}\label{teo-XP}
Let $s\geq 1$ and $d\geq 0$.
Given a graph $G$ with $n$ vertices and $k\geq 1$, the spy game decision problem with instance $G$ and $k$ is solvable in XP time $O(n^{3k+2})$.
\end{theorem}

\begin{proof}
A \emph{controlled configuration} is a \emph{spy configuration} in which there is a guard at distance at least $d$ from the spy.
First, we have to determine the \emph{winning configurations}, which are the controlled configurations from which the guards can guarantee that all the next spy configurations are controlled. We call this Phase 1.

In Phase 1, first mark all controlled (spy) configurations as winning. Next, repeat the following procedure until no more vertices of $D^*$ are dismarked. Let $C_1$ be a marked (spy) configuration. If there is a guard configuration $C_2$ such that $C_1$ leads to $C_2$ but $C_2$ cannot lead to a marked configuration, then dismark $C_1$ as winning.

After this, we have the winning configurations (for the guards) from which the guards can always force controlled configurations. After Phase 1, the winning configurations are the only ones which are marked.

Next, we start Phase 2 to decide if the guards can force a winning configuration in the game. In Phase 2, no configuration is dismarked.
Repeat the following procedure until no more vertices of $D^*$ are marked.
For every guard configuration $C$, mark the vertex $C$ of $D^*$ if there exists at least one marked out-neighbor of $C$ (in words, there is a guards' move which leads to a winning configuration for the guards after some moves).
Moreover, for every spy configuration $C$, mark the vertex $C$ in $D^*$ if all out-neighbors of $C$ are marked (in words, any spy's move will lead to a winning configuration for the guards after some moves).

Finally, at the end, if there are (not necessarily distinct) vertices $u_1,\ldots,u_k$ of $G$ such that, for every guard configuration $C$ with the guard $i$ occupying $u_i$ ($1\leq i\leq k$), $C$ is marked in $D^*$, then the guards have a winning strategy (by occupying these vertices first). Otherwise, the spy has a winning strategy. This is because the guards can force a sequence of configurations which leads to a winning configuration, from which all the next configurations in the game are controlled.

By applying breadth-first search to each vertex of $G$ as a preprocessing to check distances, we can obtain all out-neighbors of a vertex $C$ in $D^*$ in time $O(n^3)$ and consequently $D^*$ can be constructed in time $O(n^{k+4})$.

In Phase 1, the algorithm has at most $n^{k+1}$ iterations, since at least one spy configuration is dismarked in each iteration, and every iteration takes time $O(n^{k+1}n^k)$, since each spy configuration leads to at most $n$ guard configurations and each guard configuration leads to at most $n^k$ spy configurations. Thus Phase 1 takes time $O(n^{3k+2})$.
In Phase 2, the algorithm has at most $2\cdot n^{k+1}$ iterations, since at least one vertex is marked in each iteration, and every iteration takes time $O(n^{k+1}n^k)$, since each spy configuration leads to at most $n$ guard configurations and each guard configuration leads to at most $n^k$ spy configurations. Thus Phase 2 also takes time $O(n^{3k+2})$.
\qed
\end{proof}

\section{An fpt-algorithm parameterized by the $P_4$-fewness of the graph}

A graph $G$ is a $(q,q-4)$-graph for some integer $q\geq 4$ if every subset of at most $q$ vertices induces at most $q-4$ distinct $P_4$'s.
For instance, cographs and $P_4$-sparse graphs are exactly the $(q,q-4)$-graphs for $q=4$ and $q=5$, respectively. The $P_4$-fewness $q(G)$ of a graph $G$ is the minimum $q\geq 4$ such that $G$ is a $(q,q-4)$-graph \cite{q-sula,q-claudia,q-yw,q-nicolas}. These graphs have received a lot of attention in the literature (under the expression ``graphs with few $P_4$'s''), since they are on the top of a very known hierarchy of graph classes, including cographs, $P_4$-reducible, $P_4$-sparse and $P_4$-lite graphs.
They also have a nice recursive decomposition based on unions, joins, spiders and small separable p-components~\cite{BO98}, described below.

Let $G_1=(V_1,E_1)$ and $G_2=(V_2,E_2)$ be two vertex disjoint graphs.  The
{\em disjoint union} of $G_1$ and $G_2$ is the graph $G_1\cup G_2=(V_1\cup
V_2,E_1\cup E_2)$. The {\em join} of $G_1$ and $G_2$ is the graph
$G_1+G_2=(V_1\cup V_2,E_1\cup E_2\cup \{uv:\ u\in V_1,\ v\in V_2\})$.
A \emph{spider} is a graph whose vertex set has a partition $(R,C,S)$, where
$C=\{c_1,\ldots,c_p\}$ and $S=\{s_1,\ldots,s_p\}$, for $p\geq 2$, are a clique
and a stable set, respectively; $s_i$ is adjacent to $c_j$ if and only if $i=j$
(a thin spider), or $s_i$ is adjacent to $c_j$ if and only if $i\not=j$ (a thick
spider); and every vertex of $R$ is adjacent to each vertex of $C$ and
non-adjacent to each vertex of $S$.

We say that a graph is {\it p-connected} (path connected) if, for every
nontrivial bipartition of its vertex set, there is a crossing $P_4$, i.e. an
induced $P_4$ with vertices in both parts. A {\it p-component} is a maximal
p-connected subgraph. We say that a p-connected graph $H$ is \emph{separable}
if its vertex set has a bipartition $(H_1,H_2)$ such that every crossing $P_4$
has its endpoints in $H_1$ and its midpoints in $H_2$.
Note that, if $G$ is a spider $(R,C,S)$, then $G[C\cup S]$ is a separable
p-component with bipartition $(C,S)$.

\begin{theorem}{\rm [Primeval Decomposition \cite{BO98}]}\label{teo-primeval}
If $G$ is a $(q,q-4)$-graph,  then one of the following holds:
\begin{itemize}
\item[(a)] $G$ is the disjoint union or the join of two $(q,q-4)$-graphs;
\item[(b)] $G$ is a spider $(R,C,S)$ and $G[R]$ is a $(q,q-4)$-graph; 
\item[(c)] $G$ contains a separable p-component $H$ with $|V(H)|<q$ and bipartition $(H_1,H_2)$, such that $G-H$ is a $(q,q-4)$-graph and every vertex of $G-H$ is adjacent to every vertex of $H_1$ and non-adjacent to every vertex of~$H_2$; or
\item[(d)] $G$ has at most $q$ vertices or $V(G)=\emptyset$.
\end{itemize}
\end{theorem}

As a consequence, a $(q,q-4)$-graph $G$ can be decomposed by successively
applying Theorem~\ref{teo-primeval} as follows: If (a) holds, apply the
theorem to each component of $G$ or $\overline{G}$. If (b) holds, apply the
theorem to $G[R]$. Finally, if (c) holds, then apply the theorem to
$G-H$. This decomposition can be obtained in linear time~\cite{Baumann96,Jamison92}.
From this, the spy game variant in which the guards are placed first can be solved by composing the strategies of the parts obtained in each case (a)-(d).
We begin with unions and joins.

\begin{lemma}\label{lem-union-join}
Let $s\geq 1$ and $d\geq 0$. Then
$\gn_{s,d}(G_1\cup G_2)=\gn_{s,d}(G_1)+\gn_{s,d}(G_2)$.
Moreover,
\[
  \gn_{s,d}(G_1+G_2)\ =\ 
 \begin{cases}
    1, &\mbox{if $d\geq 1$ or $G_1$ and $G_2$ are complete},\\
    1, &\mbox{if $d=0$ and $s=1$ and $\min\{\cn(G_1),\cn(G_2)\}=1$},\\
    2, &\mbox{otherwise}.
 \end{cases}
\]
\end{lemma}

\begin{proof}
In $G_1\cup G_2$, since the guards are placed first, they must guarantee that they win if the spy is placed in either $G_1$ or $G_2$. Thus $\gn_{s,d}(G_1\cup G_2)=\gn_{s,d}(G_1)+\gn_{s,d}(G_2)$.

In $G_1+G_2$, if $d\geq 1$, then one guard is sufficient to control the spy: if the spy is in $G_1$ (resp. $G_2$), the guard goes to $G_2$ (resp $G_1$). If $d=0$ and $G_1$ and $G_2$ are complete graphs, then $G_1+G_2$ is also a complete graph and consequently one guard is sufficient: just go to the same vertex of the spy in each turn.
If $d=0$ and $s=1$ and $\cn(G_1)=1$, then one guard can be placed in $G_1$ and control the spy, by following her in $G_1$ using his Cops and robber strategy or just occupying the same vertex in $G_2$ if the spy goes to $G_2$.

Now assume that $d=0$, $s\geq 2$ and w.l.g. $G_1$ is not complete. Let $v_1$ and $v_2$ be two non-adjacent vertices of $G_1$. Then, in the game on $G=G_1+G_2$, the spy can go from $v_1$ to $v_2$ in one step passing through a vertex of $G_2$ and one guard cannot keep up her (recall that $d=0$). Then two guards are necessary. Moreover, this is sufficient: just maintain in each round one guard in the same vertex of the spy and the second guard in the other graph: $G_1$ (resp. $G_2$) if the spy is in $G_2$ (resp. $G_1$).
\qed
\end{proof}

The next lemma determines the guard number on spiders.

\begin{lemma}\label{lem-spider}
Let $G$ be a spider $(R,C,S)$. If $d\geq 1$ or $s=1$, then $\gn_{s,d}(G)=1$.
Moreover, if $s\geq 2$ and $G$ is a thick spider, then $\gn_{s,0}(G)=3$ if $R\ne\emptyset$, and $\gn_{s,0}(G)=2$, otherwise.
Finally, if $s\geq 2$ and $G$ is a thin spider, then $\gn_{s,0}(G)=|C|+1$ if $R\ne\emptyset$, and $\gn_{s,0}(G)=|C|$, otherwise.
\end{lemma}

\begin{proof}
From the definition of spider, every vertex of $G$ has a neighbor in the clique $C$.
Thus, if $d\geq 1$ or $s=1$, one guard in $C$ is sufficient: just move to a vertex also in $C$ at distance 1 from the spy. From now on, consider surveillance distance $d=0$ and $s\geq 2$.

First suppose that $G$ is a thick spider. Since the spy can move in one step from a vertex of $C$ (or $R$ if $R\ne\emptyset$) to any vertex of $S$, then each vertex of $S$ must have a guard in its closed neighborhood before the spy's move. Thus, since every vertex of $C$ has a non-neighbor in $S$, 2 guards are necessary if $R=\emptyset$ and $3$ guards are necessary if $R\ne\emptyset$.
Moreover, if $R=\emptyset$, $2$ guards are sufficient: if the spy is in $S$, one guard goes to the same vertex of the spy and other guard to the only non-neighbor in $C$, since it dominates all other vertices of $S$; if the spy is in $C$, both guards go to $C$.
Finally, if $R\ne\emptyset$, $3$ guards are sufficient: just maintain in all rounds two guards in distinct vertices of $C$ and one guard in the same vertex of the spy. This is possible since any two distinct vertices of $C$ dominates all vertices of $G$.

Now suppose that $G$ is a thin spider. Since the spy can move in one step from a vertex of $C$ (or $R$ if $R\ne\emptyset$) to any vertex of $S$, then each vertex of $S$ must have a distinct guard in its closed neighborhood before the spy's move. Thus, $|C|=|S|$ guards are necessary if $R=\emptyset$ and $|C|+1$ guards are necessary if $R\ne\emptyset$.
Moreover, if $R=\emptyset$, $|C|$ guards are sufficient: just maintain one guard controlling each vertex of $C$ and its only neighbor in $S$. Finally, if $R\ne\emptyset$, $|C|+1$ guards are sufficient: just maintain in all rounds every vertex of $C$ with a guard and one guard in the same vertex of the spy. This is possible since all vertices of $G$ has a neighbor in $C$.
\qed
\end{proof}

The next lemma determines the guard number on separable p-components.

\begin{lemma}\label{lem-pconn}
Let $G$ be a graph with a separable p-component $H$ with $|V(H)|<q$ and bipartition $(H_1,H_2)$ of $H$ such that any vertex of $G-H$ is adjacent to every vertex of $H_1$ and non-adjacent to every vertex of $H_2$. Then, for any $s\geq 1$ and $d\geq 0$, $\gn_{s,d}(G)=\gn_{s,d}(G_R)\leq q$, where $G_R$ is the (reduced) graph obtained from $G$ by replacing $G-H$ by two adjacent or non-adjacent vertices, depending whether $G-H$ is complete or not, respectively.
\end{lemma}

\begin{proof}
Let $v_1$ and $v_2$ be the vertices which replaced $G-H$ in order to build $G_R$.

For any game on $G$, we can simulate it on $G_R$. For this, if a player (spy or guard) goes to $H$ in the game on $G$, it goes to the same vertex of $H$ in the game on $G_R$; if a player goes to $G-H$ from a vertex not in $G-H$ in the game on $G$, it goes to $v_1$ or $v_2$ in the game on $G_R$ (preferring the same vertex of the spy, if she is in $\{v_1,v_2\}$ in the game on $G_R$); if a player moves from a vertex of $G-H$ to a neighbor in $G-H$ in the game on $G$, it maintain the position in the game on $G_R$; if the spy moves from a vertex of $G-H$ to a non-adjacent vertex of $G-H$ in the game on $G$, she goes from $v_1$ to $v_2$ or from $v_2$ to $v_1$ in the game on $G_R$.

On the other hand, for any game on $G_R$, we can simulate it on $G$. For this, if $G-H$ is not complete, let $u_1$ and $u_2$ be two non-adjacent vertices of $G-H$; otherwise, let $u_1$ and $u_2$ be vertices of $G-H$ ($u_1$ and $u_2$ may be the same vertex if $G-H$ has only one vertex). If a player (spy or guard) goes to $H$ in the game on $G_R$, it goes to the same vertex of $H$ in the game on $G$; if a player goes to $v_1$ (resp. $v_2$) from a vertex not in $\{v_1,v_2\}$ in the game on $G_R$, it goes to $u_1$ (resp. $u_2$) in the game on $G$; if a player moves from $v_1$ to $v_2$ (resp. $v_2$ to $v_1$) in the game on $G_R$, it goes from $u_1$ to $u_2$ (resp. $u_2$ to $u_1$) in the game on $G$.

With this, if the guards win in $G$, they also win in $G_R$. Moreover, if the spy wins in $G$, it is possible to adapt the spy's strategy described above in $G_R$ to guarantee a win in $G_R$. Therefore, there is a winning strategy in $G$ for any winning strategy in $G_R$, and vice-versa.
\qed
\end{proof}

\begin{theorem}
The spy game decision problem has an fpt-algorithm parameterized by the $P_4$-fewness $q(G)$ of the graph $G$, with time $O(m+q^{3q+3}\cdot n)$, where $q=q(G)$ and $m$ and $n$ are the number of edges and vertices of $G$, respectively.
\end{theorem}

\begin{proof}
From the primeval decomposition of $(q,q-4)$-graphs, Lemmas \ref{lem-union-join}, \ref{lem-spider} and \ref{lem-pconn} and the XP algorithm of Theorem \ref{teo-XP} applied in the reduced graph $G_R$ of Lemma \ref{lem-pconn} or in $G$ itself when it has at most $q$ vertices.
\qed
\end{proof}

\section{Graph products and grids}

The strong product $G_1\boxtimes G_2$ of two graphs $G_1$ and $G_2$ is the graph whose vertex set $V(G_1\boxtimes G_2)$ is the cartesian product $V(G_1)\times V(G_2)$ in which distinct vertices $(u_1,u_2)$ and $(v_1,v_2)$ are adjacent in $G_1\boxtimes G_2$ if and only if (a) $u_1=v_1$ and $u_2v_2\in E(G_2)$, or (b) $u_2=v_2$ and $u_1v_1\in E(G_1)$, or (c) $u_1v_1\in E(G_1)$ and $u_2v_2\in E(G_2)$. A King grid is the strong product of two path graphs $P_n\boxtimes P_m$. Figure \ref{fig-king02} shows the King grids $P_3\boxtimes P_3$ and $P_5\boxtimes P_5$.

\subsection{Strong product}

We first obtain a general upper bound on the guard number $\gn_{s,d}(G_1\boxtimes G_2)$ for any two graphs $G_1$ and $G_2$ and any speed $s\geq 2$ and surveillance distance $d\geq 0$.

\begin{theorem}\label{teo-upper}
Let $s\geq 1$ and $d\geq 0$.
The guard number of the strong product of two graphs $G_1$ and $G_2$ satisfies $\gn_{s,d}(G_1\boxtimes G_2)\leq\gn_{s,d}(G_1)\cdot\gn_{s,d}(G_2)$. Moreover, the equality holds if $\gn_{s,d}(G_1)=1$ or $\gn_{s,d}(G_2)=1$.
\end{theorem}

\begin{proof}
Let $k_1=\gn_{s,d}(G_1)$ and $k_2=\gn_{s,d}(G_2)$.
Consider a winning strategy for the guards in $G_1$ with $k_1$ guards and a winning strategy for the guards in $G_2$ with $k_2$ guards when they are placed first. From this, consider the following strategy for the guards in $G_1\boxtimes G_2$. For every guard $g_1$ in $G_1$ and every guard $g_2$ in $G_2$, let $g_1g_2$ be a guard in $G_1\boxtimes G_2$, resulting in a total of $k_1\cdot k_2$ guards.

If the spy is placed initially at the vertex $(u_1,u_2)$ of $G_1\boxtimes G_2$, we follow the winning strategy in $G_1$ (resp. $G_2$) considering that the spy was placed at $u_1$ (resp. $u_2$) initially. If the spy goes from a vertex $(v'_1,v'_2)$ to $(v_1,v_2)$ at some moment of the game, we follow the winning strategy in $G_1$ (resp. $G_2$) considering that the spy goes from $v'_1$ to $v_1$ (resp. from $v'_2$ to $v_2$). The spy can do that in $G_1$ and $G_2$ from the topology of the strong product: the distance in $G_1$ from $v'_1$ to $v_1$ (resp. $G_2$ from $v'_2$ to $v_2$) is at most the distance from $(v'_1,v'_2)$ to $(v_1,v_2)$ in $G_1\boxtimes G_2$, which is at most $s$ (the speed of the spy).

These two combined winning strategies in $G_1$ and $G_2$ lead to a strategy for the guards in the spy game on $G_1\boxtimes G_2$. If the guard $g_1$ of $G_1$ goes from $w'_1$ to $w_1$ and the guard $g_2$ of $G_2$ goes from $w'_2$ to $w_2$, we consider that the guard $g_1g_2$ of $G_1\boxtimes G_2$ goes from $(w'_1,w'_2)$ to $(w_1,w_2)$. The guards can do that in $G_1\boxtimes G_2$ also from the topology of the strong product: the distance from $(w'_1,w'_2)$ to $(w_1,w_2)$ in $G_1\boxtimes G_2$ is the maximum between the distance from $w'_1$ to $w_1$ in $G_1$ and the distance from $w_2'$ to $w_2$ in $G_2$, which is at most $1$.

Finally, since after a finite number of turns there is always a guard $g_1$ and a guard $g_2$ at distance at most $d$ from the spy in the games played in $G_1$ and in $G_2$, respectively, we have that after a finite number of turns there is always a guard $g_1g_2$ at distance at most $d$ from the spy in the game played in $G_1\boxtimes G_2$, as already observed from the topology of the strong product. Since this is a winning strategy for the guards in $G_1\boxtimes G_2$, we are done.

If $\gn_{s,d}(G_2)=1$, then $\gn_{s,d}(G_1\boxtimes G_2)\leq \gn_{s,d}(G_1)\cdot \gn_{s,d}(G_2)=\gn_{s,d}(G_1)$. Since $G_1$ is an induced subgraph of $G_1\boxtimes G_2$ and the distance between any two vertices of $G_1$ is equal to the distance between the same vertices in $G_1\boxtimes G_2$ (from the topology of the strong product), then $\gn_{s,d}(G_1\boxtimes G_2)=\gn_{s,d}(G_1)$. Analogously, $\gn_{s,d}(G_1\boxtimes G_2)=\gn_{s,d}(G_2)$ if $\gn_{s,d}(G_1)=1$.
\qed
\end{proof}

From this, the question about the tightness of the bound arises when the guard numbers of both graphs are at least 2.
In the following subsection, we will show that the general upper bound just derived is the best possible. We will present classes of examples where $gn_{s,d}(G_1) = gn_{s,d}(G_2) = 2$ and yet $gn_{s,d}(G_1\boxtimes G_2)\leq 2 < 4$, as well as examples where $gn_{s,d}(G_1\boxtimes G_2) = 4$ instead. In all of those examples, $G_1$ and $G_2$ will be paths, and so $G_1\boxtimes G_2$ will be King grids.

\subsection{King grids}

Regarding paths, it was proved in \cite{cohen18} that the exact value of $\gn_{s,d}(P_n)$ is
\[
  \gn_{s,d}(P_n) = \left\lceil \frac{n}{2d+2+\left\lfloor \frac{2d}{s-1} \right\rfloor} \right\rceil
\]

First we show that $\gn_{s,d}(P_{2d+3}\boxtimes P_{2d+3})\leq 2$ for any $d\geq 0$ and $s\geq 2$. Notice that $\gn_{s,d}(P_{2d+3})=2$ for any $d\geq 0$ and $s\geq 2d+2$, since a spy in the first vertex $v_0$ of the path $P_{2d+3}$ can go to the last vertex $v_{2d+2}$ of the path, but a guard surveilling the spy in the vertex $v_d$ of the path $P_{2d+3}$ cannot jump to the vertex $v_{d+2}$ to keep surveilling the spy.
Moreover, two guards in the vertices $v_d$ and $v_{d+2}$ are sufficient to surveil the spy.
That is, $\gn_{s,d}(P_{2d+3})^2=4$.

\begin{lemma}
Let $d\geq 0$ and $s\geq 1$ be fixed integers. Then $\gn_{s,d}(P_{2d+3}\boxtimes P_{2d+3})\leq 2$.
\end{lemma}

\begin{proof}
Let us show a guards winning strategy with two guards.
Let $V=\{-(d+1),\ldots,0,\ldots,(d+1)\}\subseteq\ZZ$ ($|V|=2d+3$).
Consider that the vertex set of $P_{2d+3}\boxtimes P_{2d+3}$ is $V\times V$.
The guards winning strategy with two guards in $P_{2d+3}\boxtimes P_{2d+3}$ is the following: one guard is always in the vertex $(0,0)$ and the other guard is in a vertex of $\{-1,0,1\}^2$ surveilling the spy, which is possible since the surveilling distance is $d$ and the maximum distance from $(0,0)$ in $P_{2d+3}\boxtimes P_{2d+3}$ is $d+1$. At any moment of the game, the guards change roles: the guard in $(0,0)$ goes to a vertex of $\{-1,0,1\}^2$ (surveilling the spy) and the other guard goes to $(0,0)$. Since this is a winning strategy for the guards, we are done. See examples in Figure \ref{fig-king02}.
\qed
\end{proof}

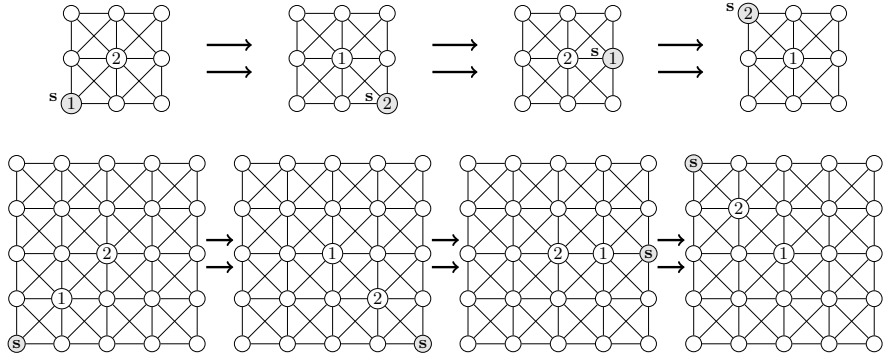
\begin{figure}[hbtp]\centering
\scalebox{0.75}{
\begin{tikzpicture}[scale=0.8]
\tikzstyle{uvertex}=[draw,circle,fill=white!25,minimum size=8pt,inner sep=1pt]
\tikzstyle{cvertex}=[draw,circle,fill=black!10,minimum size=8pt,inner sep=1pt]

\node[cvertex] (a1) at (1,1) {1};
\node[uvertex] (b1) at (1,2) {};
\node[uvertex] (c1) at (1,3) {};
\node[uvertex] (a2) at (2,1) {};
\node[uvertex] (b2) at (2,2) {2};
\node[uvertex] (c2) at (2,3) {};
\node[uvertex] (a3) at (3,1) {};
\node[uvertex] (b3) at (3,2) {};
\node[uvertex] (c3) at (3,3) {};
\path[-,very thin]
   (a1) edge (b1) edge (a2) edge (b2)
   (a2) edge (b2) edge (a3) edge (b3)
   (a3) edge (b3)
   (b1) edge (c1) edge (a2) edge (b2) edge (c2)
   (b2) edge (c2) edge (a3) edge (b3) edge (c3)
   (b3) edge (c3)
   (c1) edge (b2) edge (c2)
   (c2) edge (b3) edge (c3);
\node[]() at (0.6,1.15) {\textbf{s}};

\node[uvertex] (a1) at (6,1) {};
\node[uvertex] (b1) at (6,2) {};
\node[uvertex] (c1) at (6,3) {};
\node[uvertex] (a2) at (7,1) {};
\node[uvertex] (b2) at (7,2) {1};
\node[uvertex] (c2) at (7,3) {};
\node[cvertex] (a3) at (8,1) {2};
\node[uvertex] (b3) at (8,2) {};
\node[uvertex] (c3) at (8,3) {};
\path[-,very thin]
   (a1) edge (b1) edge (a2) edge (b2)
   (a2) edge (b2) edge (a3) edge (b3)
   (a3) edge (b3)
   (b1) edge (c1) edge (a2) edge (b2) edge (c2)
   (b2) edge (c2) edge (a3) edge (b3) edge (c3)
   (b3) edge (c3)
   (c1) edge (b2) edge (c2)
   (c2) edge (b3) edge (c3);
\node[]() at (7.6,1.15) {\textbf{s}};

\node[uvertex] (a1) at (11,1) {};
\node[uvertex] (b1) at (11,2) {};
\node[uvertex] (c1) at (11,3) {};
\node[uvertex] (a2) at (12,1) {};
\node[uvertex] (b2) at (12,2) {2};
\node[uvertex] (c2) at (12,3) {};
\node[uvertex] (a3) at (13,1) {};
\node[cvertex] (b3) at (13,2) {1};
\node[uvertex] (c3) at (13,3) {};
\path[-,very thin]
   (a1) edge (b1) edge (a2) edge (b2)
   (a2) edge (b2) edge (a3) edge (b3)
   (a3) edge (b3)
   (b1) edge (c1) edge (a2) edge (b2) edge (c2)
   (b2) edge (c2) edge (a3) edge (b3) edge (c3)
   (b3) edge (c3)
   (c1) edge (b2) edge (c2)
   (c2) edge (b3) edge (c3);
\node[]() at (12.6,2.15) {\textbf{s}};

\node[uvertex] (a1) at (16,1) {};
\node[uvertex] (b1) at (16,2) {};
\node[cvertex] (c1) at (16,3) {2};
\node[uvertex] (a2) at (17,1) {};
\node[uvertex] (b2) at (17,2) {1};
\node[uvertex] (c2) at (17,3) {};
\node[uvertex] (a3) at (18,1) {};
\node[uvertex] (b3) at (18,2) {};
\node[uvertex] (c3) at (18,3) {};
\path[-,very thin]
   (a1) edge (b1) edge (a2) edge (b2)
   (a2) edge (b2) edge (a3) edge (b3)
   (a3) edge (b3)
   (b1) edge (c1) edge (a2) edge (b2) edge (c2)
   (b2) edge (c2) edge (a3) edge (b3) edge (c3)
   (b3) edge (c3)
   (c1) edge (b2) edge (c2)
   (c2) edge (b3) edge (c3);
\node[]() at (15.6,3.15) {\textbf{s}};

\draw[->, very thick] ( 4,2-0.3) -- ( 5,2-0.3);
\draw[->, very thick] ( 9,2-0.3) -- (10,2-0.3);
\draw[->, very thick] (14,2-0.3) -- (15,2-0.3);
\draw[->, very thick] ( 4,2+0.3) -- ( 5,2+0.3);
\draw[->, very thick] ( 9,2+0.3) -- (10,2+0.3);
\draw[->, very thick] (14,2+0.3) -- (15,2+0.3);
\end{tikzpicture}}\\

\vspace{0.5cm}

\scalebox{0.75}{
\begin{tikzpicture}[scale=0.8]
\tikzstyle{uvertex}=[draw,circle,fill=white!25,minimum size=8pt,inner sep=1pt]
\tikzstyle{cvertex}=[draw,circle,fill=black!10,minimum size=8pt,inner sep=1pt]

\node[cvertex] (a0) at (0,0) {\textbf{s}};
\node[uvertex] (b0) at (0,1) {};
\node[uvertex] (c0) at (0,2) {};
\node[uvertex] (d0) at (0,3) {};
\node[uvertex] (e0) at (0,4) {};
\node[uvertex] (a1) at (1,0) {};
\node[uvertex] (b1) at (1,1) {1};
\node[uvertex] (c1) at (1,2) {};
\node[uvertex] (d1) at (1,3) {};
\node[uvertex] (e1) at (1,4) {};
\node[uvertex] (a2) at (2,0) {};
\node[uvertex] (b2) at (2,1) {};
\node[uvertex] (c2) at (2,2) {2};
\node[uvertex] (d2) at (2,3) {};
\node[uvertex] (e2) at (2,4) {};
\node[uvertex] (a3) at (3,0) {};
\node[uvertex] (b3) at (3,1) {};
\node[uvertex] (c3) at (3,2) {};
\node[uvertex] (d3) at (3,3) {};
\node[uvertex] (e3) at (3,4) {};
\node[uvertex] (a4) at (4,0) {};
\node[uvertex] (b4) at (4,1) {};
\node[uvertex] (c4) at (4,2) {};
\node[uvertex] (d4) at (4,3) {};
\node[uvertex] (e4) at (4,4) {};
\path[-,very thin]
   (a0) edge (b0) edge (a1) edge (b1)
   (a1) edge (b1) edge (a2) edge (b2)
   (a2) edge (b2) edge (a3) edge (b3)
   (a3) edge (b3) edge (a4) edge (b4)
   (a4) edge (b4)
   (b0) edge (c0) edge (a1) edge (b1) edge (c1)
   (b1) edge (c1) edge (a2) edge (b2) edge (c2)
   (b2) edge (c2) edge (a3) edge (b3) edge (c3)
   (b3) edge (c3) edge (a4) edge (b4) edge (c4)
   (b4) edge (c4)
   (c0) edge (d0) edge (b1) edge (c1) edge (d1)
   (c1) edge (d1) edge (b2) edge (c2) edge (d2)
   (c2) edge (d2) edge (b3) edge (c3) edge (d3)
   (c3) edge (d3) edge (b4) edge (c4) edge (d4)
   (c4) edge (d4)

   (d0) edge (e0) edge (c1) edge (d1) edge (e1)
   (d1) edge (e1) edge (c2) edge (d2) edge (e2)
   (d2) edge (e2) edge (c3) edge (d3) edge (e3)
   (d3) edge (e3) edge (c4) edge (d4) edge (e4)
   (d4) edge (e4)

   (e0) edge (d1) edge (e1)
   (e1) edge (d2) edge (e2)
   (e2) edge (d3) edge (e3)
   (e3) edge (d4) edge (e4);

\node[uvertex] (a0) at (0+5,0) {};
\node[uvertex] (b0) at (0+5,1) {};
\node[uvertex] (c0) at (0+5,2) {};
\node[uvertex] (d0) at (0+5,3) {};
\node[uvertex] (e0) at (0+5,4) {};
\node[uvertex] (a1) at (1+5,0) {};
\node[uvertex] (b1) at (1+5,1) {};
\node[uvertex] (c1) at (1+5,2) {};
\node[uvertex] (d1) at (1+5,3) {};
\node[uvertex] (e1) at (1+5,4) {};
\node[uvertex] (a2) at (2+5,0) {};
\node[uvertex] (b2) at (2+5,1) {};
\node[uvertex] (c2) at (2+5,2) {1};
\node[uvertex] (d2) at (2+5,3) {};
\node[uvertex] (e2) at (2+5,4) {};
\node[uvertex] (a3) at (3+5,0) {};
\node[uvertex] (b3) at (3+5,1) {2};
\node[uvertex] (c3) at (3+5,2) {};
\node[uvertex] (d3) at (3+5,3) {};
\node[uvertex] (e3) at (3+5,4) {};
\node[cvertex] (a4) at (4+5,0) {\textbf{s}};
\node[uvertex] (b4) at (4+5,1) {};
\node[uvertex] (c4) at (4+5,2) {};
\node[uvertex] (d4) at (4+5,3) {};
\node[uvertex] (e4) at (4+5,4) {};
\path[-,very thin]
   (a0) edge (b0) edge (a1) edge (b1)
   (a1) edge (b1) edge (a2) edge (b2)
   (a2) edge (b2) edge (a3) edge (b3)
   (a3) edge (b3) edge (a4) edge (b4)
   (a4) edge (b4)
   (b0) edge (c0) edge (a1) edge (b1) edge (c1)
   (b1) edge (c1) edge (a2) edge (b2) edge (c2)
   (b2) edge (c2) edge (a3) edge (b3) edge (c3)
   (b3) edge (c3) edge (a4) edge (b4) edge (c4)
   (b4) edge (c4)
   (c0) edge (d0) edge (b1) edge (c1) edge (d1)
   (c1) edge (d1) edge (b2) edge (c2) edge (d2)
   (c2) edge (d2) edge (b3) edge (c3) edge (d3)
   (c3) edge (d3) edge (b4) edge (c4) edge (d4)
   (c4) edge (d4)

   (d0) edge (e0) edge (c1) edge (d1) edge (e1)
   (d1) edge (e1) edge (c2) edge (d2) edge (e2)
   (d2) edge (e2) edge (c3) edge (d3) edge (e3)
   (d3) edge (e3) edge (c4) edge (d4) edge (e4)
   (d4) edge (e4)

   (e0) edge (d1) edge (e1)
   (e1) edge (d2) edge (e2)
   (e2) edge (d3) edge (e3)
   (e3) edge (d4) edge (e4);

\node[uvertex] (a0) at (0+10,0) {};
\node[uvertex] (b0) at (0+10,1) {};
\node[uvertex] (c0) at (0+10,2) {};
\node[uvertex] (d0) at (0+10,3) {};
\node[uvertex] (e0) at (0+10,4) {};
\node[uvertex] (a1) at (1+10,0) {};
\node[uvertex] (b1) at (1+10,1) {};
\node[uvertex] (c1) at (1+10,2) {};
\node[uvertex] (d1) at (1+10,3) {};
\node[uvertex] (e1) at (1+10,4) {};
\node[uvertex] (a2) at (2+10,0) {};
\node[uvertex] (b2) at (2+10,1) {};
\node[uvertex] (c2) at (2+10,2) {2};
\node[uvertex] (d2) at (2+10,3) {};
\node[uvertex] (e2) at (2+10,4) {};
\node[uvertex] (a3) at (3+10,0) {};
\node[uvertex] (b3) at (3+10,1) {};
\node[uvertex] (c3) at (3+10,2) {1};
\node[uvertex] (d3) at (3+10,3) {};
\node[uvertex] (e3) at (3+10,4) {};
\node[uvertex] (a4) at (4+10,0) {};
\node[uvertex] (b4) at (4+10,1) {};
\node[cvertex] (c4) at (4+10,2) {\textbf{s}};
\node[uvertex] (d4) at (4+10,3) {};
\node[uvertex] (e4) at (4+10,4) {};
\path[-,very thin]
   (a0) edge (b0) edge (a1) edge (b1)
   (a1) edge (b1) edge (a2) edge (b2)
   (a2) edge (b2) edge (a3) edge (b3)
   (a3) edge (b3) edge (a4) edge (b4)
   (a4) edge (b4)
   (b0) edge (c0) edge (a1) edge (b1) edge (c1)
   (b1) edge (c1) edge (a2) edge (b2) edge (c2)
   (b2) edge (c2) edge (a3) edge (b3) edge (c3)
   (b3) edge (c3) edge (a4) edge (b4) edge (c4)
   (b4) edge (c4)
   (c0) edge (d0) edge (b1) edge (c1) edge (d1)
   (c1) edge (d1) edge (b2) edge (c2) edge (d2)
   (c2) edge (d2) edge (b3) edge (c3) edge (d3)
   (c3) edge (d3) edge (b4) edge (c4) edge (d4)
   (c4) edge (d4)

   (d0) edge (e0) edge (c1) edge (d1) edge (e1)
   (d1) edge (e1) edge (c2) edge (d2) edge (e2)
   (d2) edge (e2) edge (c3) edge (d3) edge (e3)
   (d3) edge (e3) edge (c4) edge (d4) edge (e4)
   (d4) edge (e4)

   (e0) edge (d1) edge (e1)
   (e1) edge (d2) edge (e2)
   (e2) edge (d3) edge (e3)
   (e3) edge (d4) edge (e4);

\node[uvertex] (a0) at (0+15,0) {};
\node[uvertex] (b0) at (0+15,1) {};
\node[uvertex] (c0) at (0+15,2) {};
\node[uvertex] (d0) at (0+15,3) {};
\node[cvertex] (e0) at (0+15,4) {\textbf{s}};
\node[uvertex] (a1) at (1+15,0) {};
\node[uvertex] (b1) at (1+15,1) {};
\node[uvertex] (c1) at (1+15,2) {};
\node[uvertex] (d1) at (1+15,3) {2};
\node[uvertex] (e1) at (1+15,4) {};
\node[uvertex] (a2) at (2+15,0) {};
\node[uvertex] (b2) at (2+15,1) {};
\node[uvertex] (c2) at (2+15,2) {1};
\node[uvertex] (d2) at (2+15,3) {};
\node[uvertex] (e2) at (2+15,4) {};
\node[uvertex] (a3) at (3+15,0) {};
\node[uvertex] (b3) at (3+15,1) {};
\node[uvertex] (c3) at (3+15,2) {};
\node[uvertex] (d3) at (3+15,3) {};
\node[uvertex] (e3) at (3+15,4) {};
\node[uvertex] (a4) at (4+15,0) {};
\node[uvertex] (b4) at (4+15,1) {};
\node[uvertex] (c4) at (4+15,2) {};
\node[uvertex] (d4) at (4+15,3) {};
\node[uvertex] (e4) at (4+15,4) {};
\path[-,very thin]
   (a0) edge (b0) edge (a1) edge (b1)
   (a1) edge (b1) edge (a2) edge (b2)
   (a2) edge (b2) edge (a3) edge (b3)
   (a3) edge (b3) edge (a4) edge (b4)
   (a4) edge (b4)
   (b0) edge (c0) edge (a1) edge (b1) edge (c1)
   (b1) edge (c1) edge (a2) edge (b2) edge (c2)
   (b2) edge (c2) edge (a3) edge (b3) edge (c3)
   (b3) edge (c3) edge (a4) edge (b4) edge (c4)
   (b4) edge (c4)
   (c0) edge (d0) edge (b1) edge (c1) edge (d1)
   (c1) edge (d1) edge (b2) edge (c2) edge (d2)
   (c2) edge (d2) edge (b3) edge (c3) edge (d3)
   (c3) edge (d3) edge (b4) edge (c4) edge (d4)
   (c4) edge (d4)

   (d0) edge (e0) edge (c1) edge (d1) edge (e1)
   (d1) edge (e1) edge (c2) edge (d2) edge (e2)
   (d2) edge (e2) edge (c3) edge (d3) edge (e3)
   (d3) edge (e3) edge (c4) edge (d4) edge (e4)
   (d4) edge (e4)

   (e0) edge (d1) edge (e1)
   (e1) edge (d2) edge (e2)
   (e2) edge (d3) edge (e3)
   (e3) edge (d4) edge (e4);

\draw[->, very thick] ( 4+.2,2-0.3) -- ( 5-.2,2-0.3);
\draw[->, very thick] ( 9+.2,2-0.3) -- (10-.2,2-0.3);
\draw[->, very thick] (14+.2,2-0.3) -- (15-.2,2-0.3);
\draw[->, very thick] ( 4+.2,2+0.3) -- ( 5-.2,2+0.3);
\draw[->, very thick] ( 9+.2,2+0.3) -- (10-.2,2+0.3);
\draw[->, very thick] (14+.2,2+0.3) -- (15-.2,2+0.3);
\end{tikzpicture}}

\caption{Guards winning strategies with 2 guards in the King grids $P_3\boxtimes P_3$ and $P_5\boxtimes P_5$ with $(s,d)=(2,0)$ and $(s,d)=(4,1)$, respectively. The spy is represented by $s$ and the guards by 1 and 2.}\label{fig-king02}
\end{figure}

Next we show that $\gn_{s,d}(P_{2d+4}\boxtimes P_{2d+4})=4$ for any $d\geq 0$ and $s\geq d+2$. Notice that $\gn_{s,d}(P_{2d+4})=2$ for any $d\geq 0$ and $s\geq d+2$, since a spy in the first vertex $v_0$ of the path $P_{2d+4}$ can go in two time steps to the last vertex $v_{2d+3}$ of the path, but a guard surveilling the spy in the vertex $v_d$ of the path $P_{2d+3}$ cannot go to the vertex $v_{d+3}$ in two time steps to keep surveilling the spy. Moreover, two guards in the vertices $v_d$ and $v_{d+3}$ are sufficient to surveil the spy. That is, $\gn_{s,d}(P_{2d+4})^2=4$. 

\begin{lemma}
Let $d\geq 0$ be a fixed integer and let $s\geq d+2$. Then $\gn_{s,d}(P_{2d+4}\boxtimes P_{2d+4})=4$.
\end{lemma}

\begin{proof}
Let us show a spy winning strategy with less than four guards. Let $G=P_{2d+4}\boxtimes P_{2d+4}$
Consider that the vertex set of $G$ is $V\times V$, where $V=\{-(d+2),\ldots,0,\ldots,d+1\}\subseteq\ZZ$ ($|V|=2d+4$).
The spy winning strategy in $G$ is the following: 
initially the spy is placed at the vertex $(0,0)$ independently of the initial positions of the guards and then she chooses one of the four ``corner'' vertices to go:  $c_1=(-d-2,-d-2)$, $c_2=(-d-2,d+1)$, $c_3=(d+1,-d-2)$ or $c_4=(d+1,d+1)$. Clearly, the spy can do that since its speed $s\geq d+2$.

It is easy to see that $V(G)$ can be partitioned into four subsets $A_1=[-d-2,-1]\times[-d-2,-1]$, $A_2=[-d-2,-1]\times[0,d+1]$, $A_3=[0,d+1]\times[-d-2,-1]$ and $A_4=[0,d+1]\times[0,d+1]$ with $(d+2)^2$ vertices each. Clearly the corner vertices satisfy $c_1\in A_1$, $c_2\in A_2$, $c_3\in A_3$ and $c_4\in A_4$. Let $i\in\{1,2,3,4\}$. Notice that all vertices at distance at most $d+1$ in $G$ from the corner vertex $c_i$ belongs to $A_i$. Thus, no vertex in $A_j$ with $j\ne i$ can surveil the corner $c_i$ in one time step.
This implies that at least four guards are necessary to surveil the four corner vertices and we are done, since, from Theorem \ref{teo-upper}, four guards are sufficient. See examples in Figure \ref{fig-king03}.
\qed
\end{proof}

\begin{figure}[hbtp]\centering

\begin{tikzpicture}[scale=0.8]
\tikzstyle{uvertex}=[draw,circle,fill=white!25,minimum size=8pt,inner sep=1pt]
\tikzstyle{cvertex}=[draw,circle,fill=black!10,minimum size=8pt,inner sep=1pt]

\node[]() at (2.17,2.35) {\textbf{s}};
\node[uvertex] (a0) at (0,0) {};
\node[uvertex] (b0) at (0,1) {};
\node[uvertex] (c0) at (0,2) {};
\node[uvertex] (d0) at (0,3) {};
\node[uvertex] (a1) at (1,0) {};
\node[uvertex] (b1) at (1,1) {1};
\node[uvertex] (c1) at (1,2) {};
\node[uvertex] (d1) at (1,3) {};
\node[uvertex] (a2) at (2,0) {};
\node[uvertex] (b2) at (2,1) {3};
\node[cvertex] (c2) at (2,2) {2};
\node[uvertex] (d2) at (2,3) {};
\node[uvertex] (a3) at (3,0) {};
\node[uvertex] (b3) at (3,1) {};
\node[uvertex] (c3) at (3,2) {};
\node[uvertex] (d3) at (3,3) {};

\path[-,very thin]
   (a0) edge (b0) edge (a1) edge (b1)
   (a1) edge (b1) edge (a2) edge (b2)
   (a2) edge (b2) edge (a3) edge (b3)
   (a3) edge (b3)
   (b0) edge (c0) edge (a1) edge (b1) edge (c1)
   (b1) edge (c1) edge (a2) edge (b2) edge (c2)
   (b2) edge (c2) edge (a3) edge (b3) edge (c3)
   (b3) edge (c3)
   (c0) edge (d0) edge (b1) edge (c1) edge (d1)
   (c1) edge (d1) edge (b2) edge (c2) edge (d2)
   (c2) edge (d2) edge (b3) edge (c3) edge (d3)
   (c3) edge (d3)
   (d0) edge (c1) edge (d1)
   (d1) edge (c2) edge (d2)
   (d2) edge (c3) edge (d3);

\node[uvertex] (a0) at (0+8,0) {};
\node[uvertex] (b0) at (0+8,1) {};
\node[uvertex] (c0) at (0+8,2) {1};
\node[cvertex] (d0) at (0+8,3) {s};
\node[uvertex] (a1) at (1+8,0) {};
\node[uvertex] (b1) at (1+8,1) {};
\node[uvertex] (c1) at (1+8,2) {3};
\node[uvertex] (d1) at (1+8,3) {2};
\node[uvertex] (a2) at (2+8,0) {};
\node[uvertex] (b2) at (2+8,1) {};
\node[uvertex] (c2) at (2+8,2) {};
\node[uvertex] (d2) at (2+8,3) {};
\node[uvertex] (a3) at (3+8,0) {};
\node[uvertex] (b3) at (3+8,1) {};
\node[uvertex] (c3) at (3+8,2) {};
\node[uvertex] (d3) at (3+8,3) {};
\path[-,very thin]
   (a0) edge (b0) edge (a1) edge (b1)
   (a1) edge (b1) edge (a2) edge (b2)
   (a2) edge (b2) edge (a3) edge (b3)
   (a3) edge (b3)
   (b0) edge (c0) edge (a1) edge (b1) edge (c1)
   (b1) edge (c1) edge (a2) edge (b2) edge (c2)
   (b2) edge (c2) edge (a3) edge (b3) edge (c3)
   (b3) edge (c3)
   (c0) edge (d0) edge (b1) edge (c1) edge (d1)
   (c1) edge (d1) edge (b2) edge (c2) edge (d2)
   (c2) edge (d2) edge (b3) edge (c3) edge (d3)
   (c3) edge (d3)
   (d0) edge (c1) edge (d1)
   (d1) edge (c2) edge (d2)
   (d2) edge (c3) edge (d3);

\draw[->, very thick] (5,1.5) -- (6,1.5);

\tikzset{rectangle/.append style={draw,color=blue,rounded corners,dashed,minimum height=1.2cm},minimum width=1.2cm,thick}
\node[rectangle] (r) at (0.5,0.5) {};
\node[rectangle] (r) at (2.5,0.5) {};
\node[rectangle] (r) at (0.5,2.5) {};
\node[rectangle] (r) at (2.5,2.5) {};
\node[rectangle] (r) at (0.5+8,0.5) {};
\node[rectangle] (r) at (2.5+8,0.5) {};
\node[rectangle] (r) at (0.5+8,2.5) {};
\node[rectangle] (r) at (2.5+8,2.5) {};
\end{tikzpicture}\\

\vspace{0.5cm}

\begin{tikzpicture}[scale=0.8]
\tikzstyle{uvertex}=[draw,circle,fill=white!25,minimum size=8pt,inner sep=1pt]
\tikzstyle{cvertex}=[draw,circle,fill=black!10,minimum size=8pt,inner sep=1pt]

\node[uvertex] (a0) at (0,0) {};
\node[uvertex] (b0) at (0,1) {};
\node[uvertex] (c0) at (0,2) {};
\node[uvertex] (d0) at (0,3) {};
\node[uvertex] (e0) at (0,4) {};
\node[uvertex] (f0) at (0,5) {};
\node[uvertex] (a1) at (1,0) {};
\node[uvertex] (b1) at (1,1) {};
\node[uvertex] (c1) at (1,2) {};
\node[uvertex] (d1) at (1,3) {};
\node[uvertex] (e1) at (1,4) {};
\node[uvertex] (f1) at (1,5) {};
\node[uvertex] (a2) at (2,0) {};
\node[uvertex] (b2) at (2,1) {};
\node[uvertex] (c2) at (2,2) {1};
\node[uvertex] (d2) at (2,3) {};
\node[uvertex] (e2) at (2,4) {2};
\node[uvertex] (f2) at (2,5) {};
\node[uvertex] (a3) at (3,0) {};
\node[uvertex] (b3) at (3,1) {};
\node[uvertex] (c3) at (3,2) {};
\node[cvertex] (d3) at (3,3) {\textbf{s}};
\node[uvertex] (e3) at (3,4) {};
\node[uvertex] (f3) at (3,5) {};
\node[uvertex] (a4) at (4,0) {};
\node[uvertex] (b4) at (4,1) {};
\node[uvertex] (c4) at (4,2) {3};
\node[uvertex] (d4) at (4,3) {};
\node[uvertex] (e4) at (4,4) {};
\node[uvertex] (f4) at (4,5) {};
\node[uvertex] (a5) at (5,0) {};
\node[uvertex] (b5) at (5,1) {};
\node[uvertex] (c5) at (5,2) {};
\node[uvertex] (d5) at (5,3) {};
\node[uvertex] (e5) at (5,4) {};
\node[uvertex] (f5) at (5,5) {};

\path[-,very thin]
   (a0) edge (b0) edge (a1) edge (b1)
   (a1) edge (b1) edge (a2) edge (b2)
   (a2) edge (b2) edge (a3) edge (b3)
   (a3) edge (b3) edge (a4) edge (b4)
   (a4) edge (b4) edge (a5) edge (b5)
   (a5) edge (b5)
   (b0) edge (c0) edge (a1) edge (b1) edge (c1)
   (b1) edge (c1) edge (a2) edge (b2) edge (c2)
   (b2) edge (c2) edge (a3) edge (b3) edge (c3)
   (b3) edge (c3) edge (a4) edge (b4) edge (c4)
   (b4) edge (c4) edge (a5) edge (b5) edge (c5)
   (b5) edge (c5)
   (c0) edge (d0) edge (b1) edge (c1) edge (d1)
   (c1) edge (d1) edge (b2) edge (c2) edge (d2)
   (c2) edge (d2) edge (b3) edge (c3) edge (d3)
   (c3) edge (d3) edge (b4) edge (c4) edge (d4)
   (c4) edge (d4) edge (b5) edge (c5) edge (d5)
   (c5) edge (d5)
   (d0) edge (e0) edge (c1) edge (d1) edge (e1)
   (d1) edge (e1) edge (c2) edge (d2) edge (e2)
   (d2) edge (e2) edge (c3) edge (d3) edge (e3)
   (d3) edge (e3) edge (c4) edge (d4) edge (e4)
   (d4) edge (e4) edge (c5) edge (d5) edge (e5)
   (d5) edge (e5)
   (e0) edge (f0) edge (d1) edge (e1) edge (f1)
   (e1) edge (f1) edge (d2) edge (e2) edge (f2)
   (e2) edge (f2) edge (d3) edge (e3) edge (f3)
   (e3) edge (f3) edge (d4) edge (e4) edge (f4)
   (e4) edge (f4) edge (d5) edge (e5) edge (f5)
   (e5) edge (f5)
   (f0) edge (e1) edge (f1)
   (f1) edge (e2) edge (f2)
   (f2) edge (e3) edge (f3)
   (f3) edge (e4) edge (f4)
   (f4) edge (e5) edge (f5);

\node[uvertex] (a0) at (0+7,0) {};
\node[uvertex] (b0) at (0+7,1) {};
\node[uvertex] (c0) at (0+7,2) {};
\node[uvertex] (d0) at (0+7,3) {};
\node[uvertex] (e0) at (0+7,4) {};
\node[uvertex] (f0) at (0+7,5) {};
\node[uvertex] (a1) at (1+7,0) {};
\node[uvertex] (b1) at (1+7,1) {};
\node[uvertex] (c1) at (1+7,2) {};
\node[uvertex] (d1) at (1+7,3) {};
\node[uvertex] (e1) at (1+7,4) {};
\node[uvertex] (f1) at (1+7,5) {};
\node[uvertex] (a2) at (2+7,0) {};
\node[uvertex] (b2) at (2+7,1) {};
\node[uvertex] (c2) at (2+7,2) {};
\node[uvertex] (d2) at (2+7,3) {};
\node[uvertex] (e2) at (2+7,4) {};
\node[uvertex] (f2) at (2+7,5) {};
\node[uvertex] (a3) at (3+7,0) {};
\node[uvertex] (b3) at (3+7,1) {};
\node[uvertex] (c3) at (3+7,2) {};
\node[uvertex] (d3) at (3+7,3) {1};
\node[uvertex] (e3) at (3+7,4) {2};
\node[uvertex] (f3) at (3+7,5) {};
\node[uvertex] (a4) at (4+7,0) {};
\node[uvertex] (b4) at (4+7,1) {};
\node[uvertex] (c4) at (4+7,2) {};
\node[uvertex] (d4) at (4+7,3) {3};
\node[uvertex] (e4) at (4+7,4) {};
\node[uvertex] (f4) at (4+7,5) {};
\node[uvertex] (a5) at (5+7,0) {};
\node[uvertex] (b5) at (5+7,1) {};
\node[uvertex] (c5) at (5+7,2) {};
\node[uvertex] (d5) at (5+7,3) {};
\node[uvertex] (e5) at (5+7,4) {};
\node[cvertex] (f5) at (5+7,5) {\textbf{s}};

\path[-,very thin]
   (a0) edge (b0) edge (a1) edge (b1)
   (a1) edge (b1) edge (a2) edge (b2)
   (a2) edge (b2) edge (a3) edge (b3)
   (a3) edge (b3) edge (a4) edge (b4)
   (a4) edge (b4) edge (a5) edge (b5)
   (a5) edge (b5)
   (b0) edge (c0) edge (a1) edge (b1) edge (c1)
   (b1) edge (c1) edge (a2) edge (b2) edge (c2)
   (b2) edge (c2) edge (a3) edge (b3) edge (c3)
   (b3) edge (c3) edge (a4) edge (b4) edge (c4)
   (b4) edge (c4) edge (a5) edge (b5) edge (c5)
   (b5) edge (c5)
   (c0) edge (d0) edge (b1) edge (c1) edge (d1)
   (c1) edge (d1) edge (b2) edge (c2) edge (d2)
   (c2) edge (d2) edge (b3) edge (c3) edge (d3)
   (c3) edge (d3) edge (b4) edge (c4) edge (d4)
   (c4) edge (d4) edge (b5) edge (c5) edge (d5)
   (c5) edge (d5)
   (d0) edge (e0) edge (c1) edge (d1) edge (e1)
   (d1) edge (e1) edge (c2) edge (d2) edge (e2)
   (d2) edge (e2) edge (c3) edge (d3) edge (e3)
   (d3) edge (e3) edge (c4) edge (d4) edge (e4)
   (d4) edge (e4) edge (c5) edge (d5) edge (e5)
   (d5) edge (e5)
   (e0) edge (f0) edge (d1) edge (e1) edge (f1)
   (e1) edge (f1) edge (d2) edge (e2) edge (f2)
   (e2) edge (f2) edge (d3) edge (e3) edge (f3)
   (e3) edge (f3) edge (d4) edge (e4) edge (f4)
   (e4) edge (f4) edge (d5) edge (e5) edge (f5)
   (e5) edge (f5)
   (f0) edge (e1) edge (f1)
   (f1) edge (e2) edge (f2)
   (f2) edge (e3) edge (f3)
   (f3) edge (e4) edge (f4)
   (f4) edge (e5) edge (f5);

\draw[->, very thick] ( 5.5,2.5) -- ( 6.5,2.5);
\tikzset{rectangle/.append style={draw,color=blue,rounded corners,dashed,minimum height=2.1cm},minimum width=2.1cm,thick}
\node[rectangle] (r) at (1,1) {};
\node[rectangle] (r) at (4,1) {};
\node[rectangle] (r) at (1,4) {};
\node[rectangle] (r) at (4,4) {};
\node[rectangle] (r) at (1+7,1) {};
\node[rectangle] (r) at (4+7,1) {};
\node[rectangle] (r) at (1+7,4) {};
\node[rectangle] (r) at (4+7,4) {};
\end{tikzpicture}

\caption{Spy winning strategies with 3 guards in the King grids $P_4\boxtimes P_4$ and $P_6\boxtimes P_6$ with $(s,d)=(2,0)$ and $(s,d)=(3,1)$, respectively. The spy is represented by $s$ and the guards by 1, 2 and 3.}\label{fig-king03}
\end{figure}
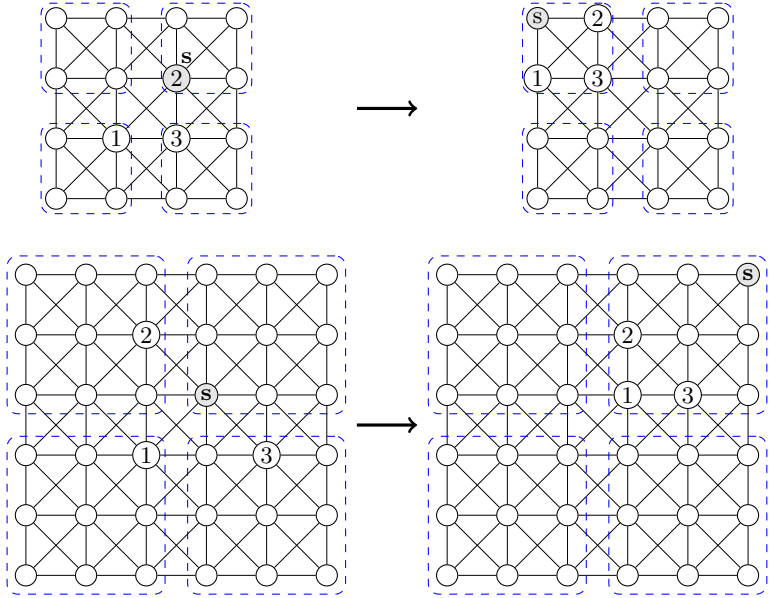

The previous lemmas showed examples with few guards. We show that $\gn_{s,d}(P_n\boxtimes P_n)$ can be very close to $\gn_{s,d}(P_n)^2$, when many guards are necessary.
More specifically, we show that $\gn_{s,d}(P_n\boxtimes P_n)\geq(\gn_{s,d}(P_n)-1)^2$ in many cases where $\gn_{s,d}(P_n)$ can be any positive integer.

\begin{lemma}
Let $d\geq 0$ and $2\leq k\leq 2d+2$ be fixed integers and let $s\geq (k-1)(2d+3)$.
Then $\gn_{s,d}(P_{k(2d+3)})=k+1$ and $k^2\leq \gn_{s,d}(P_{k(2d+3)}\boxtimes P_{k(2d+3)})\leq (k+1)^2$.
\end{lemma}

\begin{proof}
The vertex set of $P_{k(2d+3)}\boxtimes P_{k(2d+3)}$ can be partitioned into $k^2$ subsets, each of which induces a copy of the King grid $P_{2d+3}\boxtimes P_{2d+3}$. If there is always one of these subsets without a guard, the spy can go to the vertex in the center of this subset and no guard can surveil the spy, since the distance from the center to the border is $d+1$ in the King grid $P_{2d+3}\boxtimes P_{2d+3}$. The spy can do this infinitely many times since her speed is at least $(k-1)(2d+3)$ (notice that the diameter is $k(2d+3)-1$ and the maximum distance
between the centers of any two of those King grid copies is at most $k(2d+3)-1-2(d+1)=(k-1)(2d+3)$).

Thus, at least $k^2$ guards are necessary.
Moreover, since
\[
  \gn_{s,d}(P_n) = \left\lceil \frac{n}{2d+2+\left\lfloor \frac{2d}{s-1} \right\rfloor} \right\rceil,
\]
then $\gn_{s,d}(P_{k(2d+3)})=k+1$ for any $d\geq 0$, $2\leq k\leq 2d+2$ and $s\geq(k-1)(2d+3)>2d+1$. Then, from Theorem \ref{teo-upper}, $(k+1)^2$ guards are sufficient.
\qed
\end{proof}

\subsection{Cartesian and lexicographical products}

There are other well studied graph products on the vertex set $V(G_1)\times V(G_2)$, such as the cartesian product $G_1\square G_2$ and the lexicographical products $G_1\cdot G_2$ and  $G_2\cdot G_1$ (see \cite{klavzar11} for a reference).
In the cartesian product $G_1\square G_2$, $(u_1,u_2)$ and $(v_1,v_2)$ are adjacent if and only if (a) $u_1=v_1$ and $u_2v_2\in E(G_2)$ or (b) $u_2=v_2$ and $u_1v_1\in E(G_1)$. In the lexicographical product $G_1\cdot G_2$, $(u_1,u_2)$ and $(v_1,v_2)$ are adjacent if and only if (a) $u_1=v_1$ and $u_2v_2\in E(G_2)$ or (b) $u_1v_1\in E(G_1)$.

It is easy to see, from the definitions, that the cartesian product $G_1\square G_2$ is a subgraph of the strong product $G_1\boxtimes G_2$, which is a subgraph of the lexicographical products $G_1\cdot G_2$ and  $G_2\cdot G_1$.
It is also easy to find examples in which the upper bound of Theorem \ref{teo-upper} fails in the cartesian product. For example, $\gn_{2,0}(P_2\square P_2)=\gn_{2,0}(C_4)=2>1=\gn_{2,0}(P_2)^2$. Regarding the lexicographical product, failing examples are not so easy to find.
The next lemma presents an example in which the upper bound of Theorem \ref{teo-upper} fails in both cartesian product and lexicographical product.

\begin{lemma}
$\gn_{2,1}(P_5\square P_5)>\gn_{2,1}(P_5)^2$ and $\gn_{2,1}(P_5\cdot P_5)>\gn_{2,1}(P_5)^2$.
\end{lemma}

\begin{proof}
The graph products $P_5\square P_5$ and $P_5\cdot P_5$ are shown in Figure \ref{fig-prod01}. Consider that the vertex set of $P_5$ is $\{-2,-1,0,1,2\}^2$. It is easy to see that $\gn_{2,1}(P_5)=1$ (one guard is sufficient): a guard positioned in the vertex 0 can go to -1 and 1 in order to surveil the vertices -2 and 2, and the spy cannot go from -2 to 2 in one time step (and vice-versa). However, in the cartesian product $P_5\square P_5$, a spy positioned in $(0,0)$ can go to one of the ``corner'' vertices $(-2,-2)$, $(-2,2)$, $(2,-2)$ and $(2,2)$ in two time steps, which are impossible to surveil by one guard.

Now consider the lexicographical product $P_5\cdot P_5$. Let us show a spy winning strategy against one guard. Initially, the spy is positioned in the vertex $(-2,0)$. Suppose the guard is in a vertex $(x,0)$ surveilling the spy with $x\in\{-2,-1,0,1,2\}$, since otherwise the spy can go to one of the vertices $(-2,-2)$ or $(-2,2)$ in one time step and the guard cannot surveil her. Without loss of generality, assume that the guard was positioned in the vertex $(-1,0)$. In the first turn, the spy goes to the vertex $(2,0)$. The guard cannot go to $(1,0)$ or $(2,0)$ in one step, which are the only vertices of the form $(x',0)$ from which a guard controls the spy in $(2,0)$ with surveillance distance $d=1$. So, suppose w.l.g. that he goes to the vertex $(0,1)$. Thus, in the second turn, the spy goes to the vertex $(2,-2)$. Since the spy can do this infinitely many times, she wins.
\qed
\end{proof}

\begin{figure}[ht]\centering

\scalebox{0.75}{
\begin{tikzpicture}[scale=0.8]
\tikzstyle{uvertex}=[draw,circle,fill=white!25,minimum size=8pt,inner sep=1pt]
\tikzstyle{cvertex}=[draw,circle,fill=black!10,minimum size=8pt,inner sep=1pt]

\node[uvertex] (a0) at (0-7,0) {};
\node[uvertex] (b0) at (0-7,1) {};
\node[uvertex] (c0) at (0-7,2) {};
\node[uvertex] (d0) at (0-7,3) {};
\node[uvertex] (e0) at (0-7,4) {};
\node[uvertex] (a1) at (1-7,0) {};
\node[uvertex] (b1) at (1-7,1) {};
\node[uvertex] (c1) at (1-7,2) {};
\node[uvertex] (d1) at (1-7,3) {};
\node[uvertex] (e1) at (1-7,4) {};
\node[uvertex] (a2) at (2-7,0) {};
\node[uvertex] (b2) at (2-7,1) {};
\node[cvertex] (c2) at (2-7,2) {\textbf{s}};
\node[uvertex] (d2) at (2-7,3) {1};
\node[uvertex] (e2) at (2-7,4) {};
\node[uvertex] (a3) at (3-7,0) {};
\node[uvertex] (b3) at (3-7,1) {};
\node[uvertex] (c3) at (3-7,2) {};
\node[uvertex] (d3) at (3-7,3) {};
\node[uvertex] (e3) at (3-7,4) {};
\node[uvertex] (a4) at (4-7,0) {};
\node[uvertex] (b4) at (4-7,1) {};
\node[uvertex] (c4) at (4-7,2) {};
\node[uvertex] (d4) at (4-7,3) {};
\node[uvertex] (e4) at (4-7,4) {};
\path[-,very thin]
   (a0) edge (b0) edge (a1)
   (a1) edge (b1) edge (a2)
   (a2) edge (b2) edge (a3)
   (a3) edge (b3) edge (a4)
   (a4) edge (b4)
   (b0) edge (c0) edge (b1)
   (b1) edge (c1) edge (b2)
   (b2) edge (c2) edge (b3)
   (b3) edge (c3) edge (b4)
   (b4) edge (c4)
   (c0) edge (d0) edge (c1)
   (c1) edge (d1) edge (c2)
   (c2) edge (d2) edge (c3)
   (c3) edge (d3) edge (c4)
   (c4) edge (d4)

   (d0) edge (e0) edge (d1)
   (d1) edge (e1) edge (d2)
   (d2) edge (e2) edge (d3)
   (d3) edge (e3) edge (d4)
   (d4) edge (e4)

   (e0) edge (e1)
   (e1) edge (e2)
   (e2) edge (e3)
   (e3) edge (e4);

\node[uvertex] (a0) at (0,0) {};
\node[uvertex] (b0) at (0,1) {};
\node[cvertex] (c0) at (0,2) {\textbf{s}};
\node[uvertex] (d0) at (0,3) {};
\node[uvertex] (e0) at (0,4) {};
\node[uvertex] (a1) at (1,0) {};
\node[uvertex] (b1) at (1,1) {};
\node[uvertex] (c1) at (1,2) {1};
\node[uvertex] (d1) at (1,3) {};
\node[uvertex] (e1) at (1,4) {};
\node[uvertex] (a2) at (2,0) {};
\node[uvertex] (b2) at (2,1) {};
\node[uvertex] (c2) at (2,2) {};
\node[uvertex] (d2) at (2,3) {};
\node[uvertex] (e2) at (2,4) {};
\node[uvertex] (a3) at (3,0) {};
\node[uvertex] (b3) at (3,1) {};
\node[uvertex] (c3) at (3,2) {};
\node[uvertex] (d3) at (3,3) {};
\node[uvertex] (e3) at (3,4) {};
\node[uvertex] (a4) at (4,0) {};
\node[uvertex] (b4) at (4,1) {};
\node[uvertex] (c4) at (4,2) {};
\node[uvertex] (d4) at (4,3) {};
\node[uvertex] (e4) at (4,4) {};
\path[-,very thin]
   (a0) edge (b0) edge (a1) edge (b1)
   (a1) edge (b1) edge (a2) edge (b2)
   (a2) edge (b2) edge (a3) edge (b3)
   (a3) edge (b3) edge (a4) edge (b4)
   (a4) edge (b4)
   (b0) edge (c0) edge (a1) edge (b1) edge (c1)
   (b1) edge (c1) edge (a2) edge (b2) edge (c2)
   (b2) edge (c2) edge (a3) edge (b3) edge (c3)
   (b3) edge (c3) edge (a4) edge (b4) edge (c4)
   (b4) edge (c4)
   (c0) edge (d0) edge (b1) edge (c1) edge (d1)
   (c1) edge (d1) edge (b2) edge (c2) edge (d2)
   (c2) edge (d2) edge (b3) edge (c3) edge (d3)
   (c3) edge (d3) edge (b4) edge (c4) edge (d4)
   (c4) edge (d4)

   (d0) edge (e0) edge (c1) edge (d1) edge (e1)
   (d1) edge (e1) edge (c2) edge (d2) edge (e2)
   (d2) edge (e2) edge (c3) edge (d3) edge (e3)
   (d3) edge (e3) edge (c4) edge (d4) edge (e4)
   (d4) edge (e4)

   (e0) edge (d1) edge (e1)
   (e1) edge (d2) edge (e2)
   (e2) edge (d3) edge (e3)
   (e3) edge (d4) edge (e4)

	(a0) edge (b2) edge (b3) edge (b4)
	(a1) edge (b3) edge (b4)
	(a2) edge (b0) edge (b4)
	(a3) edge (b0) edge (b1)
	(a4) edge (b0) edge (b1) edge (b2)
	(b0) edge (c2) edge (c3) edge (c4)
	(b1) edge (c3) edge (c4)
	(b2) edge (c0) edge (c4)
	(b3) edge (c0) edge (c1)
	(b4) edge (c0) edge (c1) edge (c2)
	(c0) edge (d2) edge (d3) edge (d4)
	(c1) edge (d3) edge (d4)
	(c2) edge (d0) edge (d4)
	(c3) edge (d0) edge (d1)
	(c4) edge (d0) edge (d1) edge (d2)
	(d0) edge (e2) edge (e3) edge (e4)
	(d1) edge (e3) edge (e4)
	(d2) edge (e0) edge (e4)
	(d3) edge (e0) edge (e1)
	(d4) edge (e0) edge (e1) edge (e2);

\node[uvertex] (a0) at (0+5,0) {};
\node[uvertex] (b0) at (0+5,1) {};
\node[uvertex] (c0) at (0+5,2) {};
\node[uvertex] (d0) at (0+5,3) {};
\node[uvertex] (e0) at (0+5,4) {};
\node[uvertex] (a1) at (1+5,0) {};
\node[uvertex] (b1) at (1+5,1) {};
\node[uvertex] (c1) at (1+5,2) {};
\node[uvertex] (d1) at (1+5,3) {};
\node[uvertex] (e1) at (1+5,4) {};
\node[uvertex] (a2) at (2+5,0) {};
\node[uvertex] (b2) at (2+5,1) {};
\node[uvertex] (c2) at (2+5,2) {};
\node[uvertex] (d2) at (2+5,3) {1};
\node[uvertex] (e2) at (2+5,4) {};
\node[uvertex] (a3) at (3+5,0) {};
\node[uvertex] (b3) at (3+5,1) {};
\node[uvertex] (c3) at (3+5,2) {};
\node[uvertex] (d3) at (3+5,3) {};
\node[uvertex] (e3) at (3+5,4) {};
\node[uvertex] (a4) at (4+5,0) {};
\node[uvertex] (b4) at (4+5,1) {};
\node[cvertex] (c4) at (4+5,2) {\textbf{s}};
\node[uvertex] (d4) at (4+5,3) {};
\node[uvertex] (e4) at (4+5,4) {};
\path[-,very thin]
   (a0) edge (b0) edge (a1) edge (b1)
   (a1) edge (b1) edge (a2) edge (b2)
   (a2) edge (b2) edge (a3) edge (b3)
   (a3) edge (b3) edge (a4) edge (b4)
   (a4) edge (b4)
   (b0) edge (c0) edge (a1) edge (b1) edge (c1)
   (b1) edge (c1) edge (a2) edge (b2) edge (c2)
   (b2) edge (c2) edge (a3) edge (b3) edge (c3)
   (b3) edge (c3) edge (a4) edge (b4) edge (c4)
   (b4) edge (c4)
   (c0) edge (d0) edge (b1) edge (c1) edge (d1)
   (c1) edge (d1) edge (b2) edge (c2) edge (d2)
   (c2) edge (d2) edge (b3) edge (c3) edge (d3)
   (c3) edge (d3) edge (b4) edge (c4) edge (d4)
   (c4) edge (d4)

   (d0) edge (e0) edge (c1) edge (d1) edge (e1)
   (d1) edge (e1) edge (c2) edge (d2) edge (e2)
   (d2) edge (e2) edge (c3) edge (d3) edge (e3)
   (d3) edge (e3) edge (c4) edge (d4) edge (e4)
   (d4) edge (e4)

   (e0) edge (d1) edge (e1)
   (e1) edge (d2) edge (e2)
   (e2) edge (d3) edge (e3)
   (e3) edge (d4) edge (e4)

	(a0) edge (b2) edge (b3) edge (b4)
	(a1) edge (b3) edge (b4)
	(a2) edge (b0) edge (b4)
	(a3) edge (b0) edge (b1)
	(a4) edge (b0) edge (b1) edge (b2)
	(b0) edge (c2) edge (c3) edge (c4)
	(b1) edge (c3) edge (c4)
	(b2) edge (c0) edge (c4)
	(b3) edge (c0) edge (c1)
	(b4) edge (c0) edge (c1) edge (c2)
	(c0) edge (d2) edge (d3) edge (d4)
	(c1) edge (d3) edge (d4)
	(c2) edge (d0) edge (d4)
	(c3) edge (d0) edge (d1)
	(c4) edge (d0) edge (d1) edge (d2)
	(d0) edge (e2) edge (e3) edge (e4)
	(d1) edge (e3) edge (e4)
	(d2) edge (e0) edge (e4)
	(d3) edge (e0) edge (e1)
	(d4) edge (e0) edge (e1) edge (e2);

\node[uvertex] (a0) at (0+10,0) {};
\node[uvertex] (b0) at (0+10,1) {};
\node[uvertex] (c0) at (0+10,2) {};
\node[uvertex] (d0) at (0+10,3) {};
\node[uvertex] (e0) at (0+10,4) {};
\node[uvertex] (a1) at (1+10,0) {};
\node[uvertex] (b1) at (1+10,1) {};
\node[uvertex] (c1) at (1+10,2) {};
\node[uvertex] (d1) at (1+10,3) {};
\node[uvertex] (e1) at (1+10,4) {};
\node[uvertex] (a2) at (2+10,0) {};
\node[uvertex] (b2) at (2+10,1) {};
\node[uvertex] (c2) at (2+10,2) {};
\node[uvertex] (d2) at (2+10,3) {};
\node[uvertex] (e2) at (2+10,4) {};
\node[uvertex] (a3) at (3+10,0) {};
\node[uvertex] (b3) at (3+10,1) {};
\node[uvertex] (c3) at (3+10,2) {1};
\node[uvertex] (d3) at (3+10,3) {};
\node[uvertex] (e3) at (3+10,4) {};
\node[cvertex] (a4) at (4+10,0) {\textbf{s}};
\node[uvertex] (b4) at (4+10,1) {};
\node[uvertex] (c4) at (4+10,2) {};
\node[uvertex] (d4) at (4+10,3) {};
\node[uvertex] (e4) at (4+10,4) {};
\path[-,very thin]
   (a0) edge (b0) edge (a1) edge (b1)
   (a1) edge (b1) edge (a2) edge (b2)
   (a2) edge (b2) edge (a3) edge (b3)
   (a3) edge (b3) edge (a4) edge (b4)
   (a4) edge (b4)
   (b0) edge (c0) edge (a1) edge (b1) edge (c1)
   (b1) edge (c1) edge (a2) edge (b2) edge (c2)
   (b2) edge (c2) edge (a3) edge (b3) edge (c3)
   (b3) edge (c3) edge (a4) edge (b4) edge (c4)
   (b4) edge (c4)
   (c0) edge (d0) edge (b1) edge (c1) edge (d1)
   (c1) edge (d1) edge (b2) edge (c2) edge (d2)
   (c2) edge (d2) edge (b3) edge (c3) edge (d3)
   (c3) edge (d3) edge (b4) edge (c4) edge (d4)
   (c4) edge (d4)

   (d0) edge (e0) edge (c1) edge (d1) edge (e1)
   (d1) edge (e1) edge (c2) edge (d2) edge (e2)
   (d2) edge (e2) edge (c3) edge (d3) edge (e3)
   (d3) edge (e3) edge (c4) edge (d4) edge (e4)
   (d4) edge (e4)

   (e0) edge (d1) edge (e1)
   (e1) edge (d2) edge (e2)
   (e2) edge (d3) edge (e3)
   (e3) edge (d4) edge (e4)

	(a0) edge (b2) edge (b3) edge (b4)
	(a1) edge (b3) edge (b4)
	(a2) edge (b0) edge (b4)
	(a3) edge (b0) edge (b1)
	(a4) edge (b0) edge (b1) edge (b2)
	(b0) edge (c2) edge (c3) edge (c4)
	(b1) edge (c3) edge (c4)
	(b2) edge (c0) edge (c4)
	(b3) edge (c0) edge (c1)
	(b4) edge (c0) edge (c1) edge (c2)
	(c0) edge (d2) edge (d3) edge (d4)
	(c1) edge (d3) edge (d4)
	(c2) edge (d0) edge (d4)
	(c3) edge (d0) edge (d1)
	(c4) edge (d0) edge (d1) edge (d2)
	(d0) edge (e2) edge (e3) edge (e4)
	(d1) edge (e3) edge (e4)
	(d2) edge (e0) edge (e4)
	(d3) edge (e0) edge (e1)
	(d4) edge (e0) edge (e1) edge (e2);

\draw[->, very thick] ( 4+.2,2-0.3) -- ( 5-.2,2-0.3);
\draw[->, very thick] ( 9+.2,2-0.3) -- (10-.2,2-0.3);
\draw[->, very thick] ( 4+.2,2+0.3) -- ( 5-.2,2+0.3);
\draw[->, very thick] ( 9+.2,2+0.3) -- (10-.2,2+0.3);
\end{tikzpicture}}

\caption{The cartesian product $P_5\square P_5$ and the lexicographical product $P_5\cdot P_5$ with a spy winning strategy with 1 guard and $(s,d)=(2,1)$. The spy is represented by $s$ and the guard by 1.}\label{fig-prod01}
\end{figure}
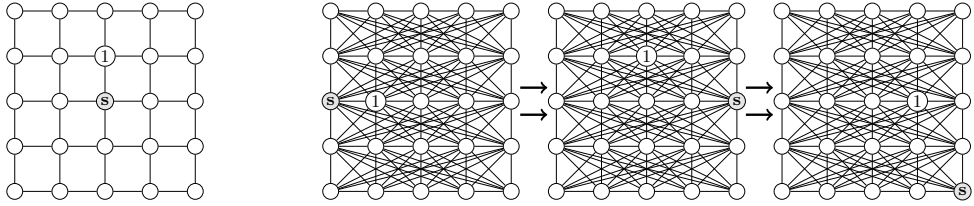

Nevertheless, regarding the lexicographical product, it is possible to prove a better general upper bound if the surveilling distance $d\geq 2$. In fact, we obtain the following stronger result.

\begin{theorem}\label{teo-upper2}
Let $s\geq 1$, $d\geq 2$ and let $G_1$ and $G_2$ be two graphs.
If $G_1$ has no isolated vertex, then $\gn_{s,d}(G_1\cdot G_2) = \gn_{s,d}(G_1)$.
Otherwise, $\gn_{s,d}(G_1\cdot G_2) = \gn_{s,d}(G_1) + (\gn_{s,d}(G_2)-1)\cdot iso(G_1)$, where $iso(G_1)$ is the number of isolated vertices in $G_1$.
\end{theorem}

\begin{proof}
First assume that $G_1$ has no isolated vertex and then every vertex of $G_1$ has a neighbor.
Notice that every guard positioned in a vertex $(u_1,u_2)$ of $G_1\cdot G_2$ is surveilling any vertex $(u_1,v_2)$ with $v_2\in V(G_2)$, since the surveilling distance is at least 2 and $u_1$ has a neighbor $v_1$, which leads to the existence of the path $(u_1,u_2)-(v_1,v_2)-(u_1,v_2)$ with length 2.
Thus, the guards can follow their winning strategy in $G_1$ with $\gn_{s,d}(G_1)$ guards in order to win the spy game in the lexicographical product $G_1\cdot G_2$, by moving only over the vertices in $\{(v_1,b)\ :\ v_1\in V(G_1)\}$, where $b$ is any fixed arbitrary vertex of $G_2$.

If $G_1$ has $iso(G_1)\geq 1$ isolated vertices, then $G_1\cdot G_2$ contains $iso(G_1)\geq 1$ components isomorphic to $G_2$, which need at least $\gn_{s,d}(G_2)\cdot iso(G_1)$ guards to surveil the spy.
Considering $G_1$ without its isolated vertices, $\gn_{s,d}(G_1)-iso(G_1)$ guards are necessary, since each isolated vertex must have a guard because the guards are placed first, and we are done.
\qed
\end{proof}

\section{Hardness in bipartite graphs with bounded diameter}

We prove in this section that deciding if the guards have a winning strategy in the spy game variant when they are placed first with speed $s$ and surveillance distance $d$ is NP-hard for any $s\geq 2$ and $d\geq 0$ even in bipartite graphs with bounded diameter. We also extend this to a W[2]-hardness proof when parameterized by the number $k$ of guards.
We obtain a reduction from the Set Cover problem  similar to the reduction of  \cite{cohen18} for general graphs. Nevertheless, the extension to bipartite graphs is more complex and the constructed bipartite graph is significantly different from the one of \cite{cohen18}.

\begin{theorem}
Let $s\geq 2$, $d\geq 0$ and $k\geq 1$ be fixed.
The problem of deciding if $k$ guards have a winning strategy in the $(s,d)$-spy game variant in which they are placed first is NP-hard and W[2]-hard (parameterized by the number $k$ of guards) even in bipartite graphs with diameter at most $4(d+1)$. Morever, the corresponding optimization problem of minimizing the number of guards is Log-APX-hard and $(1-\varepsilon)\ln n$-inapproximable in polynomial time for any constant $0<\varepsilon<1$, unless P=NP.
\end{theorem}

The proof of the theorem obtains a reduction from the \textsc{Set Cover} Problem and is divided in three lemmas regarding the relation between $s$ and $d$, where $r=r(s,d)=d\mod(s-1)$ is the remainder of the division of $d$ by $s-1$:
\begin{itemize}
\item {\bf Case 1:} $s<2(r+1)$
\item {\bf Case 2:} $s>2(r+1)$
\item {\bf Case 3:} $s=2(r+1)$
\end{itemize}

An instance of the \textsc{Set Cover} Problem is a family $\mathcal{S}=\{S_1,\ldots,S_m\}$ of sets and an integer $c$, and the objective is to decide if there exists a subfamily $\mathcal{C}=\{S_{i_1},\ldots,S_{i_c}\}\subseteq\mathcal{S}$ such that $|\mathcal{C}|\leq c$ and $S_{i_1}\cup\ldots\cup S_{i_c} = U$, where $U=S_1\cup\ldots\cup S_m$ (we say that $\mathcal{C}$ is a set cover of $U$). Given an instance $(\mathcal{S},c)$ of Set Cover, we construct a graph $G=G_{s,d}(\mathcal{S},c)$ and an integer $K=K_{s,d}(\mathcal{S},c)$ such that there exists a cover $\mathcal{C}\subseteq \mathcal{S}$ of $U$ with size at most $c$ if and only if $\gn_{s,d}(G)\leq K$. Note that the reductions presented below are actually fpt-reductions and preserve approximation ratio. Therefore, since the \textsc{Set Cover} Problem is W[2]-hard (when parameterized by the size $c$ of the set cover) and has no $(1-\varepsilon)\ln(n)$ approximation polynomial algorithm for any constant $0<\varepsilon<1$ (unless P=NP) \cite{AlonMS06,dana15}, we not only prove the NP-hardness but also the fact that the problem is W[2]-hard (when parameterized by the number of guards) and cannot be approximated in polynomial time up to some logarithmic ratio (unless $P=NP$).

\begin{definition}\label{def_pq}
Given $s\geq 2$ and $d\geq 0$, let $p=p(s,d)=d+\left\lceil\frac{d+1}{s-1}\right\rceil$ and
\[
q(s,d)=
 \begin{cases}
    0,   &\mbox{if Case 1}\\
    p,   &\mbox{if Case 2},\\
    p-1, &\mbox{if Case 3}
 \end{cases}
\]

Let $(\mathcal{S},c)$ an instance of Set Cover, where $\mathcal{S}=\{S_1,\ldots,S_m\}$, and let $U=S_1\cup\ldots\cup S_m=\{u_1,\ldots,u_n\}$. Let the number $K=K_{s,d}(\mathcal{S},c)$ of guards be
\[
K_{s,d}(\mathcal{S},c)=
 \begin{cases}
    c,   &\mbox{if Case 1},\\
    c+2, &\mbox{if Case 2 or 3}
 \end{cases}
\]

Let $G=G_{s,d}(\mathcal{S},c)$ be defined as follows: for every set $S_j\in\mathcal{S}$, create a new vertex $S_j$ in $G$ and, for every element $u_i\in U$, create a path $U_i$ with $p$ vertices $u_{i,1},\ldots,u_{i,p}$. If $u_i\in S_j$, add the edge $u_{i,1}S_j$ in $G$. Create a new vertex $z_0$ and add all possible edges between $z_0$ and $\{S_1,\ldots,S_m\}$ in $G$. Finally, if it is not Case 1, create two paths $Z=(z_1,\ldots,z_q)$ and $Z'=(z'_1,\ldots,z'_{q'})$ and add the edges $z_0z_1$ and $z_0z'_1$, where $q'=p$ in Case 2 and $q'=p+1$ in Case 3.
\end{definition}

See Figures \ref{fig-reduction}-\ref{fig-reduction3b} for examples.
In the figures, the main position of the spy is indicated with an asterisk (*) and the main positions of the guards are indicated with a big dot ($\bullet$).
Let $Z$ and $Z'$ be the set of vertices of $G$ of the form $z_t$ and $z'_t$, respectively.
Also recall that, given $1\leq i\leq n$, $U_i=\{u_{i,1},\ldots,u_{i,p}\}$.
Notice that the diameter of $G$ is at most $2(p+1)\leq 2(2d+1+1)=4(d+1)$.

The next proposition from Lemma 7 of \cite{cohen18} will be very useful in the proofs.

\begin{proposition}[\cite{cohen18}]\label{lema_pq}
Let $d\geq 1$ and  $2\leq s\leq d+1$.
Let $p=d+\lceil\frac{d+1}{s-1}\rceil$ and $r=d\mod (s-1)$.
Let $P=(u_{-1},u_0,u_1,\cdots,u_p)$ be a path graph and consider one guard playing the game on $P$ against a spy, which is the first to move, with speed $s$ and surveillance distance $d$.
\begin{itemize}
\item[(a)] if the guard starts in $u_{-1}$, the spy can evade surveillance;
\item[(b)] if the guard starts in $u_0$, the spy can evade surveillance if and only if she starts in $u_i$ with $i\geq r+1$.
\end{itemize}
\end{proposition}

In the following lemmas, we use the notations $U_{i,>j}=\{u_{i,t}\in U_i: t>j\}$ and $U_{i,\leq j}=\{u_{i,t}\in U_i: t\leq j\}$, given positive integers $i$ and $j$. Define analogously $Z_{\leq j}$, $Z_{>j}$, $Z'_{\leq j}$ and $Z'_{>j}$.

\begin{lemma}\label{lem-npc1}
Given a graph $G$ and an integer $K>0$, deciding if $\gn_{s,d}(G)\leq K$ is NP-hard for every $s\geq 2$ and $d\geq 1$ satisfying Case 1 ($s<2r+2)$, where $r=d\mod(s-1)$).
\end{lemma}

\begin{proof}
Reduction from \textsc{Set Cover}. Let $(\mathcal{S},c)$ be an instance of Set Cover. Recall Definition \ref{def_pq} and let $p=p(s,d)=d+\left\lceil\frac{d+1}{s-1}\right\rceil$, $q=q(s,d)=0$, $G=G_{s,d}(\mathcal{S},c)$ and $K=K_{s,d}(\mathcal{S},c)=c$.

First, suppose that there is no cover $\mathcal{C}$ of $U$ with at most $c$ sets in $\mathcal{S}$. We prove that the spy wins against at most $c$ guards, starting in $z_0$. Indeed, since there are at most $c$ guards and there is no cover of $U$ with $c$ sets in $\mathcal{S}$, then there exists some $1\leq i\leq n$ such that there is no guard in $N[U_i]$. Thus, both in case $s\leq d+1$ from Proposition \ref{lema_pq}(a) and in case  $s>d+1$, the spy can go to $u_{i,p}$, winning the game in the path $S_{c(i)}-u_{i,1}-\ldots-u_{i,p}$.

Now, suppose that there is a cover $\mathcal{C}=\{S_{j_1},\ldots,S_{j_c}\}$ of $U$ with $c$ sets in $\mathcal{S}$.
We prove that $c$ guards win if they are placed first.
The strategy of the guards is as follows. Occupy initially the vertices $S_{j_1},\ldots,S_{j_c}$. Since $\mathcal{C}$ is a cover of $U$, we can define for any element $u_i\in U$ an index $c(i)$ such that $u_i\in S_{c(i)}\in\mathcal{C}$.
If the spy is in a vertex of $U_{i,\leq r}$ for some $i$, then the guards occupy the initial vertices, controlling the spy both in case $s\leq d+1$ as per Proposition \ref{lema_pq}(b) and in case $s>d+1$. If the spy moved to a vertex of $U_{i,>r}$ from a vertex not in $U_{i,>r}$, then the guard in $S_{c(i)}$ goes to $u_{i,1}$, controlling the spy in the path $u_{i,1}-\ldots-u_{i,p}$ both in case $s\leq d+1$ from Proposition \ref{lema_pq}(b) and in case $s>d+1$. Moreover, since $s<2(r+1)$, the spy cannot go in one step from a vertex of $U_{i,>r}$ to a vertex of $U_{j,>r}$ with $i\ne j$. Then, if the spy leaves $U_{i,>r}$, then the guards reoccupy the initial vertices: the guard in $u_{i,1}$ goes to $S_{c(i)}$. With this strategy, the guards win the game.
\qed
\end{proof}

\input{figs/case1.tex}

\begin{lemma}\label{lem-npc2}
Given a graph $G$ and an integer $K>0$, deciding if $\gn_{s,d}(G)\leq K$ is NP-hard for every $s\geq 2$ and $d\geq 1$ satisfying Case 2 ($s>2r+2$, where $r=d\mod(s-1)$).
\end{lemma}

\begin{proof}
Reduction from \textsc{Set Cover}. Let $(\mathcal{S},c)$ be an instance of Set Cover. Recall Definition \ref{def_pq} and let $p=p(s,d)=d+\lceil\frac{d+1}{s-1}\rceil$, $q=q(s,d)=p$, $G=G_{s,d}(\mathcal{S},c)$ and $K=K_{s,d}(\mathcal{S},c)=c+2$.

First, suppose that there is no cover $\mathcal{C}$ of $U$ with at most $c$ sets in $\mathcal{S}$. We prove that the spy wins against at most $c+2$ guards, starting in $z_{r+1}$. 
Both in case $s\leq d+1$ from Proposition \ref{lema_pq}(a) in the path $z_0-z_q$ and in case $s>d+1$, we may assume w.l.g. that there is a guard in $z_1$. Moreover, since the spy can go to $z'_{r+1}$ in one step, we may assume w.l.g. both in case $s\geq d+1$ from Proposition \ref{lema_pq}(a) in the path $z_0-z'_q$ and in case $s>d+1$ that there is a guard in $z_0$.
Since there are $c+2$ guards, then there is at most $c$ guards outside $\{z_0,z_1,z'_1,\ldots,z_q,z'_q\}$. Since there is no cover of $U$ with $c$ sets in $\mathcal{S}$, then there exists some $1\leq i\leq n$ such that there is no guard in $N[U_i]$. Thus, in one step, the spy can go to $u_{i,r+1}$ (recall that $s\geq 2r+3$), winning the game in the path $z_0-S_{c(i)}-u_{i,1}-\ldots-u_{i,p}$ both in case $s\leq d+1$ from Proposition \ref{lema_pq}(b) and in case $s>d+1$.

Now, suppose that there is a cover $\mathcal{C}=\{S_{j_1},\ldots,S_{j_c}\}$ of $U$ with $c$ sets in $\mathcal{S}$. We prove that $c+2$ guards win if they are placed first. The strategy of the guards is as follows. Occupy initially the vertices $z_0,z_0,S_{j_1},\ldots,S_{j_c}$ (two guards in $z_0$). Since $\mathcal{C}$ is a cover of $U$, we can define for any element $u_i\in U$ an index $c(i)$ such that $u_i\in S_{c(i)}\in\mathcal{C}$.
If the spy is not in a vertex of $U_{i,>r}$, $Z_{>r}$ or $Z'_{>r}$, then the guards occupy the initial vertices, controlling the spy both in case $s\leq d+1$ as per Proposition \ref{lema_pq}(b) and in case $s>d+1$. If the spy moved to a vertex of $U_{i,>r}$ from a vertex not in $U_{i,>r}$, then the guard in $S_{c(i)}$ goes to $u_{i,1}$ (controlling the spy in the path $u_{i,1}-\ldots-u_{i,p}$ both in case $s\leq d+1$ from Proposition \ref{lema_pq}(b)) and in case $s>d+1$, and one guard occupying $z_0$ goes to $S_{c(i)}$.
If the spy goes from a vertex of $U_{i,>r}$ to a vertex of $U_{j,>r}$ with $c(i)=c(j)$, then the guard in $S_{c(i)}$ goes to $u_{j,1}$ (controlling the spy) and the guard in $u_{i,1}$ goes to $S_{c(i)}$.
If the spy goes from a vertex of $U_{i,>r}$ to a vertex of $U_{j,>r}$ with $c(i)\ne c(j)$, then the guard in $S_{c(j)}$ goes to $u_{j,1}$ (controlling the spy), the guard in $z_0$ goes to $S_{c(j)}$, the guard in $S_{c(i)}$ goes to $z_0$ and the guard in $u_{i,1}$ goes to $S_{c(i)}$. 
If the spy goes to $Z_{>r}$ from a vertex not in $Z_{>r}$, then the spy in $z_0$ goes to $z_1$ (controlling the spy), the spy in $S_{c(i)}$ goes to $z_0$ and the spy in $u_{i,1}$ goes to $S_{c(i)}$. Analogously for $Z'$.
If the spy goes from $U_{i,>r}$ to $U_{j,\leq r}$ or $Z_{<r}$ for some $j$, then the guards occupy the initial configuration: the guard in $u_{i,1}$ goes to $S_{c(i)}$ and the guard in $S_{c(i)}$ goes to $z_0$ (two guards in $z_0$), controlling the spy both in case $s\leq d+1$ as per Proposition \ref{lema_pq}(a) and in case $s>d+1$.
With this strategy, the guards win the game.
\qed
\end{proof}

\input{figs/case2.tex}

\begin{lemma}\label{lem-npc3}
Given a graph $G$ and an integer $K$, deciding if $\gn_{s,d}(G)\leq K$ is NP-hard for every $s\geq 2$ and $d\geq 1$ satisfying Case 3 ($s=2r+2$, where $r=d\mod(s-1)$).
\end{lemma}

\begin{proof}
Reduction from \textsc{Set Cover}. Let $(\mathcal{S},c)$ be an instance of Set Cover. Recall Definition \ref{def_pq} and let $p=p(s,d)=d+\lceil\frac{d+1}{s-1}\rceil$, $q=q(s,d)=d+\lceil\frac{d}{s-1}\rceil$, $G=G_{s,d}(\mathcal{S},c)$ and $K=K_{s,d}(\mathcal{S},c)=c+2$.

First, suppose that there is no cover $\mathcal{C}$ of $U$ with at most $c$ sets in $\mathcal{S}$. We prove that the spy wins in one step against at most $c+2$ guards, starting in $z_r$. Indeed, both in case $s\leq d+1$ from Proposition \ref{lema_pq}(b) on the path $z'_2-z'_1-z_0-z_1-\ldots-z_{p-1}$ and in case $s>d+1$, we may assume w.l.g. that there is one guard in $z_0$. Moreover, since the spy can go in one step to $z'_{r+2}$ (recall that $s=2r+2$), we may assume w.l.g. that there is another guard in $z'_1$ (otherwise the spy wins both in case $s\leq d+1$ from Proposition \ref{lema_pq}(b) and in case $s>d+1$).
Since there are $c+2$ guards, then there is at most $c$ guards outside $\{z_0,z_1,z'_1,\ldots\}$. Since there is no cover of $U$ with $c$ sets in $\mathcal{S}$, then there exists some $1\leq i\leq n$ such that there is no guard in $N[U_i]$. Thus, the spy goes from $z_r$ to $u_{i,r+1}$ in one step (note that the distance is $r+(r+1)+1=2r+2=s$) and wins both in case $s\leq d+1$ from Proposition \ref{lema_pq}(a) in the path $z_0-S_{c(i)}-u_{i,1}-\ldots-u_{i,p}$ and in case $s>d+1$.

Now, suppose that there is a cover $\mathcal{C}=\{S_{j_1},\ldots,S_{j_c}\}$ of $U$ with $c$ sets in $\mathcal{S}$. We prove that $c+2$ guards win if they are placed first. The strategy of the guards is as follows. Occupy initially the vertices $z_0,z'_1,S_{j_1},\ldots,S_{j_c}$. Since $\mathcal{C}$ is a cover of $U$, we can define for any element $u_i\in U$ an index $c(i)$ such that $u_i\in S_{c(i)}\in\mathcal{C}$.
If the spy is in a vertex of $U_{i,\leq r}$ for some $i$, $Z_{\leq r}$ or $Z'_{\leq r+1}$, then the guards occupy the initial vertices, controlling the spy both in case $s\leq d+1$ as per Proposition \ref{lema_pq}(b) and in case $s>d+1$.
If the spy moved to a vertex of $U_{i,>r}$ from a vertex not in $U_{i,>r}$, then the guard in $S_{c(i)}$ goes to $u_{i,1}$ (controlling the spy in the path $u_{i,1}-\ldots-u_{i,p}$ both in case $s\leq d+1$ from Proposition \ref{lema_pq}(b) and in case $s>d+1$), the guard in $z_0$ goes to $S_{c(i)}$ and the guard in $z'_1$ goes to $z_0$.
If the spy goes from $u_{i,r+1}$ to $u_{j,r+1}$ with $c(i)=c(j)$ (recall that $s=2r+2$), then the guard in $S_{c(i)}$ goes to $u_{j,1}$ (controlling the spy) and the guard in $u_{i,1}$ goes to $S_{c(i)}$.
If the spy goes from $u_{i,r+1}$ to $u_{j,r+1}$ with $c(i)\ne c(j)$, then the guard in $S_{c(j)}$ goes to $u_{j,1}$ (controlling the spy), the guard in $z_0$ goes to $S_{c(j)}$, the guard in $S_{c(i)}$ goes to $z_0$ and the guard in $u_{i,1}$ goes to $S_{c(i)}$. 
Also notice that, if the spy is in a vertex of $U_{i,>r}$, the guard in $z_0$ controls the paths $z_0-\ldots-z_r$ and $z'_1-\ldots-z'_{r+2}$ both in case $s\leq d+1$ from Proposition \ref{lema_pq}(b) and in case $s>d+1$, since $s=2r+2$ and the spy cannot go from $u_{i,r+1}$ to $z'_{r+2}$.
If the spy goes from a vertex of $U_{i,>r}$ to a vertex of $U_{j,\leq r}$, $Z_{\leq r}$ or $Z'_{\leq r+1}$, then the guards occupy the initial configuration: the guard in $z_0$ goes to $z'_1$, the guard in $S_{c(i)}$ goes to $z_0$ and the guard in $u_{i,1}$ goes to $S_{c(i)}$.
With this strategy, the guards win the game.
\qed
\end{proof}

\input{figs/case3.tex}

\bibliographystyle{splncs04}

\input{spygame-bip.bbl}
\end{document}

%% file: figs/case1.tex
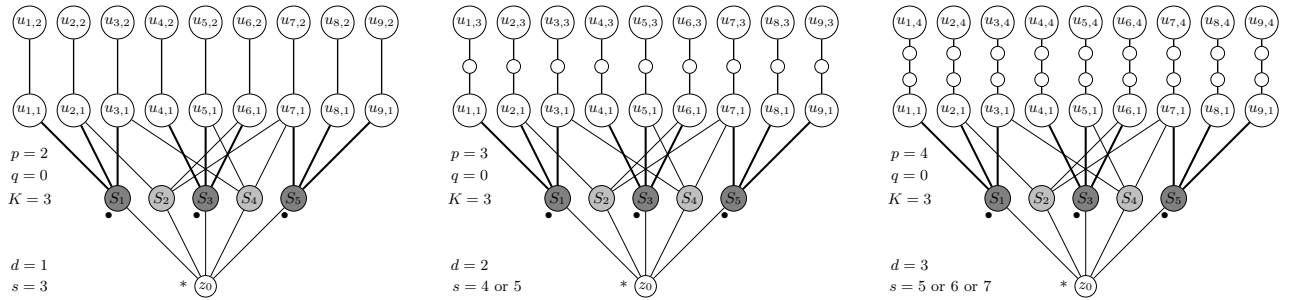
\begin{figure}[ht]\centering\scalebox{0.65}{

\begin{tikzpicture}[scale=0.9]
\tikzstyle{vertex}=[draw,circle,fill=white!25,minimum size=8pt,inner sep=1pt]
\tikzstyle{cvertex}=[draw,circle,fill=black!50,minimum size=10pt,inner sep=1pt]
\tikzstyle{xvertex}=[draw,circle,fill=gray!50,minimum size=10pt,inner sep=1pt]
\tikzstyle{zvertex}=[draw,circle,fill=white!50,minimum size=12pt,inner sep=1pt]

\node[]() at (1.0,0) {$s=3$};
\node[]() at (1,0.5) {$d=1$};
\node[]() at (1,3.0) {$p=2$};
\node[]() at (1,2.5) {$q=0$};
\node[]() at (1,2.0) {$K=3$};

\node[zvertex] (z0) at (5,0) {$z_0$};
\node at (4.5,0) {*};

\node[cvertex] (x1) at (3,2) {$S_1$}; \node at (2.8,1.6) {$\bullet$};
\node[xvertex] (x2) at (4,2) {$S_2$};
\node[cvertex] (x3) at (5,2) {$S_3$}; \node at (4.8,1.6) {$\bullet$};
\node[xvertex] (x4) at (6,2) {$S_4$};
\node[cvertex] (x5) at (7,2) {$S_5$}; \node at (6.8,1.6) {$\bullet$};

\node[vertex] (y11) at (1,4) {$u_{1,1}$};
\node[vertex] (y21) at (2,4) {$u_{2,1}$};
\node[vertex] (y31) at (3,4) {$u_{3,1}$};
\node[vertex] (y41) at (4,4) {$u_{4,1}$};
\node[vertex] (y51) at (5,4) {$u_{5,1}$};
\node[vertex] (y61) at (6,4) {$u_{6,1}$};
\node[vertex] (y71) at (7,4) {$u_{7,1}$};
\node[vertex] (y81) at (8,4) {$u_{8,1}$};
\node[vertex] (y91) at (9,4) {$u_{9,1}$};

\node[vertex] (y12) at (1,6) {$u_{1,2}$};
\node[vertex] (y22) at (2,6) {$u_{2,2}$};
\node[vertex] (y32) at (3,6) {$u_{3,2}$};
\node[vertex] (y42) at (4,6) {$u_{4,2}$};
\node[vertex] (y52) at (5,6) {$u_{5,2}$};
\node[vertex] (y62) at (6,6) {$u_{6,2}$};
\node[vertex] (y72) at (7,6) {$u_{7,2}$};
\node[vertex] (y82) at (8,6) {$u_{8,2}$};
\node[vertex] (y92) at (9,6) {$u_{9,2}$};

\path[-,thin]
(z0) edge (x1) edge (x2) edge (x3) edge (x4) edge (x5);

\path[-,very thick]
(x1) edge (y11) edge (y21) edge (y31)
(x3) edge (y41) edge (y51) edge (y61)
(x5) edge (y71) edge (y81) edge (y91);

\path[-,very thin]
(x2) edge (y21) edge (y61) edge (y71)
(x4) edge (y31) edge (y51) edge (y71);

\path[-,thick]
(y11) edge (y12)
(y21) edge (y22)
(y31) edge (y32)
(y41) edge (y42)
(y51) edge (y52)
(y61) edge (y62)
(y71) edge (y72)
(y81) edge (y82)
(y91) edge (y92);

\node[]() at (11.4,0) {$s=4$ or $5$};
\node[]() at (11,0.5) {$d=2$};
\node[]() at (11,3.0) {$p=3$};
\node[]() at (11,2.5) {$q=0$};
\node[]() at (11,2.0) {$K=3$};

\node[zvertex] (z0) at (15,0) {$z_0$};
\node at (14.5,0) {*};

\node[cvertex] (x1) at (13,2) {$S_1$}; \node at (12.8,1.6) {$\bullet$};
\node[xvertex] (x2) at (14,2) {$S_2$};
\node[cvertex] (x3) at (15,2) {$S_3$}; \node at (14.8,1.6) {$\bullet$};
\node[xvertex] (x4) at (16,2) {$S_4$};
\node[cvertex] (x5) at (17,2) {$S_5$}; \node at (16.8,1.6) {$\bullet$};

\node[vertex] (y11) at (11,4) {$u_{1,1}$};
\node[vertex] (y21) at (12,4) {$u_{2,1}$};
\node[vertex] (y31) at (13,4) {$u_{3,1}$};
\node[vertex] (y41) at (14,4) {$u_{4,1}$};
\node[vertex] (y51) at (15,4) {$u_{5,1}$};
\node[vertex] (y61) at (16,4) {$u_{6,1}$};
\node[vertex] (y71) at (17,4) {$u_{7,1}$};
\node[vertex] (y81) at (18,4) {$u_{8,1}$};
\node[vertex] (y91) at (19,4) {$u_{9,1}$};

\node[vertex] (y12) at (11,5) {};
\node[vertex] (y22) at (12,5) {};
\node[vertex] (y32) at (13,5) {};
\node[vertex] (y42) at (14,5) {};
\node[vertex] (y52) at (15,5) {};
\node[vertex] (y62) at (16,5) {};
\node[vertex] (y72) at (17,5) {};
\node[vertex] (y82) at (18,5) {};
\node[vertex] (y92) at (19,5) {};

\node[vertex] (y13) at (11,6) {$u_{1,3}$};
\node[vertex] (y23) at (12,6) {$u_{2,3}$};
\node[vertex] (y33) at (13,6) {$u_{3,3}$};
\node[vertex] (y43) at (14,6) {$u_{4,3}$};
\node[vertex] (y53) at (15,6) {$u_{5,3}$};
\node[vertex] (y63) at (16,6) {$u_{6,3}$};
\node[vertex] (y73) at (17,6) {$u_{7,3}$};
\node[vertex] (y83) at (18,6) {$u_{8,3}$};
\node[vertex] (y93) at (19,6) {$u_{9,3}$};

\path[-,thin]
(z0) edge (x1) edge (x2) edge (x3) edge (x4) edge (x5);

\path[-,very thick]
(x1) edge (y11) edge (y21) edge (y31)
(x3) edge (y41) edge (y51) edge (y61)
(x5) edge (y71) edge (y81) edge (y91);

\path[-,very thin]
(x2) edge (y21) edge (y61) edge (y71)
(x4) edge (y31) edge (y51) edge (y71);

\path[-,thick]
(y11) edge (y12) (y12) edge (y13)
(y21) edge (y22) (y22) edge (y23)
(y31) edge (y32) (y32) edge (y33)
(y41) edge (y42) (y42) edge (y43)
(y51) edge (y52) (y52) edge (y53)
(y61) edge (y62) (y62) edge (y63)
(y71) edge (y72) (y72) edge (y73)
(y81) edge (y82) (y82) edge (y83)
(y91) edge (y92) (y92) edge (y93);

\node[]() at (21.7,0) {$s=5$ or $6$ or $7$};
\node[]() at (21,0.5) {$d=3$};
\node[]() at (21,3.0) {$p=4$};
\node[]() at (21,2.5) {$q=0$};
\node[]() at (21,2.0) {$K=3$};

\node[zvertex] (zz0) at (25,0) {$z_0$};
\node at (24.5,0) {*};

\node[cvertex] (xx1) at (23,2) {$S_1$}; \node at (22.8,1.6) {$\bullet$};
\node[xvertex] (xx2) at (24,2) {$S_2$};
\node[cvertex] (xx3) at (25,2) {$S_3$}; \node at (24.8,1.6) {$\bullet$};
\node[xvertex] (xx4) at (26,2) {$S_4$};
\node[cvertex] (xx5) at (27,2) {$S_5$}; \node at (26.8,1.6) {$\bullet$};

\node[vertex] (yy11) at (21,4) {$u_{1,1}$};
\node[vertex] (yy21) at (22,4) {$u_{2,1}$};
\node[vertex] (yy31) at (23,4) {$u_{3,1}$};
\node[vertex] (yy41) at (24,4) {$u_{4,1}$};
\node[vertex] (yy51) at (25,4) {$u_{5,1}$};
\node[vertex] (yy61) at (26,4) {$u_{6,1}$};
\node[vertex] (yy71) at (27,4) {$u_{7,1}$};
\node[vertex] (yy81) at (28,4) {$u_{8,1}$};
\node[vertex] (yy91) at (29,4) {$u_{9,1}$};

\node[vertex] (yy12) at (21,4.7) {};
\node[vertex] (yy22) at (22,4.7) {};
\node[vertex] (yy32) at (23,4.7) {};
\node[vertex] (yy42) at (24,4.7) {};
\node[vertex] (yy52) at (25,4.7) {};
\node[vertex] (yy62) at (26,4.7) {};
\node[vertex] (yy72) at (27,4.7) {};
\node[vertex] (yy82) at (28,4.7) {};
\node[vertex] (yy92) at (29,4.7) {};

\node[vertex] (yy13) at (21,5.3) {};
\node[vertex] (yy23) at (22,5.3) {};
\node[vertex] (yy33) at (23,5.3) {};
\node[vertex] (yy43) at (24,5.3) {};
\node[vertex] (yy53) at (25,5.3) {};
\node[vertex] (yy63) at (26,5.3) {};
\node[vertex] (yy73) at (27,5.3) {};
\node[vertex] (yy83) at (28,5.3) {};
\node[vertex] (yy93) at (29,5.3) {};

\node[vertex] (yy14) at (21,6) {$u_{1,4}$};
\node[vertex] (yy24) at (22,6) {$u_{2,4}$};
\node[vertex] (yy34) at (23,6) {$u_{3,4}$};
\node[vertex] (yy44) at (24,6) {$u_{4,4}$};
\node[vertex] (yy54) at (25,6) {$u_{5,4}$};
\node[vertex] (yy64) at (26,6) {$u_{6,4}$};
\node[vertex] (yy74) at (27,6) {$u_{7,4}$};
\node[vertex] (yy84) at (28,6) {$u_{8,4}$};
\node[vertex] (yy94) at (29,6) {$u_{9,4}$};

\path[-,thin]
(zz0) edge (xx1) edge (xx2) edge (xx3) edge (xx4) edge (xx5);

\path[-,very thick]
(xx1) edge (yy11) edge (yy21) edge (yy31)
(xx3) edge (yy41) edge (yy51) edge (yy61)
(xx5) edge (yy71) edge (yy81) edge (yy91);

\path[-,very thin]
(xx2) edge (yy21) edge (yy61) edge (yy71)
(xx4) edge (yy31) edge (yy51) edge (yy71);

\path[-,thick]
(yy11) edge (yy12) (yy12) edge (yy13) (yy13) edge (yy14)
(yy21) edge (yy22) (yy22) edge (yy23) (yy23) edge (yy24)
(yy31) edge (yy32) (yy32) edge (yy33) (yy33) edge (yy34)
(yy41) edge (yy42) (yy42) edge (yy43) (yy43) edge (yy44)
(yy51) edge (yy52) (yy52) edge (yy53) (yy53) edge (yy54)
(yy61) edge (yy62) (yy62) edge (yy63) (yy63) edge (yy64)
(yy71) edge (yy72) (yy72) edge (yy73) (yy73) edge (yy74)
(yy81) edge (yy82) (yy82) edge (yy83) (yy83) edge (yy84)
(yy91) edge (yy92) (yy92) edge (yy93) (yy93) edge (yy94);
\end{tikzpicture}}

\caption{Case 1 ($s<2r+2$) when $d+1<s<2d+2$. Reduction from Set Cover instance $(\mathcal{S},c)$ in Lemma~\ref{lem-npc1} proof, where $c=3$, $S_1=\{1,2,3\}$, $S_2=\{2,6,7\}$, $S_3=\{4,5,6\}$, $S_4=\{3,5,7\}$, $S_5=\{7,8,9\}$ and $U=\{1,2,3,4,5,6,7,8,9\}$.}\label{fig-reduction}
\end{figure}

\begin{figure}[ht]\centering\scalebox{0.65}{

\begin{tikzpicture}[scale=0.9]
\tikzstyle{vertex}=[draw,circle,fill=white!25,minimum size=8pt,inner sep=1pt]
\tikzstyle{cvertex}=[draw,circle,fill=black!50,minimum size=10pt,inner sep=1pt]
\tikzstyle{xvertex}=[draw,circle,fill=gray!50,minimum size=10pt,inner sep=1pt]
\tikzstyle{zvertex}=[draw,circle,fill=white!50,minimum size=12pt,inner sep=1pt]

\node[]() at (1,0.0) {$s=3$};
\node[]() at (1,0.5) {$d=3$};
\node[]() at (1,3.0) {$p=5$};
\node[]() at (1,2.5) {$q=0$};
\node[]() at (1,2.0) {$K=3$};

\node at (5.35,0) {*};
\node[zvertex] (z0) at (5,0) {$z_0$};

\node[cvertex] (x1) at (3,2) {$S_1$}; \node at (2.8,1.6) {$\bullet$};
\node[xvertex] (x2) at (4,2) {$S_2$};
\node[cvertex] (x3) at (5,2) {$S_3$}; \node at (4.8,1.6) {$\bullet$};
\node[xvertex] (x4) at (6,2) {$S_4$};
\node[cvertex] (x5) at (7,2) {$S_5$}; \node at (6.8,1.6) {$\bullet$};

\node[vertex] (y11) at (1,4) {$u_{1,1}$};
\node[vertex] (y21) at (2,4) {$u_{2,1}$};
\node[vertex] (y31) at (3,4) {$u_{3,1}$};
\node[vertex] (y41) at (4,4) {$u_{4,1}$};
\node[vertex] (y51) at (5,4) {$u_{5,1}$};
\node[vertex] (y61) at (6,4) {$u_{6,1}$};
\node[vertex] (y71) at (7,4) {$u_{7,1}$};
\node[vertex] (y81) at (8,4) {$u_{8,1}$};
\node[vertex] (y91) at (9,4) {$u_{9,1}$};

\node[vertex] (y12) at (1,4.8) {};
\node[vertex] (y22) at (2,4.8) {};
\node[vertex] (y32) at (3,4.8) {};
\node[vertex] (y42) at (4,4.8) {};
\node[vertex] (y52) at (5,4.8) {};
\node[vertex] (y62) at (6,4.8) {};
\node[vertex] (y72) at (7,4.8) {};
\node[vertex] (y82) at (8,4.8) {};
\node[vertex] (y92) at (9,4.8) {};

\node[vertex] (y13) at (1,5.4) {};
\node[vertex] (y23) at (2,5.4) {};
\node[vertex] (y33) at (3,5.4) {};
\node[vertex] (y43) at (4,5.4) {};
\node[vertex] (y53) at (5,5.4) {};
\node[vertex] (y63) at (6,5.4) {};
\node[vertex] (y73) at (7,5.4) {};
\node[vertex] (y83) at (8,5.4) {};
\node[vertex] (y93) at (9,5.4) {};

\node[vertex] (y14) at (1,6.0) {};
\node[vertex] (y24) at (2,6.0) {};
\node[vertex] (y34) at (3,6.0) {};
\node[vertex] (y44) at (4,6.0) {};
\node[vertex] (y54) at (5,6.0) {};
\node[vertex] (y64) at (6,6.0) {};
\node[vertex] (y74) at (7,6.0) {};
\node[vertex] (y84) at (8,6.0) {};
\node[vertex] (y94) at (9,6.0) {};

\node[vertex] (y15) at (1,6.8) {$u_{1,5}$};
\node[vertex] (y25) at (2,6.8) {$u_{2,5}$};
\node[vertex] (y35) at (3,6.8) {$u_{3,5}$};
\node[vertex] (y45) at (4,6.8) {$u_{4,5}$};
\node[vertex] (y55) at (5,6.8) {$u_{5,5}$};
\node[vertex] (y65) at (6,6.8) {$u_{6,5}$};
\node[vertex] (y75) at (7,6.8) {$u_{7,5}$};
\node[vertex] (y85) at (8,6.8) {$u_{8,5}$};
\node[vertex] (y95) at (9,6.8) {$u_{9,5}$};

\path[-,thin]
(z0) edge (x1) edge (x2) edge (x3) edge (x4) edge (x5);

\path[-,very thick]
(x1) edge (y11) edge (y21) edge (y31)
(x3) edge (y41) edge (y51) edge (y61)
(x5) edge (y71) edge (y81) edge (y91);

\path[-,very thin]
(x2) edge (y21) edge (y61) edge (y71)
(x4) edge (y31) edge (y51) edge (y71);

\path[-,thick]
(y11) edge (y12) (y12) edge (y13) (y13) edge (y14) (y14) edge (y15)
(y21) edge (y22) (y22) edge (y23) (y23) edge (y24) (y24) edge (y25)
(y31) edge (y32) (y32) edge (y33) (y33) edge (y34) (y34) edge (y35)
(y41) edge (y42) (y42) edge (y43) (y43) edge (y44) (y44) edge (y45)
(y51) edge (y52) (y52) edge (y53) (y53) edge (y54) (y54) edge (y55)
(y61) edge (y62) (y62) edge (y63) (y63) edge (y64) (y64) edge (y65)
(y71) edge (y72) (y72) edge (y73) (y73) edge (y74) (y74) edge (y75)
(y81) edge (y82) (y82) edge (y83) (y83) edge (y84) (y84) edge (y85)
(y91) edge (y92) (y92) edge (y93) (y93) edge (y94) (y94) edge (y95);

\node[]() at (11,0.0) {$s=3$};
\node[]() at (11,0.5) {$d=5$};
\node[]() at (11,3.0) {$p=8$};
\node[]() at (11,2.5) {$q=0$};
\node[]() at (11,2.0) {$K=3$};

\node at (15.35,0) {*};
\node[zvertex] (zz0) at (15,0) {$z_0$};

\node[cvertex] (xx1) at (13,2) {$S_1$}; \node at (12.8,1.6) {$\bullet$};
\node[xvertex] (xx2) at (14,2) {$S_2$};
\node[cvertex] (xx3) at (15,2) {$S_3$}; \node at (14.8,1.6) {$\bullet$};
\node[xvertex] (xx4) at (16,2) {$S_4$};
\node[cvertex] (xx5) at (17,2) {$S_5$}; \node at (16.8,1.6) {$\bullet$};

\node[vertex] (yy11) at (11,4) {$u_{1,1}$};
\node[vertex] (yy21) at (12,4) {$u_{2,1}$};
\node[vertex] (yy31) at (13,4) {$u_{3,1}$};
\node[vertex] (yy41) at (14,4) {$u_{4,1}$};
\node[vertex] (yy51) at (15,4) {$u_{5,1}$};
\node[vertex] (yy61) at (16,4) {$u_{6,1}$};
\node[vertex] (yy71) at (17,4) {$u_{7,1}$};
\node[vertex] (yy81) at (18,4) {$u_{8,1}$};
\node[vertex] (yy91) at (19,4) {$u_{9,1}$};

\node[vertex] (yy12) at (11,4.8) {};
\node[vertex] (yy22) at (12,4.8) {};
\node[vertex] (yy32) at (13,4.8) {};
\node[vertex] (yy42) at (14,4.8) {};
\node[vertex] (yy52) at (15,4.8) {};
\node[vertex] (yy62) at (16,4.8) {};
\node[vertex] (yy72) at (17,4.8) {};
\node[vertex] (yy82) at (18,4.8) {};
\node[vertex] (yy92) at (19,4.8) {};

\node[vertex] (yy13) at (11,5.4) {};
\node[vertex] (yy23) at (12,5.4) {};
\node[vertex] (yy33) at (13,5.4) {};
\node[vertex] (yy43) at (14,5.4) {};
\node[vertex] (yy53) at (15,5.4) {};
\node[vertex] (yy63) at (16,5.4) {};
\node[vertex] (yy73) at (17,5.4) {};
\node[vertex] (yy83) at (18,5.4) {};
\node[vertex] (yy93) at (19,5.4) {};

\node[vertex] (yy14) at (11,6.0) {};
\node[vertex] (yy24) at (12,6.0) {};
\node[vertex] (yy34) at (13,6.0) {};
\node[vertex] (yy44) at (14,6.0) {};
\node[vertex] (yy54) at (15,6.0) {};
\node[vertex] (yy64) at (16,6.0) {};
\node[vertex] (yy74) at (17,6.0) {};
\node[vertex] (yy84) at (18,6.0) {};
\node[vertex] (yy94) at (19,6.0) {};

\node[vertex] (yy15) at (11,6.6) {};
\node[vertex] (yy25) at (12,6.6) {};
\node[vertex] (yy35) at (13,6.6) {};
\node[vertex] (yy45) at (14,6.6) {};
\node[vertex] (yy55) at (15,6.6) {};
\node[vertex] (yy65) at (16,6.6) {};
\node[vertex] (yy75) at (17,6.6) {};
\node[vertex] (yy85) at (18,6.6) {};
\node[vertex] (yy95) at (19,6.6) {};

\node[vertex] (yy16) at (11,7.2) {};
\node[vertex] (yy26) at (12,7.2) {};
\node[vertex] (yy36) at (13,7.2) {};
\node[vertex] (yy46) at (14,7.2) {};
\node[vertex] (yy56) at (15,7.2) {};
\node[vertex] (yy66) at (16,7.2) {};
\node[vertex] (yy76) at (17,7.2) {};
\node[vertex] (yy86) at (18,7.2) {};
\node[vertex] (yy96) at (19,7.2) {};

\node[vertex] (yy17) at (11,7.8) {};
\node[vertex] (yy27) at (12,7.8) {};
\node[vertex] (yy37) at (13,7.8) {};
\node[vertex] (yy47) at (14,7.8) {};
\node[vertex] (yy57) at (15,7.8) {};
\node[vertex] (yy67) at (16,7.8) {};
\node[vertex] (yy77) at (17,7.8) {};
\node[vertex] (yy87) at (18,7.8) {};
\node[vertex] (yy97) at (19,7.8) {};

\node[vertex] (yy18) at (11,8.6) {$u_{1,8}$};
\node[vertex] (yy28) at (12,8.6) {$u_{2,8}$};
\node[vertex] (yy38) at (13,8.6) {$u_{3,8}$};
\node[vertex] (yy48) at (14,8.6) {$u_{4,8}$};
\node[vertex] (yy58) at (15,8.6) {$u_{5,8}$};
\node[vertex] (yy68) at (16,8.6) {$u_{6,8}$};
\node[vertex] (yy78) at (17,8.6) {$u_{7,8}$};
\node[vertex] (yy88) at (18,8.6) {$u_{8,8}$};
\node[vertex] (yy98) at (19,8.6) {$u_{9,8}$};

\path[-,thin]
(zz0) edge (xx1) edge (xx2) edge (xx3) edge (xx4) edge (xx5);

\path[-,very thick]
(xx1) edge (yy11) edge (yy21) edge (yy31)
(xx3) edge (yy41) edge (yy51) edge (yy61)
(xx5) edge (yy71) edge (yy81) edge (yy91);

\path[-,very thin]
(xx2) edge (yy21) edge (yy61) edge (yy71)
(xx4) edge (yy31) edge (yy51) edge (yy71);

\path[-,thick]
(yy11) edge (yy12) (yy12) edge (yy13) (yy13) edge (yy14) (yy14) edge (yy15) (yy15) edge (yy16) (yy16) edge (yy17) (yy17) edge (yy18)
(yy21) edge (yy22) (yy22) edge (yy23) (yy23) edge (yy24) (yy24) edge (yy25) (yy25) edge (yy26) (yy26) edge (yy27) (yy27) edge (yy28)
(yy31) edge (yy32) (yy32) edge (yy33) (yy33) edge (yy34) (yy34) edge (yy35) (yy35) edge (yy36) (yy36) edge (yy37) (yy37) edge (yy38)
(yy41) edge (yy42) (yy42) edge (yy43) (yy43) edge (yy44) (yy44) edge (yy45) (yy45) edge (yy46) (yy46) edge (yy47) (yy47) edge (yy48)
(yy51) edge (yy52) (yy52) edge (yy53) (yy53) edge (yy54) (yy54) edge (yy55) (yy55) edge (yy56) (yy56) edge (yy57) (yy57) edge (yy58)
(yy61) edge (yy62) (yy62) edge (yy63) (yy63) edge (yy64) (yy64) edge (yy65) (yy65) edge (yy66) (yy66) edge (yy67) (yy67) edge (yy68)
(yy71) edge (yy72) (yy72) edge (yy73) (yy73) edge (yy74) (yy74) edge (yy75) (yy75) edge (yy76) (yy76) edge (yy77) (yy77) edge (yy78)
(yy81) edge (yy82) (yy82) edge (yy83) (yy83) edge (yy84) (yy84) edge (yy85) (yy85) edge (yy86) (yy86) edge (yy87) (yy87) edge (yy88)
(yy91) edge (yy92) (yy92) edge (yy93) (yy93) edge (yy94) (yy94) edge (yy95) (yy95) edge (yy96) (yy96) edge (yy97) (yy97) edge (yy98);

\node[]() at (21,0.0) {$s=4$};
\node[]() at (21,0.5) {$d=5$};
\node[]() at (21,3.0) {$p=7$};
\node[]() at (21,2.5) {$q=0$};
\node[]() at (21,2.0) {$K=3$};

\node at (25.35,0) {*};
\node[zvertex] (zz0) at (25,0) {$z_0$};

\node[cvertex] (xx1) at (23,2) {$S_1$}; \node at (22.8,1.6) {$\bullet$};
\node[xvertex] (xx2) at (24,2) {$S_2$};
\node[cvertex] (xx3) at (25,2) {$S_3$}; \node at (24.8,1.6) {$\bullet$};
\node[xvertex] (xx4) at (26,2) {$S_4$};
\node[cvertex] (xx5) at (27,2) {$S_5$}; \node at (26.8,1.6) {$\bullet$};

\node[vertex] (yy11) at (21,4) {$u_{1,1}$};
\node[vertex] (yy21) at (22,4) {$u_{2,1}$};
\node[vertex] (yy31) at (23,4) {$u_{3,1}$};
\node[vertex] (yy41) at (24,4) {$u_{4,1}$};
\node[vertex] (yy51) at (25,4) {$u_{5,1}$};
\node[vertex] (yy61) at (26,4) {$u_{6,1}$};
\node[vertex] (yy71) at (27,4) {$u_{7,1}$};
\node[vertex] (yy81) at (28,4) {$u_{8,1}$};
\node[vertex] (yy91) at (29,4) {$u_{9,1}$};

\node[vertex] (yy12) at (21,4.8) {};
\node[vertex] (yy22) at (22,4.8) {};
\node[vertex] (yy32) at (23,4.8) {};
\node[vertex] (yy42) at (24,4.8) {};
\node[vertex] (yy52) at (25,4.8) {};
\node[vertex] (yy62) at (26,4.8) {};
\node[vertex] (yy72) at (27,4.8) {};
\node[vertex] (yy82) at (28,4.8) {};
\node[vertex] (yy92) at (29,4.8) {};

\node[vertex] (yy13) at (21,5.4) {};
\node[vertex] (yy23) at (22,5.4) {};
\node[vertex] (yy33) at (23,5.4) {};
\node[vertex] (yy43) at (24,5.4) {};
\node[vertex] (yy53) at (25,5.4) {};
\node[vertex] (yy63) at (26,5.4) {};
\node[vertex] (yy73) at (27,5.4) {};
\node[vertex] (yy83) at (28,5.4) {};
\node[vertex] (yy93) at (29,5.4) {};

\node[vertex] (yy14) at (21,6.0) {};
\node[vertex] (yy24) at (22,6.0) {};
\node[vertex] (yy34) at (23,6.0) {};
\node[vertex] (yy44) at (24,6.0) {};
\node[vertex] (yy54) at (25,6.0) {};
\node[vertex] (yy64) at (26,6.0) {};
\node[vertex] (yy74) at (27,6.0) {};
\node[vertex] (yy84) at (28,6.0) {};
\node[vertex] (yy94) at (29,6.0) {};

\node[vertex] (yy15) at (21,6.6) {};
\node[vertex] (yy25) at (22,6.6) {};
\node[vertex] (yy35) at (23,6.6) {};
\node[vertex] (yy45) at (24,6.6) {};
\node[vertex] (yy55) at (25,6.6) {};
\node[vertex] (yy65) at (26,6.6) {};
\node[vertex] (yy75) at (27,6.6) {};
\node[vertex] (yy85) at (28,6.6) {};
\node[vertex] (yy95) at (29,6.6) {};

\node[vertex] (yy16) at (21,7.2) {};
\node[vertex] (yy26) at (22,7.2) {};
\node[vertex] (yy36) at (23,7.2) {};
\node[vertex] (yy46) at (24,7.2) {};
\node[vertex] (yy56) at (25,7.2) {};
\node[vertex] (yy66) at (26,7.2) {};
\node[vertex] (yy76) at (27,7.2) {};
\node[vertex] (yy86) at (28,7.2) {};
\node[vertex] (yy96) at (29,7.2) {};

\node[vertex] (yy17) at (21,8) {$u_{1,7}$};
\node[vertex] (yy27) at (22,8) {$u_{2,7}$};
\node[vertex] (yy37) at (23,8) {$u_{3,7}$};
\node[vertex] (yy47) at (24,8) {$u_{4,7}$};
\node[vertex] (yy57) at (25,8) {$u_{5,7}$};
\node[vertex] (yy67) at (26,8) {$u_{6,7}$};
\node[vertex] (yy77) at (27,8) {$u_{7,7}$};
\node[vertex] (yy87) at (28,8) {$u_{8,7}$};
\node[vertex] (yy97) at (29,8) {$u_{9,7}$};

\path[-,thin]
(zz0) edge (xx1) edge (xx2) edge (xx3) edge (xx4) edge (xx5);

\path[-,very thick]
(xx1) edge (yy11) edge (yy21) edge (yy31)
(xx3) edge (yy41) edge (yy51) edge (yy61)
(xx5) edge (yy71) edge (yy81) edge (yy91);

\path[-,very thin]
(xx2) edge (yy21) edge (yy61) edge (yy71)
(xx4) edge (yy31) edge (yy51) edge (yy71);

\path[-,thick]
(yy11) edge (yy12) (yy12) edge (yy13) (yy13) edge (yy14) (yy14) edge (yy15) (yy15) edge (yy16) (yy16) edge (yy17)
(yy21) edge (yy22) (yy22) edge (yy23) (yy23) edge (yy24) (yy24) edge (yy25) (yy25) edge (yy26) (yy26) edge (yy27)
(yy31) edge (yy32) (yy32) edge (yy33) (yy33) edge (yy34) (yy34) edge (yy35) (yy35) edge (yy36) (yy36) edge (yy37)
(yy41) edge (yy42) (yy42) edge (yy43) (yy43) edge (yy44) (yy44) edge (yy45) (yy45) edge (yy46) (yy46) edge (yy47)
(yy51) edge (yy52) (yy52) edge (yy53) (yy53) edge (yy54) (yy54) edge (yy55) (yy55) edge (yy56) (yy56) edge (yy57)
(yy61) edge (yy62) (yy62) edge (yy63) (yy63) edge (yy64) (yy64) edge (yy65) (yy65) edge (yy66) (yy66) edge (yy67)
(yy71) edge (yy72) (yy72) edge (yy73) (yy73) edge (yy74) (yy74) edge (yy75) (yy75) edge (yy76) (yy76) edge (yy77)
(yy81) edge (yy82) (yy82) edge (yy83) (yy83) edge (yy84) (yy84) edge (yy85) (yy85) edge (yy86) (yy86) edge (yy87)
(yy91) edge (yy92) (yy92) edge (yy93) (yy93) edge (yy94) (yy94) edge (yy95) (yy95) edge (yy96) (yy96) edge (yy97);
\end{tikzpicture}}

\caption{Case 1 ($s<2r+2$) when $s\leq d+1$. Reduction for Figure \ref{fig-reduction} Set Cover instance in Lemma~\ref{lem-npc1} proof.}\label{fig-reduction3a}
\end{figure}

%% file: figs/case2.tex
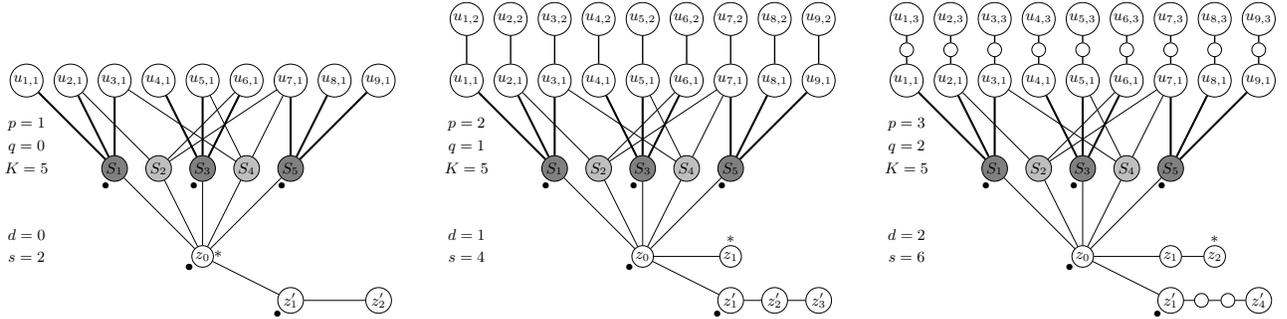
\begin{figure}[ht]\centering\scalebox{0.65}{

\begin{tikzpicture}[scale=0.9]
\tikzstyle{vertex}=[draw,circle,fill=white!25,minimum size=8pt,inner sep=1pt]
\tikzstyle{cvertex}=[draw,circle,fill=black!50,minimum size=10pt,inner sep=1pt]
\tikzstyle{xvertex}=[draw,circle,fill=gray!50,minimum size=10pt,inner sep=1pt]
\tikzstyle{zvertex}=[draw,circle,fill=white!50,minimum size=12pt,inner sep=1pt]

\node[]() at (1,0.0) {$s\geq 3$};
\node[]() at (1,0.5) {$d=0$};
\node[]() at (1,3.0) {$p=1$};
\node[]() at (1,2.5) {$q=1$};
\node[]() at (1,2.0) {$K=5$};

\node at (7,0.35) {*};
\node at (6.7,-.25) {$\bullet$};
\node at (4.7,-.25) {$\bullet$};
\node[zvertex] (z0) at (5,0) {$z_0$};
\node[zvertex] (z1) at (7,0) {$z_1$};
\node[zvertex] (zl1) at (7,-1) {$z'_1$};

\node[cvertex] (x1) at (3,2) {$S_1$}; \node at (2.8,1.6) {$\bullet$};
\node[xvertex] (x2) at (4,2) {$S_2$};
\node[cvertex] (x3) at (5,2) {$S_3$}; \node at (4.8,1.6) {$\bullet$};
\node[xvertex] (x4) at (6,2) {$S_4$};
\node[cvertex] (x5) at (7,2) {$S_5$}; \node at (6.8,1.6) {$\bullet$};

\node[vertex] (y11) at (1,4) {$u_{1,1}$};
\node[vertex] (y21) at (2,4) {$u_{2,1}$};
\node[vertex] (y31) at (3,4) {$u_{3,1}$};
\node[vertex] (y41) at (4,4) {$u_{4,1}$};
\node[vertex] (y51) at (5,4) {$u_{5,1}$};
\node[vertex] (y61) at (6,4) {$u_{6,1}$};
\node[vertex] (y71) at (7,4) {$u_{7,1}$};
\node[vertex] (y81) at (8,4) {$u_{8,1}$};
\node[vertex] (y91) at (9,4) {$u_{9,1}$};

\path[-,thin]
(z0) edge (z1) edge (zl1) edge (x1) edge (x2) edge (x3) edge (x4) edge (x5);

\path[-,very thick]
(x1) edge (y11) edge (y21) edge (y31)
(x3) edge (y41) edge (y51) edge (y61)
(x5) edge (y71) edge (y81) edge (y91);

\path[-,very thin]
(x2) edge (y21) edge (y61) edge (y71)
(x4) edge (y31) edge (y51) edge (y71);

\node[]() at (11,0)   {$s\geq 5$};
\node[]() at (11,0.5) {$d=1$};
\node[]() at (11,3.0) {$p=2$};
\node[]() at (11,2.5) {$q=2$};
\node[]() at (11,2.0) {$K=5$};

\node at (18,0.35) {*};
\node at (16.7,-.25) {$\bullet$};
\node at (14.7,-.25) {$\bullet$};
\node[zvertex] (zz0) at (15, 0) {$z_0$};
\node[vertex] (zz1)  at (17, 0) {$z_1$};
\node[vertex] (zzl1) at (17,-1) {$z'_1$};
\node[vertex] (zz2)  at (18, 0) {$z_2$};
\node[vertex] (zzl2) at (18,-1) {$z'_2$};

\node[cvertex] (xx1) at (13,2) {$S_1$}; \node at (12.8,1.6) {$\bullet$};
\node[xvertex] (xx2) at (14,2) {$S_2$};
\node[cvertex] (xx3) at (15,2) {$S_3$}; \node at (14.8,1.6) {$\bullet$};
\node[xvertex] (xx4) at (16,2) {$S_4$};
\node[cvertex] (xx5) at (17,2) {$S_5$}; \node at (16.8,1.6) {$\bullet$};

\node[vertex] (yy11) at (11,4) {$u_{1,1}$};
\node[vertex] (yy21) at (12,4) {$u_{2,1}$};
\node[vertex] (yy31) at (13,4) {$u_{3,1}$};
\node[vertex] (yy41) at (14,4) {$u_{4,1}$};
\node[vertex] (yy51) at (15,4) {$u_{5,1}$};
\node[vertex] (yy61) at (16,4) {$u_{6,1}$};
\node[vertex] (yy71) at (17,4) {$u_{7,1}$};
\node[vertex] (yy81) at (18,4) {$u_{8,1}$};
\node[vertex] (yy91) at (19,4) {$u_{9,1}$};

\node[vertex] (yy12) at (11,5.4) {$u_{1,2}$};
\node[vertex] (yy22) at (12,5.4) {$u_{2,2}$};
\node[vertex] (yy32) at (13,5.4) {$u_{3,2}$};
\node[vertex] (yy42) at (14,5.4) {$u_{4,2}$};
\node[vertex] (yy52) at (15,5.4) {$u_{5,2}$};
\node[vertex] (yy62) at (16,5.4) {$u_{6,2}$};
\node[vertex] (yy72) at (17,5.4) {$u_{7,2}$};
\node[vertex] (yy82) at (18,5.4) {$u_{8,2}$};
\node[vertex] (yy92) at (19,5.4) {$u_{9,2}$};

\path[-,thin]
(zz0) edge (xx1) edge (xx2) edge (xx3) edge (xx4) edge (xx5)
(zz1) edge (zz0) edge (zz2)
(zzl1) edge (zz0) edge (zzl2);

\path[-,very thick]
(xx1) edge (yy11) edge (yy21) edge (yy31)
(xx3) edge (yy41) edge (yy51) edge (yy61)
(xx5) edge (yy71) edge (yy81) edge (yy91);

\path[-,very thin]
(xx2) edge (yy21) edge (yy61) edge (yy71)
(xx4) edge (yy31) edge (yy51) edge (yy71);

\path[-,thick]
(yy11) edge (yy12)
(yy21) edge (yy22)
(yy31) edge (yy32)
(yy41) edge (yy42)
(yy51) edge (yy52)
(yy61) edge (yy62)
(yy71) edge (yy72)
(yy81) edge (yy82)
(yy91) edge (yy92);

\node[]() at (21,0.0) {$s\geq 7$};
\node[]() at (21,0.5) {$d=2$};
\node[]() at (21,3.0) {$p=3$};
\node[]() at (21,2.5) {$q=3$};
\node[]() at (21,2.0) {$K=5$};

\node at (29,0.35) {*};
\node at (24.7,-.25) {$\bullet$};
\node at (26.7,-.25) {$\bullet$};
\node[zvertex] (z0) at (25, 0) {$z_0$};
\node[vertex] (zz1) at (27, 0) {$z_1$};
\node[vertex] (zzl1)at (27,-1) {$z'_1$};
\node[vertex] (zz2) at (28, 0) {$z_2$};
\node[vertex] (zzl2)at (28,-1) {$z'_2$};
\node[vertex] (zz3) at (29, 0) {$z_3$};
\node[vertex] (zzl3)at (29,-1) {$z'_3$};

\node[cvertex] (x1) at (23,2) {$S_1$}; \node at (22.8,1.6) {$\bullet$};
\node[xvertex] (x2) at (24,2) {$S_2$};
\node[cvertex] (x3) at (25,2) {$S_3$}; \node at (24.8,1.6) {$\bullet$};
\node[xvertex] (x4) at (26,2) {$S_4$};
\node[cvertex] (x5) at (27,2) {$S_5$}; \node at (26.8,1.6) {$\bullet$};

\node[vertex] (y11) at (21,4) {$u_{1,1}$};
\node[vertex] (y21) at (22,4) {$u_{2,1}$};
\node[vertex] (y31) at (23,4) {$u_{3,1}$};
\node[vertex] (y41) at (24,4) {$u_{4,1}$};
\node[vertex] (y51) at (25,4) {$u_{5,1}$};
\node[vertex] (y61) at (26,4) {$u_{6,1}$};
\node[vertex] (y71) at (27,4) {$u_{7,1}$};
\node[vertex] (y81) at (28,4) {$u_{8,1}$};
\node[vertex] (y91) at (29,4) {$u_{9,1}$};

\node[vertex] (y12) at (21,4.7) {};
\node[vertex] (y22) at (22,4.7) {};
\node[vertex] (y32) at (23,4.7) {};
\node[vertex] (y42) at (24,4.7) {};
\node[vertex] (y52) at (25,4.7) {};
\node[vertex] (y62) at (26,4.7) {};
\node[vertex] (y72) at (27,4.7) {};
\node[vertex] (y82) at (28,4.7) {};
\node[vertex] (y92) at (29,4.7) {};

\node[vertex] (y13) at (21,5.4) {$u_{1,3}$};
\node[vertex] (y23) at (22,5.4) {$u_{2,3}$};
\node[vertex] (y33) at (23,5.4) {$u_{3,3}$};
\node[vertex] (y43) at (24,5.4) {$u_{4,3}$};
\node[vertex] (y53) at (25,5.4) {$u_{5,3}$};
\node[vertex] (y63) at (26,5.4) {$u_{6,3}$};
\node[vertex] (y73) at (27,5.4) {$u_{7,3}$};
\node[vertex] (y83) at (28,5.4) {$u_{8,3}$};
\node[vertex] (y93) at (29,5.4) {$u_{9,3}$};

\path[-,thin]
(z0) edge (x1) edge (x2) edge (x3) edge (x4) edge (x5)
(zz1) edge (z0) edge (zz2) (zz2) edge (zz3)
(zzl1) edge (z0) edge (zzl2) (zzl2) edge (zzl3);

\path[-,very thick]
(x1) edge (y11) edge (y21) edge (y31)
(x3) edge (y41) edge (y51) edge (y61)
(x5) edge (y71) edge (y81) edge (y91);

\path[-,very thin]
(x2) edge (y21) edge (y61) edge (y71)
(x4) edge (y31) edge (y51) edge (y71);

\path[-,thick]
(y11) edge (y12) (y12) edge (y13)
(y21) edge (y22) (y22) edge (y23)
(y31) edge (y32) (y32) edge (y33)
(y41) edge (y42) (y42) edge (y43)
(y51) edge (y52) (y52) edge (y53)
(y61) edge (y62) (y62) edge (y63)
(y71) edge (y72) (y72) edge (y73)
(y81) edge (y82) (y82) edge (y83)
(y91) edge (y92) (y92) edge (y93);
\end{tikzpicture}}

\caption{Case 2 ($s>2r+2$) when $s>2d+2$ and consequently $r=d$. Reduction for Figure \ref{fig-reduction} Set Cover instance in Lemma~\ref{lem-npc2} proof.}\label{fig-reduction2}
\end{figure}

\begin{figure}[ht]\centering\scalebox{0.65}{
\begin{tikzpicture}[scale=0.9]
\tikzstyle{vertex}=[draw,circle,fill=white!25,minimum size=8pt,inner sep=1pt]
\tikzstyle{cvertex}=[draw,circle,fill=black!50,minimum size=10pt,inner sep=1pt]
\tikzstyle{xvertex}=[draw,circle,fill=gray!50,minimum size=10pt,inner sep=1pt]
\tikzstyle{zvertex}=[draw,circle,fill=white!50,minimum size=12pt,inner sep=1pt]

\node[]() at (1,0.0) {$s=3$};
\node[]() at (1,0.5) {$d=2$};
\node[]() at (1,3.0) {$p=4$};
\node[]() at (1,2.5) {$q=4$};
\node[]() at (1,2.0) {$K=5$};

\node at (7.25,0.25) {*};
\node at (4.7,-.25) {$\bullet$};
\node at (6.7,-.25) {$\bullet$};
\node[zvertex](z0) at (5,0) {$z_0$};
\node[vertex] (z1) at (7.0,0) {$z_1$};
\node[vertex] (z2) at (7.7,0) {};
\node[vertex] (z3) at (8.3,0) {};
\node[vertex] (z4) at (9.0,0) {$z_4$};
\node[vertex] (zl1) at (7.0,-1) {$z'_1$};
\node[vertex] (zl2) at (7.7,-1) {};
\node[vertex] (zl3) at (8.3,-1) {};
\node[vertex] (zl4) at (9.0,-1) {$z'_4$};

\node[cvertex] (xx1) at (3,2) {$S_1$}; \node at (2.8,1.6) {$\bullet$};
\node[xvertex] (xx2) at (4,2) {$S_2$};
\node[cvertex] (xx3) at (5,2) {$S_3$}; \node at (4.8,1.6) {$\bullet$};
\node[xvertex] (xx4) at (6,2) {$S_4$};
\node[cvertex] (xx5) at (7,2) {$S_5$}; \node at (6.8,1.6) {$\bullet$};

\node[vertex] (yy11) at (1,4) {$u_{1,1}$};
\node[vertex] (yy21) at (2,4) {$u_{2,1}$};
\node[vertex] (yy31) at (3,4) {$u_{3,1}$};
\node[vertex] (yy41) at (4,4) {$u_{4,1}$};
\node[vertex] (yy51) at (5,4) {$u_{5,1}$};
\node[vertex] (yy61) at (6,4) {$u_{6,1}$};
\node[vertex] (yy71) at (7,4) {$u_{7,1}$};
\node[vertex] (yy81) at (8,4) {$u_{8,1}$};
\node[vertex] (yy91) at (9,4) {$u_{9,1}$};

\node[vertex] (yy12) at (1,4.7) {};
\node[vertex] (yy22) at (2,4.7) {};
\node[vertex] (yy32) at (3,4.7) {};
\node[vertex] (yy42) at (4,4.7) {};
\node[vertex] (yy52) at (5,4.7) {};
\node[vertex] (yy62) at (6,4.7) {};
\node[vertex] (yy72) at (7,4.7) {};
\node[vertex] (yy82) at (8,4.7) {};
\node[vertex] (yy92) at (9,4.7) {};

\node[vertex] (yy13) at (1,5.3) {};
\node[vertex] (yy23) at (2,5.3) {};
\node[vertex] (yy33) at (3,5.3) {};
\node[vertex] (yy43) at (4,5.3) {};
\node[vertex] (yy53) at (5,5.3) {};
\node[vertex] (yy63) at (6,5.3) {};
\node[vertex] (yy73) at (7,5.3) {};
\node[vertex] (yy83) at (8,5.3) {};
\node[vertex] (yy93) at (9,5.3) {};

\node[vertex] (yy14) at (1,6) {$u_{1,4}$};
\node[vertex] (yy24) at (2,6) {$u_{2,4}$};
\node[vertex] (yy34) at (3,6) {$u_{3,4}$};
\node[vertex] (yy44) at (4,6) {$u_{4,4}$};
\node[vertex] (yy54) at (5,6) {$u_{5,4}$};
\node[vertex] (yy64) at (6,6) {$u_{6,4}$};
\node[vertex] (yy74) at (7,6) {$u_{7,4}$};
\node[vertex] (yy84) at (8,6) {$u_{8,4}$};
\node[vertex] (yy94) at (9,6) {$u_{9,4}$};

\path[-,thin]
(z0) edge (xx1) edge (xx2) edge (xx3) edge (xx4) edge (xx5)
(z1) edge (z0) edge (z2) (z3) edge (z2) edge (z4)
(zl1) edge (z0) edge (zl2) (zl3) edge (zl2) edge (zl4);

\path[-,very thick]
(xx1) edge (yy11) edge (yy21) edge (yy31)
(xx3) edge (yy41) edge (yy51) edge (yy61)
(xx5) edge (yy71) edge (yy81) edge (yy91);

\path[-,very thin]
(xx2) edge (yy21) edge (yy61) edge (yy71)
(xx4) edge (yy31) edge (yy51) edge (yy71);

\path[-,thick]
(yy11) edge (yy12) (yy12) edge (yy13) (yy13) edge (yy14)
(yy21) edge (yy22) (yy22) edge (yy23) (yy23) edge (yy24)
(yy31) edge (yy32) (yy32) edge (yy33) (yy33) edge (yy34)
(yy41) edge (yy42) (yy42) edge (yy43) (yy43) edge (yy44)
(yy51) edge (yy52) (yy52) edge (yy53) (yy53) edge (yy54)
(yy61) edge (yy62) (yy62) edge (yy63) (yy63) edge (yy64)
(yy71) edge (yy72) (yy72) edge (yy73) (yy73) edge (yy74)
(yy81) edge (yy82) (yy82) edge (yy83) (yy83) edge (yy84)
(yy91) edge (yy92) (yy92) edge (yy93) (yy93) edge (yy94);

\node[]() at (11,0.0) {$s=4$};
\node[]() at (11,0.5) {$d=3$};
\node[]() at (11,3.0) {$p=5$};
\node[]() at (11,2.5) {$q=5$};
\node[]() at (11,2.0) {$K=5$};

\node at (17.05,0.25) {*};
\node at (14.7,-.25) {$\bullet$};
\node at (16.5,-.25) {$\bullet$};
\node[zvertex](z0) at (15.0, 0) {$z_0$};
\node[vertex] (z1) at (16.8, 0) {$z_1$};
\node[vertex] (z2) at (17.4, 0) {};
\node[vertex] (z3) at (17.9, 0) {};
\node[vertex] (z4) at (18.4, 0) {};
\node[vertex] (z5) at (19.0, 0) {$z_5$};
\node[vertex] (zl1) at (16.8,-1){$z'_1$};
\node[vertex] (zl2) at (17.4,-1){};
\node[vertex] (zl3) at (17.9,-1){};
\node[vertex] (zl4) at (18.4,-1){};
\node[vertex] (zl5) at (19.0,-1){$z'_5$};

\node[cvertex] (x1) at (13,2) {$S_1$}; \node at (12.8,1.6) {$\bullet$};
\node[xvertex] (x2) at (14,2) {$S_2$};
\node[cvertex] (x3) at (15,2) {$S_3$}; \node at (14.8,1.6) {$\bullet$};
\node[xvertex] (x4) at (16,2) {$S_4$};
\node[cvertex] (x5) at (17,2) {$S_5$}; \node at (16.8,1.6) {$\bullet$};

\node[vertex] (y11) at (11,4) {$u_{1,1}$};
\node[vertex] (y21) at (12,4) {$u_{2,1}$};
\node[vertex] (y31) at (13,4) {$u_{3,1}$};
\node[vertex] (y41) at (14,4) {$u_{4,1}$};
\node[vertex] (y51) at (15,4) {$u_{5,1}$};
\node[vertex] (y61) at (16,4) {$u_{6,1}$};
\node[vertex] (y71) at (17,4) {$u_{7,1}$};
\node[vertex] (y81) at (18,4) {$u_{8,1}$};
\node[vertex] (y91) at (19,4) {$u_{9,1}$};

\node[vertex] (y12) at (11,4.8) {};
\node[vertex] (y22) at (12,4.8) {};
\node[vertex] (y32) at (13,4.8) {};
\node[vertex] (y42) at (14,4.8) {};
\node[vertex] (y52) at (15,4.8) {};
\node[vertex] (y62) at (16,4.8) {};
\node[vertex] (y72) at (17,4.8) {};
\node[vertex] (y82) at (18,4.8) {};
\node[vertex] (y92) at (19,4.8) {};

\node[vertex] (y13) at (11,5.4) {};
\node[vertex] (y23) at (12,5.4) {};
\node[vertex] (y33) at (13,5.4) {};
\node[vertex] (y43) at (14,5.4) {};
\node[vertex] (y53) at (15,5.4) {};
\node[vertex] (y63) at (16,5.4) {};
\node[vertex] (y73) at (17,5.4) {};
\node[vertex] (y83) at (18,5.4) {};
\node[vertex] (y93) at (19,5.4) {};

\node[vertex] (y14) at (11,6.0) {};
\node[vertex] (y24) at (12,6.0) {};
\node[vertex] (y34) at (13,6.0) {};
\node[vertex] (y44) at (14,6.0) {};
\node[vertex] (y54) at (15,6.0) {};
\node[vertex] (y64) at (16,6.0) {};
\node[vertex] (y74) at (17,6.0) {};
\node[vertex] (y84) at (18,6.0) {};
\node[vertex] (y94) at (19,6.0) {};

\node[vertex] (y15) at (11,6.8) {$u_{1,5}$};
\node[vertex] (y25) at (12,6.8) {$u_{2,5}$};
\node[vertex] (y35) at (13,6.8) {$u_{3,5}$};
\node[vertex] (y45) at (14,6.8) {$u_{4,5}$};
\node[vertex] (y55) at (15,6.8) {$u_{5,5}$};
\node[vertex] (y65) at (16,6.8) {$u_{6,5}$};
\node[vertex] (y75) at (17,6.8) {$u_{7,5}$};
\node[vertex] (y85) at (18,6.8) {$u_{8,5}$};
\node[vertex] (y95) at (19,6.8) {$u_{9,5}$};

\path[-,thin]
(z0) edge (x1) edge (x2) edge (x3) edge (x4) edge (x5)
(z1) edge (z0) edge (z2) (z3) edge (z2) edge (z4) (z4) edge (z5)
(zl1)edge (z0) edge (zl2)(zl3)edge (zl2)edge (zl4)(zl4)edge (zl5);

\path[-,very thick]
(x1) edge (y11) edge (y21) edge (y31)
(x3) edge (y41) edge (y51) edge (y61)
(x5) edge (y71) edge (y81) edge (y91);

\path[-,very thin]
(x2) edge (y21) edge (y61) edge (y71)
(x4) edge (y31) edge (y51) edge (y71);

\path[-,thick]
(y11) edge (y12) (y12) edge (y13) (y13) edge (y14) (y14) edge (y15)
(y21) edge (y22) (y22) edge (y23) (y23) edge (y24) (y24) edge (y25)
(y31) edge (y32) (y32) edge (y33) (y33) edge (y34) (y34) edge (y35)
(y41) edge (y42) (y42) edge (y43) (y43) edge (y44) (y44) edge (y45)
(y51) edge (y52) (y52) edge (y53) (y53) edge (y54) (y54) edge (y55)
(y61) edge (y62) (y62) edge (y63) (y63) edge (y64) (y64) edge (y65)
(y71) edge (y72) (y72) edge (y73) (y73) edge (y74) (y74) edge (y75)
(y81) edge (y82) (y82) edge (y83) (y83) edge (y84) (y84) edge (y85)
(y91) edge (y92) (y92) edge (y93) (y93) edge (y94) (y94) edge (y95);

\node[]() at (21,0.0) {$s=3$};
\node[]() at (21,0.5) {$d=4$};
\node[]() at (21,3.0) {$p=7$};
\node[]() at (21,2.5) {$q=7$};
\node[]() at (21,2.0) {$K=5$};

\node at (26.05,0.25) {*};
\node at (24.7,-.25) {$\bullet$};
\node at (25.5,-.25) {$\bullet$};
\node[zvertex](z0) at (25.0, 0) {$z_0$};
\node[vertex] (z1) at (25.8, 0) {$z_1$};
\node[vertex] (z2) at (26.4, 0) {};
\node[vertex] (z3) at (26.9, 0) {};
\node[vertex] (z4) at (27.4, 0) {};
\node[vertex] (z5) at (27.9, 0) {};
\node[vertex] (z6) at (28.4, 0) {};
\node[vertex] (z7) at (29.0, 0) {$z_7$};
\node[vertex] (zl1) at (25.8,-1){$z'_1$};
\node[vertex] (zl2) at (26.4,-1){};
\node[vertex] (zl3) at (26.9,-1){};
\node[vertex] (zl4) at (27.4,-1){};
\node[vertex] (zl5) at (27.9,-1){};
\node[vertex] (zl6) at (28.4,-1){};
\node[vertex] (zl7) at (29.0,-1){$z'_7$};

\node[cvertex] (xx1) at (23,2) {$S_1$}; \node at (22.8,1.6) {$\bullet$};
\node[xvertex] (xx2) at (24,2) {$S_2$};
\node[cvertex] (xx3) at (25,2) {$S_3$}; \node at (24.8,1.6) {$\bullet$};
\node[xvertex] (xx4) at (26,2) {$S_4$};
\node[cvertex] (xx5) at (27,2) {$S_5$}; \node at (26.8,1.6) {$\bullet$};

\node[vertex] (yy11) at (21,4) {$u_{1,1}$};
\node[vertex] (yy21) at (22,4) {$u_{2,1}$};
\node[vertex] (yy31) at (23,4) {$u_{3,1}$};
\node[vertex] (yy41) at (24,4) {$u_{4,1}$};
\node[vertex] (yy51) at (25,4) {$u_{5,1}$};
\node[vertex] (yy61) at (26,4) {$u_{6,1}$};
\node[vertex] (yy71) at (27,4) {$u_{7,1}$};
\node[vertex] (yy81) at (28,4) {$u_{8,1}$};
\node[vertex] (yy91) at (29,4) {$u_{9,1}$};

\node[vertex] (yy12) at (21,4.8) {};
\node[vertex] (yy22) at (22,4.8) {};
\node[vertex] (yy32) at (23,4.8) {};
\node[vertex] (yy42) at (24,4.8) {};
\node[vertex] (yy52) at (25,4.8) {};
\node[vertex] (yy62) at (26,4.8) {};
\node[vertex] (yy72) at (27,4.8) {};
\node[vertex] (yy82) at (28,4.8) {};
\node[vertex] (yy92) at (29,4.8) {};

\node[vertex] (yy13) at (21,5.4) {};
\node[vertex] (yy23) at (22,5.4) {};
\node[vertex] (yy33) at (23,5.4) {};
\node[vertex] (yy43) at (24,5.4) {};
\node[vertex] (yy53) at (25,5.4) {};
\node[vertex] (yy63) at (26,5.4) {};
\node[vertex] (yy73) at (27,5.4) {};
\node[vertex] (yy83) at (28,5.4) {};
\node[vertex] (yy93) at (29,5.4) {};

\node[vertex] (yy14) at (21,6.0) {};
\node[vertex] (yy24) at (22,6.0) {};
\node[vertex] (yy34) at (23,6.0) {};
\node[vertex] (yy44) at (24,6.0) {};
\node[vertex] (yy54) at (25,6.0) {};
\node[vertex] (yy64) at (26,6.0) {};
\node[vertex] (yy74) at (27,6.0) {};
\node[vertex] (yy84) at (28,6.0) {};
\node[vertex] (yy94) at (29,6.0) {};

\node[vertex] (yy15) at (21,6.6) {};
\node[vertex] (yy25) at (22,6.6) {};
\node[vertex] (yy35) at (23,6.6) {};
\node[vertex] (yy45) at (24,6.6) {};
\node[vertex] (yy55) at (25,6.6) {};
\node[vertex] (yy65) at (26,6.6) {};
\node[vertex] (yy75) at (27,6.6) {};
\node[vertex] (yy85) at (28,6.6) {};
\node[vertex] (yy95) at (29,6.6) {};

\node[vertex] (yy16) at (21,7.2) {};
\node[vertex] (yy26) at (22,7.2) {};
\node[vertex] (yy36) at (23,7.2) {};
\node[vertex] (yy46) at (24,7.2) {};
\node[vertex] (yy56) at (25,7.2) {};
\node[vertex] (yy66) at (26,7.2) {};
\node[vertex] (yy76) at (27,7.2) {};
\node[vertex] (yy86) at (28,7.2) {};
\node[vertex] (yy96) at (29,7.2) {};

\node[vertex] (yy17) at (21,8) {$u_{1,7}$};
\node[vertex] (yy27) at (22,8) {$u_{2,7}$};
\node[vertex] (yy37) at (23,8) {$u_{3,7}$};
\node[vertex] (yy47) at (24,8) {$u_{4,7}$};
\node[vertex] (yy57) at (25,8) {$u_{5,7}$};
\node[vertex] (yy67) at (26,8) {$u_{6,7}$};
\node[vertex] (yy77) at (27,8) {$u_{7,7}$};
\node[vertex] (yy87) at (28,8) {$u_{8,7}$};
\node[vertex] (yy97) at (29,8) {$u_{9,7}$};

\path[-,thin]
(z0) edge (xx1) edge (xx2) edge (xx3) edge (xx4) edge (xx5)
(z1) edge (z0) edge (z2) (z3) edge (z2) edge (z4) (z5) edge (z4) edge (z6) (z6) edge (z7)
(zl1)edge (z0) edge (zl2)(zl3)edge (zl2)edge (zl4)(zl5)edge (zl4)edge (zl6)(zl6)edge (zl7);

\path[-,very thick]
(xx1) edge (yy11) edge (yy21) edge (yy31)
(xx3) edge (yy41) edge (yy51) edge (yy61)
(xx5) edge (yy71) edge (yy81) edge (yy91);

\path[-,very thin]
(xx2) edge (yy21) edge (yy61) edge (yy71)
(xx4) edge (yy31) edge (yy51) edge (yy71);

\path[-,thick]
(yy11) edge (yy12) (yy12) edge (yy13) (yy13) edge (yy14) (yy14) edge (yy15) (yy15) edge (yy16) (yy16) edge (yy17)
(yy21) edge (yy22) (yy22) edge (yy23) (yy23) edge (yy24) (yy24) edge (yy25) (yy25) edge (yy26) (yy26) edge (yy27)
(yy31) edge (yy32) (yy32) edge (yy33) (yy33) edge (yy34) (yy34) edge (yy35) (yy35) edge (yy36) (yy36) edge (yy37)
(yy41) edge (yy42) (yy42) edge (yy43) (yy43) edge (yy44) (yy44) edge (yy45) (yy45) edge (yy46) (yy46) edge (yy47)
(yy51) edge (yy52) (yy52) edge (yy53) (yy53) edge (yy54) (yy54) edge (yy55) (yy55) edge (yy56) (yy56) edge (yy57)
(yy61) edge (yy62) (yy62) edge (yy63) (yy63) edge (yy64) (yy64) edge (yy65) (yy65) edge (yy66) (yy66) edge (yy67)
(yy71) edge (yy72) (yy72) edge (yy73) (yy73) edge (yy74) (yy74) edge (yy75) (yy75) edge (yy76) (yy76) edge (yy77)
(yy81) edge (yy82) (yy82) edge (yy83) (yy83) edge (yy84) (yy84) edge (yy85) (yy85) edge (yy86) (yy86) edge (yy87)
(yy91) edge (yy92) (yy92) edge (yy93) (yy93) edge (yy94) (yy94) edge (yy95) (yy95) edge (yy96) (yy96) edge (yy97);

\end{tikzpicture}}

\caption{Case 2 ($s>2r+2$) when $s\leq d+1$ and consequently $r<d$. Reduction for Figure \ref{fig-reduction} Set Cover instance in Lemma~\ref{lem-npc2} proof.}\label{fig-reduction3c}
\end{figure}
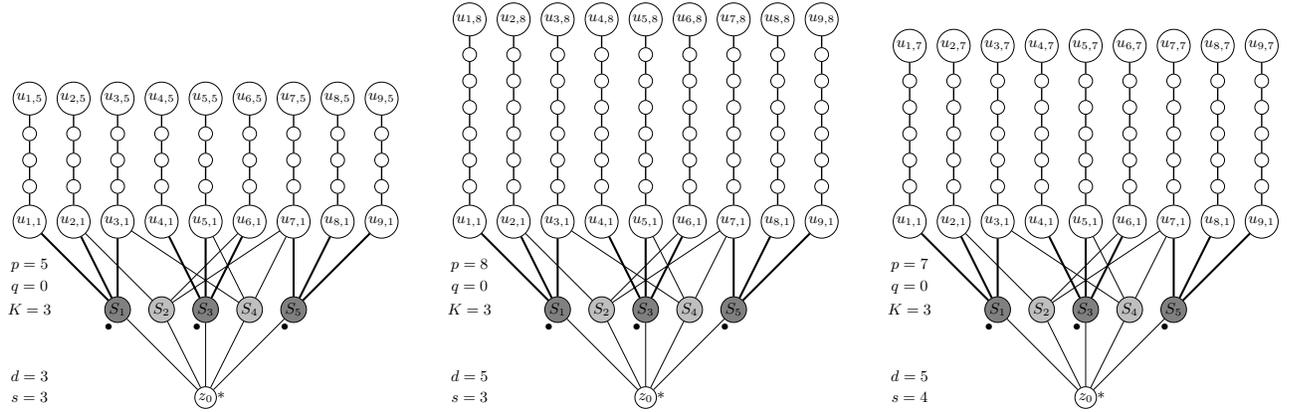

%% file: figs/case3.tex
\begin{figure}[ht]\centering\scalebox{0.65}{

\begin{tikzpicture}[scale=0.9]
\tikzstyle{vertex}=[draw,circle,fill=white!25,minimum size=8pt,inner sep=1pt]
\tikzstyle{cvertex}=[draw,circle,fill=black!50,minimum size=10pt,inner sep=1pt]
\tikzstyle{xvertex}=[draw,circle,fill=gray!50,minimum size=10pt,inner sep=1pt]
\tikzstyle{zvertex}=[draw,circle,fill=white!50,minimum size=12pt,inner sep=1pt]

\node[]() at (1,0.0) {$s=2$};
\node[]() at (1,0.5) {$d=0$};
\node[]() at (1,3.0) {$p=1$};
\node[]() at (1,2.5) {$q=0$};
\node[]() at (1,2.0) {$K=5$};

\node at (5.35,0) {*};
\node at (4.7,-.25) {$\bullet$};
\node at (4.7,-.25) {$\bullet$};
\node at (6.7,-1.3) {$\bullet$};
\node[zvertex] (z0) at (5,0) {$z_0$};
\node[vertex] (zl1) at (7,-1){$z'_1$};
\node[vertex] (zl2) at (9,-1){$z'_2$};
\path[-,thin] (zl1) edge (z0) edge (zl2);

\node[cvertex] (x1) at (3,2) {$S_1$}; \node at (2.8,1.6) {$\bullet$};
\node[xvertex] (x2) at (4,2) {$S_2$};
\node[cvertex] (x3) at (5,2) {$S_3$}; \node at (4.8,1.6) {$\bullet$};
\node[xvertex] (x4) at (6,2) {$S_4$};
\node[cvertex] (x5) at (7,2) {$S_5$}; \node at (6.8,1.6) {$\bullet$};

\node[vertex] (y11) at (1,4) {$u_{1,1}$};
\node[vertex] (y21) at (2,4) {$u_{2,1}$};
\node[vertex] (y31) at (3,4) {$u_{3,1}$};
\node[vertex] (y41) at (4,4) {$u_{4,1}$};
\node[vertex] (y51) at (5,4) {$u_{5,1}$};
\node[vertex] (y61) at (6,4) {$u_{6,1}$};
\node[vertex] (y71) at (7,4) {$u_{7,1}$};
\node[vertex] (y81) at (8,4) {$u_{8,1}$};
\node[vertex] (y91) at (9,4) {$u_{9,1}$};

\path[-,thin]
(z0) edge (x1) edge (x2) edge (x3) edge (x4) edge (x5);

\path[-,very thick]
(x1) edge (y11) edge (y21) edge (y31)
(x3) edge (y41) edge (y51) edge (y61)
(x5) edge (y71) edge (y81) edge (y91);

\path[-,very thin]
(x2) edge (y21) edge (y61) edge (y71)
(x4) edge (y31) edge (y51) edge (y71);

\node[]() at (11,0.0) {$s=4$};
\node[]() at (11,0.5) {$d=1$};
\node[]() at (11,3.0) {$p=2$};
\node[]() at (11,2.5) {$q=1$};
\node[]() at (11,2.0) {$K=5$};

\node at (17,0.35) {*};
\node at (14.7,-.25) {$\bullet$};
\node at (16.7,-1.3) {$\bullet$};
\node[zvertex](z0) at (15,0) {$z_0$};
\node[vertex] (z1) at (17,0) {$z_1$};
\node[vertex] (zl1) at (17,-1){$z'_1$};
\node[vertex] (zl2) at (18,-1){$z'_2$};
\node[vertex] (zl3) at (19,-1){$z'_3$};
\path[-,thin] (zl1) edge (z0) edge (zl2) (zl2) edge (zl3);

\node[cvertex] (xx1) at (13,2) {$S_1$}; \node at (12.8,1.6) {$\bullet$};
\node[xvertex] (xx2) at (14,2) {$S_2$};
\node[cvertex] (xx3) at (15,2) {$S_3$}; \node at (14.8,1.6) {$\bullet$};
\node[xvertex] (xx4) at (16,2) {$S_4$};
\node[cvertex] (xx5) at (17,2) {$S_5$}; \node at (16.8,1.6) {$\bullet$};

\node[vertex] (yy11) at (11,4) {$u_{1,1}$};
\node[vertex] (yy21) at (12,4) {$u_{2,1}$};
\node[vertex] (yy31) at (13,4) {$u_{3,1}$};
\node[vertex] (yy41) at (14,4) {$u_{4,1}$};
\node[vertex] (yy51) at (15,4) {$u_{5,1}$};
\node[vertex] (yy61) at (16,4) {$u_{6,1}$};
\node[vertex] (yy71) at (17,4) {$u_{7,1}$};
\node[vertex] (yy81) at (18,4) {$u_{8,1}$};
\node[vertex] (yy91) at (19,4) {$u_{9,1}$};

\node[vertex] (yy12) at (11,5.4) {$u_{1,2}$};
\node[vertex] (yy22) at (12,5.4) {$u_{2,2}$};
\node[vertex] (yy32) at (13,5.4) {$u_{3,2}$};
\node[vertex] (yy42) at (14,5.4) {$u_{4,2}$};
\node[vertex] (yy52) at (15,5.4) {$u_{5,2}$};
\node[vertex] (yy62) at (16,5.4) {$u_{6,2}$};
\node[vertex] (yy72) at (17,5.4) {$u_{7,2}$};
\node[vertex] (yy82) at (18,5.4) {$u_{8,2}$};
\node[vertex] (yy92) at (19,5.4) {$u_{9,2}$};

\path[-,thin]
(z0) edge (xx1) edge (xx2) edge (xx3) edge (xx4) edge (xx5)
(z0) edge (z1);

\path[-,very thick]
(xx1) edge (yy11) edge (yy21) edge (yy31)
(xx3) edge (yy41) edge (yy51) edge (yy61)
(xx5) edge (yy71) edge (yy81) edge (yy91);

\path[-,very thin]
(xx2) edge (yy21) edge (yy61) edge (yy71)
(xx4) edge (yy31) edge (yy51) edge (yy71);

\path[-,thick]
(yy11) edge (yy12)
(yy21) edge (yy22)
(yy31) edge (yy32)
(yy41) edge (yy42)
(yy51) edge (yy52)
(yy61) edge (yy62)
(yy71) edge (yy72)
(yy81) edge (yy82)
(yy91) edge (yy92);

\node[]() at (21,0.0) {$s=6$};
\node[]() at (21,0.5) {$d=2$};
\node[]() at (21,3.0) {$p=3$};
\node[]() at (21,2.5) {$q=2$};
\node[]() at (21,2.0) {$K=5$};

\node at (28,0.35) {*};
\node at (24.7,-.25) {$\bullet$};
\node at (26.7,-1.3) {$\bullet$};
\node[zvertex] (z0) at (25, 0) {$z_0$};
\node[vertex] (zz1) at (27, 0) {$z_1$};
\node[vertex] (zz2) at (28, 0) {$z_2$};
\node[vertex] (zl1) at (27,-1){$z'_1$};
\node[vertex] (zl2) at (27.7,-1){};
\node[vertex] (zl3) at (28.3,-1){};
\node[vertex] (zl4) at (29,-1){$z'_4$};
\path[-,thin] (zl1) edge (z0) edge (zl2) (zl3) edge (zl2) edge (zl4);

\node[cvertex] (x1) at (23,2) {$S_1$}; \node at (22.8,1.6) {$\bullet$};
\node[xvertex] (x2) at (24,2) {$S_2$};
\node[cvertex] (x3) at (25,2) {$S_3$}; \node at (24.8,1.6) {$\bullet$};
\node[xvertex] (x4) at (26,2) {$S_4$};
\node[cvertex] (x5) at (27,2) {$S_5$}; \node at (26.8,1.6) {$\bullet$};

\node[vertex] (y11) at (21,4) {$u_{1,1}$};
\node[vertex] (y21) at (22,4) {$u_{2,1}$};
\node[vertex] (y31) at (23,4) {$u_{3,1}$};
\node[vertex] (y41) at (24,4) {$u_{4,1}$};
\node[vertex] (y51) at (25,4) {$u_{5,1}$};
\node[vertex] (y61) at (26,4) {$u_{6,1}$};
\node[vertex] (y71) at (27,4) {$u_{7,1}$};
\node[vertex] (y81) at (28,4) {$u_{8,1}$};
\node[vertex] (y91) at (29,4) {$u_{9,1}$};

\node[vertex] (y12) at (21,4.7) {};
\node[vertex] (y22) at (22,4.7) {};
\node[vertex] (y32) at (23,4.7) {};
\node[vertex] (y42) at (24,4.7) {};
\node[vertex] (y52) at (25,4.7) {};
\node[vertex] (y62) at (26,4.7) {};
\node[vertex] (y72) at (27,4.7) {};
\node[vertex] (y82) at (28,4.7) {};
\node[vertex] (y92) at (29,4.7) {};

\node[vertex] (y13) at (21,5.4) {$u_{1,3}$};
\node[vertex] (y23) at (22,5.4) {$u_{2,3}$};
\node[vertex] (y33) at (23,5.4) {$u_{3,3}$};
\node[vertex] (y43) at (24,5.4) {$u_{4,3}$};
\node[vertex] (y53) at (25,5.4) {$u_{5,3}$};
\node[vertex] (y63) at (26,5.4) {$u_{6,3}$};
\node[vertex] (y73) at (27,5.4) {$u_{7,3}$};
\node[vertex] (y83) at (28,5.4) {$u_{8,3}$};
\node[vertex] (y93) at (29,5.4) {$u_{9,3}$};

\path[-,thin]
(z0) edge (x1) edge (x2) edge (x3) edge (x4) edge (x5)
(zz1) edge (z0) edge (zz2);

\path[-,very thick]
(x1) edge (y11) edge (y21) edge (y31)
(x3) edge (y41) edge (y51) edge (y61)
(x5) edge (y71) edge (y81) edge (y91);

\path[-,very thin]
(x2) edge (y21) edge (y61) edge (y71)
(x4) edge (y31) edge (y51) edge (y71);

\path[-,thick]
(y11) edge (y12) (y12) edge (y13)
(y21) edge (y22) (y22) edge (y23)
(y31) edge (y32) (y32) edge (y33)
(y41) edge (y42) (y42) edge (y43)
(y51) edge (y52) (y52) edge (y53)
(y61) edge (y62) (y62) edge (y63)
(y71) edge (y72) (y72) edge (y73)
(y81) edge (y82) (y82) edge (y83)
(y91) edge (y92) (y92) edge (y93);
\end{tikzpicture}}

\caption{Case 3 ($s=2r+2$) when $s=2d+2$ and consequently $r=d$. Reduction for Figure \ref{fig-reduction} Set Cover instance in Lemma \ref{lem-npc3} proof.}\label{fig-reduction2b}
\end{figure}

\begin{figure}[ht]\centering\scalebox{0.65}{

\begin{tikzpicture}[scale=0.9]
\tikzstyle{vertex}=[draw,circle,fill=white!25,minimum size=8pt,inner sep=1pt]
\tikzstyle{cvertex}=[draw,circle,fill=black!50,minimum size=10pt,inner sep=1pt]
\tikzstyle{xvertex}=[draw,circle,fill=gray!50,minimum size=10pt,inner sep=1pt]
\tikzstyle{zvertex}=[draw,circle,fill=white!50,minimum size=12pt,inner sep=1pt]

\node[]() at (1,0.0) {$s=2$};
\node[]() at (1,0.5) {$d=1$};
\node[]() at (1,3.0) {$p=3$};
\node[]() at (1,2.5) {$q=2$};
\node[]() at (1,2.0) {$K=5$};

\node at (5.35,0.11) {*};
\node at (4.7,-.25) {$\bullet$};
\node at (6.7,-1.3) {$\bullet$};
\node[zvertex] (z0) at (5, 0) {$z_0$};
\node[vertex] (z1)  at (7, 0) {$z_1$};
\node[vertex] (z2)  at (8, 0) {$z_2$};
\node[vertex] (zl1) at (7,-1) {$z'_1$};
\node[vertex] (zl2) at (7.7,-1) {};
\node[vertex] (zl3) at (8.3,-1) {};
\node[vertex] (zl4) at (9,-1) {$z'_4$};
\path[-,thin] (zl1) edge (z0) edge (zl2) (zl3) edge (zl2) edge (zl4);

\node[cvertex] (x1) at (3,2) {$S_1$}; \node at (2.8,1.6) {$\bullet$};
\node[xvertex] (x2) at (4,2) {$S_2$};
\node[cvertex] (x3) at (5,2) {$S_3$}; \node at (4.8,1.6) {$\bullet$};
\node[xvertex] (x4) at (6,2) {$S_4$};
\node[cvertex] (x5) at (7,2) {$S_5$}; \node at (6.8,1.6) {$\bullet$};

\node[vertex] (y11) at (1,4) {$u_{1,1}$};
\node[vertex] (y21) at (2,4) {$u_{2,1}$};
\node[vertex] (y31) at (3,4) {$u_{3,1}$};
\node[vertex] (y41) at (4,4) {$u_{4,1}$};
\node[vertex] (y51) at (5,4) {$u_{5,1}$};
\node[vertex] (y61) at (6,4) {$u_{6,1}$};
\node[vertex] (y71) at (7,4) {$u_{7,1}$};
\node[vertex] (y81) at (8,4) {$u_{8,1}$};
\node[vertex] (y91) at (9,4) {$u_{9,1}$};

\node[vertex] (y12) at (1,4.7) {};
\node[vertex] (y22) at (2,4.7) {};
\node[vertex] (y32) at (3,4.7) {};
\node[vertex] (y42) at (4,4.7) {};
\node[vertex] (y52) at (5,4.7) {};
\node[vertex] (y62) at (6,4.7) {};
\node[vertex] (y72) at (7,4.7) {};
\node[vertex] (y82) at (8,4.7) {};
\node[vertex] (y92) at (9,4.7) {};

\node[vertex] (y13) at (1,5.4) {$u_{1,3}$};
\node[vertex] (y23) at (2,5.4) {$u_{2,3}$};
\node[vertex] (y33) at (3,5.4) {$u_{3,3}$};
\node[vertex] (y43) at (4,5.4) {$u_{4,3}$};
\node[vertex] (y53) at (5,5.4) {$u_{5,3}$};
\node[vertex] (y63) at (6,5.4) {$u_{6,3}$};
\node[vertex] (y73) at (7,5.4) {$u_{7,3}$};
\node[vertex] (y83) at (8,5.4) {$u_{8,3}$};
\node[vertex] (y93) at (9,5.4) {$u_{9,3}$};

\path[-,thin]
(z0) edge (x1) edge (x2) edge (x3) edge (x4) edge (x5)
(z1) edge (z0) edge (z2);

\path[-,very thick]
(x1) edge (y11) edge (y21) edge (y31)
(x3) edge (y41) edge (y51) edge (y61)
(x5) edge (y71) edge (y81) edge (y91);

\path[-,very thin]
(x2) edge (y21) edge (y61) edge (y71)
(x4) edge (y31) edge (y51) edge (y71);

\path[-,thick]
(y11) edge (y12) (y12) edge (y13)
(y21) edge (y22) (y22) edge (y23)
(y31) edge (y32) (y32) edge (y33)
(y41) edge (y42) (y42) edge (y43)
(y51) edge (y52) (y52) edge (y53)
(y61) edge (y62) (y62) edge (y63)
(y71) edge (y72) (y72) edge (y73)
(y81) edge (y82) (y82) edge (y83)
(y91) edge (y92) (y92) edge (y93);

\node[]() at (11,0.0) {$s=2$};
\node[]() at (11,0.5) {$d=2$};
\node[]() at (11,3.0) {$p=5$};
\node[]() at (11,2.5) {$q=4$};
\node[]() at (11,2.0) {$K=5$};

\node at (15.35,0.11) {*};
\node at (14.7,-.25) {$\bullet$};
\node at (16,-1.3) {$\bullet$};
\node[zvertex](z0) at (15,0) {$z_0$};
\node[vertex] (z1) at (16.3, 0) {$z_1$};
\node[vertex] (z2) at (16.9, 0) {};
\node[vertex] (z3) at (17.4, 0) {};
\node[vertex] (z4) at (18.0, 0) {$z_4$};
\node[vertex] (zl1) at (16.3,-1) {$z'_1$};
\node[vertex] (zl2) at (16.9,-1) {};
\node[vertex] (zl3) at (17.4,-1) {};
\node[vertex] (zl4) at (17.9,-1) {};
\node[vertex] (zl5) at (18.4,-1) {};
\node[vertex] (zl6) at (19,-1) {$z'_6$};
\path[-,thin] (zl1) edge (z0) edge (zl2) (zl3) edge (zl2) edge (zl4) (zl5) edge (zl4) edge (zl6);

\node[cvertex] (x1) at (13,2) {$S_1$}; \node at (12.8,1.6) {$\bullet$};
\node[xvertex] (x2) at (14,2) {$S_2$};
\node[cvertex] (x3) at (15,2) {$S_3$}; \node at (14.8,1.6) {$\bullet$};
\node[xvertex] (x4) at (16,2) {$S_4$};
\node[cvertex] (x5) at (17,2) {$S_5$}; \node at (16.8,1.6) {$\bullet$};

\node[vertex] (y11) at (11,4) {$u_{1,1}$};
\node[vertex] (y21) at (12,4) {$u_{2,1}$};
\node[vertex] (y31) at (13,4) {$u_{3,1}$};
\node[vertex] (y41) at (14,4) {$u_{4,1}$};
\node[vertex] (y51) at (15,4) {$u_{5,1}$};
\node[vertex] (y61) at (16,4) {$u_{6,1}$};
\node[vertex] (y71) at (17,4) {$u_{7,1}$};
\node[vertex] (y81) at (18,4) {$u_{8,1}$};
\node[vertex] (y91) at (19,4) {$u_{9,1}$};

\node[vertex] (y12) at (11,4.8) {};
\node[vertex] (y22) at (12,4.8) {};
\node[vertex] (y32) at (13,4.8) {};
\node[vertex] (y42) at (14,4.8) {};
\node[vertex] (y52) at (15,4.8) {};
\node[vertex] (y62) at (16,4.8) {};
\node[vertex] (y72) at (17,4.8) {};
\node[vertex] (y82) at (18,4.8) {};
\node[vertex] (y92) at (19,4.8) {};

\node[vertex] (y13) at (11,5.4) {};
\node[vertex] (y23) at (12,5.4) {};
\node[vertex] (y33) at (13,5.4) {};
\node[vertex] (y43) at (14,5.4) {};
\node[vertex] (y53) at (15,5.4) {};
\node[vertex] (y63) at (16,5.4) {};
\node[vertex] (y73) at (17,5.4) {};
\node[vertex] (y83) at (18,5.4) {};
\node[vertex] (y93) at (19,5.4) {};

\node[vertex] (y14) at (11,6.0) {};
\node[vertex] (y24) at (12,6.0) {};
\node[vertex] (y34) at (13,6.0) {};
\node[vertex] (y44) at (14,6.0) {};
\node[vertex] (y54) at (15,6.0) {};
\node[vertex] (y64) at (16,6.0) {};
\node[vertex] (y74) at (17,6.0) {};
\node[vertex] (y84) at (18,6.0) {};
\node[vertex] (y94) at (19,6.0) {};

\node[vertex] (y15) at (11,6.8) {$u_{1,5}$};
\node[vertex] (y25) at (12,6.8) {$u_{2,5}$};
\node[vertex] (y35) at (13,6.8) {$u_{3,5}$};
\node[vertex] (y45) at (14,6.8) {$u_{4,5}$};
\node[vertex] (y55) at (15,6.8) {$u_{5,5}$};
\node[vertex] (y65) at (16,6.8) {$u_{6,5}$};
\node[vertex] (y75) at (17,6.8) {$u_{7,5}$};
\node[vertex] (y85) at (18,6.8) {$u_{8,5}$};
\node[vertex] (y95) at (19,6.8) {$u_{9,5}$};

\path[-,thin]
(z0) edge (x1) edge (x2) edge (x3) edge (x4) edge (x5)
(z1) edge (z0) edge (z2) (z3) edge (z2) edge (z4);

\path[-,very thick]
(x1) edge (y11) edge (y21) edge (y31)
(x3) edge (y41) edge (y51) edge (y61)
(x5) edge (y71) edge (y81) edge (y91);

\path[-,very thin]
(x2) edge (y21) edge (y61) edge (y71)
(x4) edge (y31) edge (y51) edge (y71);

\path[-,thick]
(y11) edge (y12) (y12) edge (y13) (y13) edge (y14) (y14) edge (y15)
(y21) edge (y22) (y22) edge (y23) (y23) edge (y24) (y24) edge (y25)
(y31) edge (y32) (y32) edge (y33) (y33) edge (y34) (y34) edge (y35)
(y41) edge (y42) (y42) edge (y43) (y43) edge (y44) (y44) edge (y45)
(y51) edge (y52) (y52) edge (y53) (y53) edge (y54) (y54) edge (y55)
(y61) edge (y62) (y62) edge (y63) (y63) edge (y64) (y64) edge (y65)
(y71) edge (y72) (y72) edge (y73) (y73) edge (y74) (y74) edge (y75)
(y81) edge (y82) (y82) edge (y83) (y83) edge (y84) (y84) edge (y85)
(y91) edge (y92) (y92) edge (y93) (y93) edge (y94) (y94) edge (y95);

\node[]() at (21,0.0) {$s=4$};
\node[]() at (21,0.5) {$d=4$};
\node[]() at (21,3.0) {$p=6$};
\node[]() at (21,2.5) {$q=5$};
\node[]() at (21,2.0) {$K=5$};

\node at (26.05,0.25) {*};
\node at (24.7,-.25) {$\bullet$};
\node at (25.5,-1.3) {$\bullet$};
\node[zvertex](z0) at (25,0) {$z_0$};
\node[vertex] (z1) at (25.8,0){$z_1$};
\node[vertex] (z2) at (26.4,0){};
\node[vertex] (z3) at (26.9,0){};
\node[vertex] (z4) at (27.4,0){};
\node[vertex] (z5) at (28.0,0){$z_5$};
\node[vertex] (zl1) at (25.8,-1){$z'_1$};
\node[vertex] (zl2) at (26.4,-1){};
\node[vertex] (zl3) at (26.9,-1){};
\node[vertex] (zl4) at (27.4,-1){};
\node[vertex] (zl5) at (27.9,-1){};
\node[vertex] (zl6) at (28.4,-1){};
\node[vertex] (zl7) at (29.0,-1){$z'_7$};
\path[-,thin] (zl1) edge (z0) edge (zl2) (zl3) edge (zl2) edge (zl4) (zl5) edge (zl4) edge (zl6) (zl7) edge (zl6);

\node[cvertex] (xx1) at (23,2) {$S_1$}; \node at (22.8,1.6) {$\bullet$};
\node[xvertex] (xx2) at (24,2) {$S_2$};
\node[cvertex] (xx3) at (25,2) {$S_3$}; \node at (24.8,1.6) {$\bullet$};
\node[xvertex] (xx4) at (26,2) {$S_4$};
\node[cvertex] (xx5) at (27,2) {$S_5$}; \node at (26.8,1.6) {$\bullet$};

\node[vertex] (yy11) at (21,4) {$u_{1,1}$};
\node[vertex] (yy21) at (22,4) {$u_{2,1}$};
\node[vertex] (yy31) at (23,4) {$u_{3,1}$};
\node[vertex] (yy41) at (24,4) {$u_{4,1}$};
\node[vertex] (yy51) at (25,4) {$u_{5,1}$};
\node[vertex] (yy61) at (26,4) {$u_{6,1}$};
\node[vertex] (yy71) at (27,4) {$u_{7,1}$};
\node[vertex] (yy81) at (28,4) {$u_{8,1}$};
\node[vertex] (yy91) at (29,4) {$u_{9,1}$};

\node[vertex] (yy12) at (21,4.8) {};
\node[vertex] (yy22) at (22,4.8) {};
\node[vertex] (yy32) at (23,4.8) {};
\node[vertex] (yy42) at (24,4.8) {};
\node[vertex] (yy52) at (25,4.8) {};
\node[vertex] (yy62) at (26,4.8) {};
\node[vertex] (yy72) at (27,4.8) {};
\node[vertex] (yy82) at (28,4.8) {};
\node[vertex] (yy92) at (29,4.8) {};

\node[vertex] (yy13) at (21,5.4) {};
\node[vertex] (yy23) at (22,5.4) {};
\node[vertex] (yy33) at (23,5.4) {};
\node[vertex] (yy43) at (24,5.4) {};
\node[vertex] (yy53) at (25,5.4) {};
\node[vertex] (yy63) at (26,5.4) {};
\node[vertex] (yy73) at (27,5.4) {};
\node[vertex] (yy83) at (28,5.4) {};
\node[vertex] (yy93) at (29,5.4) {};

\node[vertex] (yy14) at (21,6.0) {};
\node[vertex] (yy24) at (22,6.0) {};
\node[vertex] (yy34) at (23,6.0) {};
\node[vertex] (yy44) at (24,6.0) {};
\node[vertex] (yy54) at (25,6.0) {};
\node[vertex] (yy64) at (26,6.0) {};
\node[vertex] (yy74) at (27,6.0) {};
\node[vertex] (yy84) at (28,6.0) {};
\node[vertex] (yy94) at (29,6.0) {};

\node[vertex] (yy15) at (21,6.6) {};
\node[vertex] (yy25) at (22,6.6) {};
\node[vertex] (yy35) at (23,6.6) {};
\node[vertex] (yy45) at (24,6.6) {};
\node[vertex] (yy55) at (25,6.6) {};
\node[vertex] (yy65) at (26,6.6) {};
\node[vertex] (yy75) at (27,6.6) {};
\node[vertex] (yy85) at (28,6.6) {};
\node[vertex] (yy95) at (29,6.6) {};

\node[vertex] (yy17) at (21,7.4) {$u_{1,6}$};
\node[vertex] (yy27) at (22,7.4) {$u_{2,6}$};
\node[vertex] (yy37) at (23,7.4) {$u_{3,6}$};
\node[vertex] (yy47) at (24,7.4) {$u_{4,6}$};
\node[vertex] (yy57) at (25,7.4) {$u_{5,6}$};
\node[vertex] (yy67) at (26,7.4) {$u_{6,6}$};
\node[vertex] (yy77) at (27,7.4) {$u_{7,6}$};
\node[vertex] (yy87) at (28,7.4) {$u_{8,6}$};
\node[vertex] (yy97) at (29,7.4) {$u_{9,6}$};

\path[-,thin]
(z0) edge (xx1) edge (xx2) edge (xx3) edge (xx4) edge (xx5)
(z1) edge (z0) edge (z2) (z3) edge (z2) edge (z4) (z5) edge (z4);

\path[-,very thick]
(xx1) edge (yy11) edge (yy21) edge (yy31)
(xx3) edge (yy41) edge (yy51) edge (yy61)
(xx5) edge (yy71) edge (yy81) edge (yy91);

\path[-,very thin]
(xx2) edge (yy21) edge (yy61) edge (yy71)
(xx4) edge (yy31) edge (yy51) edge (yy71);

\path[-,thick]
(yy11) edge (yy12) (yy12) edge (yy13) (yy13) edge (yy14) (yy14) edge (yy15) (yy15) edge (yy16)
(yy21) edge (yy22) (yy22) edge (yy23) (yy23) edge (yy24) (yy24) edge (yy25) (yy25) edge (yy26)
(yy31) edge (yy32) (yy32) edge (yy33) (yy33) edge (yy34) (yy34) edge (yy35) (yy35) edge (yy36)
(yy41) edge (yy42) (yy42) edge (yy43) (yy43) edge (yy44) (yy44) edge (yy45) (yy45) edge (yy46)
(yy51) edge (yy52) (yy52) edge (yy53) (yy53) edge (yy54) (yy54) edge (yy55) (yy55) edge (yy56)
(yy61) edge (yy62) (yy62) edge (yy63) (yy63) edge (yy64) (yy64) edge (yy65) (yy65) edge (yy66)
(yy71) edge (yy72) (yy72) edge (yy73) (yy73) edge (yy74) (yy74) edge (yy75) (yy75) edge (yy76)
(yy81) edge (yy82) (yy82) edge (yy83) (yy83) edge (yy84) (yy84) edge (yy85) (yy85) edge (yy86)
(yy91) edge (yy92) (yy92) edge (yy93) (yy93) edge (yy94) (yy94) edge (yy95) (yy95) edge (yy96);
\end{tikzpicture}}

\caption{Case 3 ($s=2r+2$) when $s\leq d+1$ and consequently $r<d$. Reduction for Figure \ref{fig-reduction} Set Cover instance in Lemma~\ref{lem-npc3} proof.}\label{fig-reduction3b}
\end{figure}

%% file: spygame-bip.bbl
\begin{thebibliography}{10}
\providecommand{\url}[1]{\texttt{#1}}
\providecommand{\urlprefix}{URL }
\providecommand{\doi}[1]{https://doi.org/#1}

\bibitem{AF84}
Aigner, M., Fromme, M.: A game of cops and robbers. Discrete Applied
  Mathematics  \textbf{8},  1--12 (1984)

\bibitem{AlonM11}
Alon, N., Mehrabian, A.: On a generalization of {M}eyniel's conjecture on the
  cops and robbers game. The Electronic Jounal of Combinatorics  \textbf{18}(1)
  (2011)

\bibitem{AlonMS06}
Alon, N., Moshkovitz, D., Safra, S.: Algorithmic construction of sets for
  \emph{k}-restrictions. {ACM} Transactions on Algorithms  \textbf{2}(2),
  153--177 (2006)

\bibitem{BO98}
Babel, L., Olariu, S.: On the structure of graphs with few {P}4s. Discrete
  Applied Mathematics  \textbf{84},  1--13 (1998)

\bibitem{BalisterBBN16}
Balister, P., Bollobás, B., Narayanan, B., Shaw, A.: Catching a fast robber on
  the grid. Journal of Combinatorial Theory, Series A  \textbf{152},  341 --
  352 (2017)

\bibitem{Baumann96}
Baumann, S.: A linear algorithm for the homogeneous decomposition of graphs.
  Zentrum f\"ur Mathematik, Technische Universit\"at M\"unchen pp. Report No.
  M--9615 (1996)

\bibitem{BonatoCP10}
Bonato, A., Chiniforooshan, E., Pralat, P.: Cops and robbers from a distance.
  Theoretical Computer Science  \textbf{411}(43),  3834--3844 (2010)

\bibitem{BonatoN11}
Bonato, A., Nowakovski, R.: The game of Cops and Robber on Graphs. American
  Mathematical Society (2011)

\bibitem{q-yw}
Campêlo, M., Huiban, C., Sampaio, R., Wakabayashi, Y.: Hardness and
  inapproximability of convex recoloring problems. Theoretical Computer Science
   \textbf{533},  15--25 (2014)

\bibitem{q-sula}
Campos, V., Klein, S., Sampaio, R., Silva, A.: Fixed-parameter algorithms for
  the cocoloring problem. Discrete Applied Mathematics  \textbf{167},  52--60
  (2014)

\bibitem{q-claudia}
Campos, V., Linhares-Sales, C., Sampaio, R., Maia, A.K.: Maximization coloring
  problems on graphs with few {P}4s. Discrete Applied Math.  \textbf{164},
  539--546 (2014)

\bibitem{ChalopinCNV11}
Chalopin, J., Chepoi, V., Nisse, N., Vax{\`e}s, Y.: Cop and robber games when
  the robber can hide and ride. SIAM Journal on Discrete Math  \textbf{25}(1),
  333--359 (2011)

\bibitem{cohen16}
Cohen, N., Hilaire, M., Martins, N.A., Nisse, N., P{\'e}rennes, S.: {Spy-Game
  on Graphs}. In: 8th International Conference on Fun with Algorithms (FUN
  2016). Leibniz Intern. Proc. in Informatics (LIPIcs), vol.~49, pp.
  10:1--10:16 (2016)

\bibitem{cohen18}
Cohen, N., Martins, N.A., Mc~Inerney, F., Nisse, N., P{\'e}rennes, S., Sampaio,
  R.M.: Spy-game on graphs: Complexity and simple topologies. Theoretical
  Computer Science  \textbf{725},  1 -- 15 (2018)

\bibitem{cohen18alg}
Cohen, N., Mc~Inerney, F., Nisse, N., P{\'e}rennes, S.: Study of a
  combinatorial game in graphs through linear programming. Algorithmica
  \textbf{82},  212 -- 244 (2020)

\bibitem{FominGKNS10}
Fomin, F.V., Golovach, P.A., Kratochv\'{\i}l, J., Nisse, N., Suchan, K.:
  Pursuing a fast robber on a graph. Theoretical Computer Science
  \textbf{411}(7-9),  1167--1181 (2010)

\bibitem{FominGL10}
Fomin, F.V., Golovach, P.A., Lokshtanov, D.: Cops and robber game without
  recharging. In: 12th Scandinavian Symp. and Workshops on Algorithm Theory
  (SWAT). LNCS, vol.~6139, pp. 273--284. Springer (2010)

\bibitem{FominGP12}
Fomin, F., Golovach, P.A., Pralat, P.: Cops and robber with constraints. SIAM
  Journal on Discrete Mathematics  \textbf{26}(2),  571--590 (2012)

\bibitem{GHH05}
Goddard, W., Hedetniemi, S., Hedetniemi, S.: Eternal security in graphs.
  Journal of Combinatorial Mathematics and Combinatorial Computing  \textbf{52}
  (2005)

\bibitem{GoldwasserK08}
Goldwasser, J.L., Klostermeyer, W.: Tight bounds for eternal dominating sets in
  graphs. Discrete Mathematics  \textbf{308}(12),  2589--2593 (2008)

\bibitem{klavzar11}
Hammack, R., Imrich, W., Klav\u{z}ar, S.: Handbook of Product Graphs. CRC Press
  (2011)

\bibitem{Jamison92}
Jamison, B., Olariu, S.: A tree representation for {P}4-sparse graphs. Discrete
  Applied Mathematics  \textbf{35},  115--129 (1992)

\bibitem{KM09}
Klostermeyer, W., MacGillivray, G.: Eternal dominating sets in graphs. Journal
  of Combinatorial Mathematics and Combinatorial Computing  \textbf{68} (2009)

\bibitem{Klostermeyer2011}
Klostermeyer, W., Mynhardt, C.: Graphs with equal eternal vertex cover and
  eternal domination numbers. Discrete Mathematics  \textbf{311}(14),
  1371--1379 (2011)

\bibitem{q-nicolas}
Linhares-Sales, C., Maia, A.K., Martins, N., Sampaio, R.: Restricted coloring
  problems on graphs with few {P}4s. Annals of Operations Research
  \textbf{217},  385–397 (2014)

\bibitem{dana15}
Moshkovitz, D.: The projection games conjecture and the {NP}-hardness of ln
  $n$-approximating set-cover. Theory of Computing  \textbf{11}(7),  221--235
  (2015)

\bibitem{NowakowskiW83}
Nowakowski, R.J., Winkler, P.: Vertex-to-vertex pursuit in a graph. Discrete
  Mathematics  \textbf{43},  235--239 (1983)

\bibitem{zermelo13}
Zermelo, E.: {Über eine Anwendung der Mengenlehre auf die Theorie des
  Schachspiels}. In: Proc. Fifth International Congr. Math. pp. 501--504 (1913)

\end{thebibliography}
